\documentclass[sigconf]{acmart}
\AtBeginDocument{%
  }

\setcopyright{acmlicensed}
\copyrightyear{2025}
\acmYear{2025}
\acmDOI{XXXXXXX.XXXXXXX}
\acmConference[CCS '25]{Anonymous Submission to CCS.}{October 13--17,
  2025}{Taipei, Taiwan}
\acmISBN{978-1-4503-XXXX-X/2018/06}

\def\ccs{1}
\def\anonymous{1}

\newif\ifcompactitemandenum
\compactitemandenumtrue
\usepackage{paralist}
\usepackage{enumitem}

\usepackage{nicodemus}
\usepackage{tcolorbox}
\usepackage[capitalise]{cleveref}
\usepackage{macros}
\usepackage{macros_be}
\usepackage[lambda]{cryptocode}
\usepackage{pseudocodemacros}

\usepackage{pifont}
\usepackage{tikz-cd}
\usetikzlibrary{shapes.multipart}
\usepackage[procnames]{listings}
\usepackage{color}

\usepackage{booktabs}
\usepackage{graphicx}
\usepackage{array}
\usepackage{adjustbox}

\usepackage{pgfplots}
\usepackage{csvsimple}
\usepackage{subcaption}

\newtheorem{remark}{Remark}

\begin{document}

\title[Robust Distributed Arrays]{Robust Distributed Arrays: \\
 Provably Secure Networking for Data Availability Sampling}

\begin{abstract}
Data Availability Sampling (DAS), a central component of Ethereum's roadmap, enables clients to verify data availability without requiring any single client to download the entire dataset.
DAS operates by having clients randomly retrieve individual symbols of erasure-encoded data from a peer-to-peer network.
While the cryptographic and encoding aspects of DAS have recently undergone formal analysis, the peer-to-peer networking layer remains underexplored, with a lack of security definitions and efficient, provably secure constructions.

In this work, we address this gap by introducing a novel distributed data structure that can serve as the networking layer for DAS, which we call \emph{robust distributed arrays}.
That is, we rigorously define a robustness property of a distributed data structure in an open permissionless network, that mimics a collection of arrays.

Then, we give a simple and efficient construction and formally prove its robustness.
Notably, every individual node is required to store only small portions of the data, and accessing array positions incurs minimal latency.
The robustness of our construction relies solely on the presence of a minimal \emph{absolute} number of honest nodes in the network. In particular, we avoid any honest majority assumption.
Beyond DAS, we anticipate that robust distributed arrays can have wider applications in distributed systems.

\end{abstract}

\end{CCSXML}




\maketitle


\section{Introduction}

Blockchains are distributed systems that implement replicated state machines, whose states are regularly updated through distributed consensus protocols.
In a nutshell, these systems provide two main guarantees called \emph{liveness} and \emph{safety}.
The former ensures the continuous progression of the machine's state, while the latter guarantees that the machine never transitions into an invalid state.

\superparagraph{Strong Safety.}
Popular blockchains, such as Bitcoin and Ethereum, are designed to offer a strong form of safety.
They maintain the integrity of their state even when the majority of parties running the consensus protocol are acting maliciously.
To achieve such a strong form of safety, \emph{all} network participants need to individually verify every state transition proposed by the consensus protocol.
Invalid state transitions are rejected by the network participants, regardless of whether they have been approved by consensus.
While this approach provides strong safety guarantees, it also imposes heavy requirements on all network participants in terms of bandwidth, computing power, and storage, since every participant needs to download and verify every single state transition.

\superparagraph{Scaling Blockchains.}
As blockchains grow in popularity and adoption, the size of their state continues to increase and requiring full replication of the state machine among all network participants is becoming a performance bottleneck.
Towards building scalable blockchains, we would like to reduce the burden placed on the network participants, while maintaining the strong safety guarantees blockchains currently provide.
Ideally, one would like a solution which does not require network participants to download all state related data to verify the validity of each state transition.

At first glance, succinct non-interactive arguments of knowledge (SNARKs)~\cite{STOC:Kilian92,DBLP:journals/siamcomp/Micali00,EC:Groth16} appear to provide a solution to the above problem.
Whenever consensus proposes a new state transition, it can be accompanied by a succinct proof, attesting to the validity of the transition.
Network participants can verify the succinct proof, without needing to download all data associated with the state transition itself.
While this solution addresses the question of how to ensure the validity of state transitions, it does not address the question of how to ensure \emph{availability} of all state related data.
In addition to succinct proofs, the outlined approach requires a way for network participants to collectively check that all data is in principle available, preferably without requiring them to download the data in full.

\superparagraph{Data Availability Sampling.}
Towards addressing the above issue, Al-Bassam et al.~\cite{FC:ASBK21} have introduced the concept of \emph{data availability sampling} (DAS), which was subsequently formalized by Hall-Andersen, Simkin, and Wagner~\cite{cryptoeprint:2023/1079} and then studied by several other works~\cite{C:HalSimWag24,cryptoeprint:2024/1362,cryptoeprint:2025/034,evans2025accidental}.
In the DAS setting, a potentially malicious party encodes a data block and provides oracle access to the encoding.
A set of independent verifiers randomly probe the encoding through the provided oracle.
If the verifying parties independently accept the responses they receive from the adversarially controlled oracle, then the security properties of DAS ensure that the parties could pool their individual transcripts together to reconstruct some well-defined data.

The works on DAS mentioned above focus only on \emph{how to encode} the data, but do not study \emph{how to store it}.
In a real-world setting, the encoded data would not be stored inside a hypothetical oracle, but rather within a distributed system.
Network participants should be able to store new encodings or individual symbols that they have reconstructed, and query individual symbols of already existing encodings.
Understanding how to construct efficient and secure distributed storage networks is of crucial importance for making DAS a viable approach for scaling real-world blockchains.

\subsection{Our Contribution}

In this work, we formally define and construct secure distributed storage systems, which we call \emph{Robust Distributed Arrays (RDAs)}.
In combination with DAS encodings, RDAs enable a full-fledged approach for implementing DAS in the real world.

\superparagraph{Formal Model.}
Our formal model for RDAs precisely defines the security guarantees that should be provided in the presence of an adversary who controls some subset of parties in the network.
A particular focus is placed on avoiding idealizing or simplifying assumptions that are difficult to implement in reality.
Among other things, our formal model encapsulates both honest and malicious parties joining and leaving the network in an unpredictable, maliciously scheduled fashion and it does not assume that all honest parties can directly communicate with each other by default.

\superparagraph{Construction.}
In our model, we then construct concretely efficient RDAs and rigorously prove their security.
Our construction remains responsive and secure even if a majority of participants are malicious.
Concretely, we only assume that a sufficiently large \emph{absolute} number of honest parties is online in sufficiently large time intervals, irrespective of how many malicious parties there are in the network.
Our solution incurs minimal latency overheads for reading or writing the data stored within the system.

From an engineering perspective, our construction is conceptually simple, making it easier to implement securely.
It can be adapted to different trade-offs between the required data storage per party, the required bandwidth per party, and the required number of honest parties in the system.
We evaluate the practical performance of our system across different parameter settings and demonstrate that it achieves provable security, while maintaining high efficiency in many realistic scenarios.
As an example, we show that 5000 honest network participants, each storing 1\% of the data and being connected to 10\% of the other peers, allow for provably ensuring that 90\% of the data are available at all points in time.

\subsection{Related Work}
Several existing lines of research works are related to what we do here.
All of them fall short of providing a networking solution for DAS in one or multiple ways.
In the following, we will provide a short overview of the main related works and defer a more in-depth comparison to~\cref{sec:relwork}.

\superparagraph{Distributed Hash Tables.}
Conceptually, the systems we are aiming to build in this work are strongly reminiscent of distributed hash tables (DHTs).
These are data structures, maintained by a network of parties, that allow for storing key/value pairs and searching for values corresponding to given keys.
A large body of research works has studied how to construct asymptotically and concretely efficient DHTs, with the most notable constructions being~CAN~\cite{candht}, Pastry~\cite{Pastry}, Chord~\cite{Chord}, and Kademlia~\cite{Kademlia}.
Commonly, DHTs prioritize minimizing the number of peers that each network participant needs to be connected to, while still being able to store and retrieve data efficiently.
The above works show that, in a network with $n$ parties, each party only needs to maintain $\mathsf{polylog}(n)$ peer-to-peer connections to store and retrieve data with latency $\mathsf{polylog}(n)$.
All of the above works only consider benign settings without the presence of an adversarial power that may try to maliciously disturb the network.

\superparagraph{Robust Distributed Hash Tables.}
Motivated by the lack of security guarantees for DHTs in the presence of malicious behavior, several works have studied DHTs in the context of various adversarial influences~\cite{SODA:FiaSai02,DBLP:conf/iptps/NaorW03,DBLP:conf/icalp/AwerbuchS04,DBLP:conf/esa/FiatSY05,STOC:Scheideler05,DBLP:conf/spaa/AwerbuchS06,DBLP:conf/opodis/AwerbuchS06,DBLP:conf/iptps/AwerbuchS07,skademlia,DBLP:journals/sigops/SenF12,DBLP:journals/ton/YoungKGK13,DBLP:conf/podc/GuerraouiHK13}.
Many of these works rely on quorum-based approaches for providing security guarantees.
From a high-level perspective, the idea behind these approaches is to take a data structure that works in the benign setting and to then replace individual parties by sets of parties called quorums.
One then argues that the data structure works correctly in combination with the quorums, as long as each quorum contains an honest majority of parties.
As the benign DHTs, these approaches all incur polylogarithmic latencies per operation.

In contrast to the above works on DHTs, we consider a highly adversarial setting, where a large number of network participants may be corrupted.
Whereas the above approaches require (at least) an honest majority among all network participants to ensure honest majorities in all quorums, our RDAs remain secure and functional in the presence of a dishonest majority.
Furthermore, our construction only incurs small \emph{constant latencies} per operation, rather than polylogarithmic ones.
On the downside, our construction requires each network participant to maintain approximately $\sqrt{n}$ peer-to-peer connections.

\superparagraph{Networking for DAS.}
A few works have explored the networking in the specific context of DAS.
Kr{\'o}l et al.~\cite{DASp2p} informally identify functionality and efficiency requirements for networks underlying DAS and discuss why classical solutions, like DHTs, are not directly applicable.
Cortes-Goicoechea et al.~\cite{kademliaDASlimits} analyze the suitability of Kademlia~\cite{Kademlia} in the context of DAS in terms of efficiency, but not security.
Ascigil et al.~\cite{ethresearPANDASPractical} propose a networking solution for DAS, which has similarities with our approach.
Unfortunately their approach requires each network participant to have almost complete knowledge of all other parties.
In addition, their work lacks both formal security definitions and proofs.

In our work, we provide security notions for the networking layer and prove that our solution provides the stated security guarantees.
Our approach does not require network participants to have complete or close to complete knowledge of the network.

\section{Notation and Model}
\label{sec:prelims}
\label{sec:prelims:model}
Here, we specify notation and our model. We recall relevant mathematical tools in \cref{sec:appendix:mathtools}.

\ifnum\ccs=0
\superparagraph{Notation.}
    We denote the set of the first $N$ positive integers by $[N] = \{1,\dots,N\}$.
    From time to time, we use standard asymptotic notation like $\Omega,\BigO,\Theta,\omega$, where the variable over which it is defined shall be clear from the context.
    We usually use typewriter font (e.g., $\texttt{msg}$ instead of $\text{msg}$) to distinguish explicit strings from variables. For example, we may say that a party sets $x \coloneqq 5$ and sends $(\msgStore, x)$ to another party, which means that $\msgStore$ is a string and the string $(\msgStore,5)$ is actually sent.

\superparagraph{Pseudocode.}
When writing our protocols, we give both a verbal description as well as pseudocode. The pseudocode specifies both instructions that parties execute when interfaces of the protocol are called by applications, as well as message handlers that are executed when receiving certain messages.
We use a standard imperative syntax for our pseudocode, including constructs like $\pcfor$, $\pcif$, $\pcparse$. We also assume instructions to send messages ($\textbf{send} \dots \textbf{to} \dots$), to wait for an amount of time ($\textbf{wait for} \dots$), and to receive messages within an amount of time ($\textbf{on receiving} \dots \textbf{within } \dots$). Unless specified otherwise, all maps and sets are implicitly initialized as being empty and all numerical variables are initialized as zero.
\fi

\superparagraph{Random Oracle Model.}
We will analyze our construction in the random oracle model~\cite{CCS:BelRog93}. To recall, all parties in this model have oracle access to a function $\Hash\colon\XXX \rightarrow \YYY$, which is uniformly random.
\ifnum\ccs=0
That is, we assume that initially, for all $x \in \XXX$, we have $\Hash(x) = \bot$.
This holds until the first party queries $\Hash(x)$. At this point, a random $y \in \YYY$ is sampled uniformly and we define $\Hash(x) \coloneqq y$, before returning $\Hash(x)$ to the party.
Any subsequent call $\Hash(x)$ then returns $y$. In particular, all outputs are independently sampled from $\YYY$.
In other words, the random oracle is a uniform random function from $\XXX$ to $\YYY$. Any randomness in our experiments is also taken over the random coins used to determine the values of $\Hash$, i.e., over the random oracle.
\fi

\superparagraph{Network Model.}
Throughout, we use the terms \emph{node} and \emph{party} interchangeably.
We consider a synchronous network, i.e., any message is delivered within a finite and known delay, and all parties have synchronized clocks.
This allows protocols and our security experiment to proceed in rounds, which we also refer to as points in time. Precisely, if a party sends a message in round $\tau$, then it is received at some point in round $\tau+1$.
Strictly speaking, the message could already arrive in round $\tau$, but we assume for convenience that it arrives exactly in round $\tau+1$.
This is without loss of generality, as the receiving party can simply wait until $\tau+1$ before processing the message.
The network we consider is open (aka permissionless), i.e., parties can join and leave over time.
Parties cannot automatically communicate with each other.
Instead, a party $\party$ can only send a message to $\party'$ if it knows their address, which it could have learned from some third party.
As soon as $\party$ knows the address of $\party'$, we assume that $\party$ can initiate communication and the two parties can communicate via an (not necessarily private) authenticated point-to-point channel. A party always knows its own address.
Of course, if a party does not already know any other party when joining, it could never reach any other party.
Therefore, we will assume that when a party joins, it gets as input a list of \emph{bootstrap nodes}.
We discuss below how this list is specified.
In our concrete protocol, only a subset of the parties plays the role of bootstrap nodes and nodes are aware of their own role.



\superparagraph{Protocols.}
A protocol $\Pi$ is specified by an ensemble of interfaces and corresponding instructions.
Applications running on a node can call such an interface with some inputs, which results in the node executing the instructions specified by the protocol, e.g., sending messages to other nodes.
A protocol must also specify instructions for joining, and instructions to react to incoming messages.
We call the parties that are currently online and take part in a protocol \emph{active} for $\Pi$.
We assume that $\Pi$ also specifies initialization instructions $\Pi.\iInit$ to be executed in a coordinated setup\footnote{%
For this work, we only need protocols where the coordinated setup completes within a single round at time $\tau = 0$. We use this to simplify the definitions.}
by a set $\mathbf{P}$ of initially active honest parties that know each other.
This allows protocols to build some initial topology among these parties.

\superparagraph{Adversarial Model.}
We consider an adversary that is in full control of a subset of all parties, and we call these parties \emph{malicious}, \emph{adversarial}, or \emph{byzantine}.
The set of malicious parties is assumed to be fixed and disjoint from the honest parties, but not necessarily bounded.
That is, we consider some form of \emph{static} corruptions in which honest parties remain honest.
For security, we will tolerate arbitrarily many malicious parties, i.e., we do not assume any protection against Sybil attacks.
On the other hand, efficiency will depend on the total number of parties, so that adding Sybil protection on top is a good idea.
Malicious parties can deviate arbitrarily from the protocol and collude at all times, and our adversary is rushing.
This means that it sees all messages of honest parties in a certain round before deciding on its own actions.
We will assume that the adversary is \emph{computationally bounded}, and treat cryptographic primitives as idealized objects, as common in distributed systems literature.
When it comes to honest parties, we let the adversary adaptively schedule their interface calls and inputs.

\superparagraph{Join and Leave Schedules.}
While joining and leaving of malicious parties is under full control of the adversary, we assume that honest parties join and leave according to an arbitrary schedule, which is not controlled by the adversary.
More precisely, a \emph{join-leave schedule} for a protocol $\Pi$ is given by a function\footnote{We do not allow the join-leave schedule to depend on the random oracle, i.e., the random oracle is instantiated after the join-leave schedule (as a function) is fixed.} $\JoinLeaveSchedule$ that outputs the following data $\JoinLeaveSchedule(\tau)$ for each point in time $\tau \geq 0$:
\begin{myitemize}
 \item A list of honest parties $\mathbf{P}$ that are supposed to join the protocol in round $\tau$, with special considerations at time $\tau = 0$, see below.
 \item For each such party $\party \in \mathbf{P}$, a set $\{\party_1,\ldots,\party_t\}$ of \emph{bootstrap} nodes (depending on $\party$) via which $\party$ should join.
 \item For each such party $\party \in \mathbf{P}$, some auxiliary data $\AuxRole$ (depending on $\party$).
 \item A list of honest parties $\widetilde{\mathbf{P}}$ that are supposed to leave in round $\tau$.
\end{myitemize}
Without loss of generality, we only consider schedules in which only parties are instructed to leave that are currently active\footnote{%
Due to how our experiment works, ''active`` at time $\tau$ means that the join-leave schedule specified it as a party that should join at some time $\tau' \leq \tau$, but did not specify it as a party that should leave at any time $\tau''\leq \tau$.
Note that we consider parties active in the time slots where they join and leave.
}. Otherwise, if an honest party did not yet join (or already left again) but should leave according to the schedule, this can simply be ignored.
The list $\mathbf{P}$ of honest parties output at time $\tau = 0$ are the parties that run the coordinated setup. Each of those parties runs $\Pi.\iInit(\mathbf{P},\AuxRole_1,\ldots,\AuxRole_{\ell})$, where the $\AuxRole_i$ are their respective auxiliary data as above. See \cref{def:experiment} for a more precise description.
For time $\tau > 0$, if the join-leave schedule specifies that $\party$ should join via $\{\party_1,\ldots,\party_t\}$ with auxiliary data $\AuxRole$, the party $\party$ instead calls
$\Pi.\iJoin(\{\party_1,\ldots,\party_t, \party'_1,\ldots, \party'_{t'}\},\AuxRole)$, where $\party'_1,\ldots,\party'_{t'}$ are additional (potentially malicious) bootstrap nodes adaptively chosen by the adversary.
Honest parties never join a second time after they have left.
Equivalently, we assume that honest nodes get a fresh identifier and forget all previous state if they were to re-join.
Also, we require that any honest node $\party_i$ that is designated as a bootstrap node for the joining of $\party$ by $\JoinLeaveSchedule(\tau)$ at time $\tau$ is actually active at time $\tau - 1$, i.e.\ $\party_i$ joined itself \emph{before} $\party$ and did not yet leave.
Note that $\iJoin$ might take multiple rounds to complete\ifnum\ccs=0\footnote{%
There are multiple ways to formally model algorithms spanning multiple rounds (e.g.\ via a ``sleep functionality'' or sending messages to yourself and stop, then resume upon receiving that message).
For this work, we assume that each interface call and processing each received message each starts a new thread on the affected party. We then have a ``sleep and wake up the next round (but in the same thread)'' operation available.
This is solely to make it easier to precisely talk about terminating calls or returning values from calls: interface calls terminate (possibly with a returned value) iff the thread that was started by that call finishes.
}\fi~and other interface calls are not meaningful until then; we model this by having our security definitions provide no guarantees for interfaces that are called before $\iJoin$ has completed.

Our security definitions will universally quantify over join-leave schedules $\JoinLeaveSchedule$.
Importantly, we remark that this means joining and leaving of honest parties is \emph{independent} of the protocol execution and random coins sampled during the protocol.
In particular, parties that decide to join or leave based on the status of the protocol but otherwise behave honestly are treated as malicious.
Further, we allow protocols to make assumptions on the schedule: our security notions require us to define a class $\AdmissibleSchedules$ of \emph{admissible} schedules and we then only quantify over all admissible schedules $\JoinLeaveSchedule \in \AdmissibleSchedules$.
One example may be the class of all schedules in which honest parties always stay active for at least time $\Delta$, or the class of all schedules such that there is always at least one active honest node.
The admissibility condition might also put requirements on the number of parties involved in the coordinated setup.







With these assumptions, we have now implicitly specified a security experiment, and our security definitions for protocols will make use of this experiment.
\ifnum\ccs=0
We make the experiment explicit in the following definition.

\begin{definition}[Security Experiment]\label{def:experiment}
    The experiment proceeds in rounds, starting at $\tau = 0$.
    For a given adversary $\adversary$ and a join-leave schedule $\JoinLeaveSchedule$, it works as follows:
    \begin{myenumerate}
        \item \emph{Initialization.}
            At $\tau = 0$, we consider the list of honest parties $\mathbf{P} = (\party_1,\ldots,\party_{\ell})$ determined by $\JoinLeaveSchedule(0)$.
            For each such party $\party_i\in\mathbf{P}$, we have some auxiliary data $\AuxRole_i$ that depends on $i$.
            Note that the set of bootstrap nodes is necessarily empty (because those would need to have joined strictly before time $0$).
            Each of those initially active parties $\party_i$ runs $\Pi.\iInit(\mathbf{P},\AuxRole_1,\ldots,\AuxRole_{\ell})$.
            We require this to not span multiple rounds.
            Next, the adversary $\adversary$ is started. It has oracle access to $\JoinLeaveSchedule$. We skip Steps~\ref{step:joins}--\ref{step:events} for $\tau = 0$.
        \item\label{step:joins} \emph{Honest Parties Join.}
            In each round, honest parties join according to $\JoinLeaveSchedule(\tau)$. For each joining honest party $\party$, call $\Pi.\iJoin( \{ \party_1,\ldots,\party_{t},\allowbreak\party'_1,\ldots,\party'_{t'}\},\AuxRole)$,
            where $\party_1,\ldots,\party_t$ and $\AuxRole$ are specified by $\JoinLeaveSchedule(\tau)$ for $\party$ and $\party'_1,\ldots,\party'_{t'}$ are specified by the adversary as in Step~\ref{step:determine_extra_bootstrap_nodes}.
            Also, continue any $\iJoin$ - calls that were started in past rounds, but did not finish yet.
        \item\label{step:events} \emph{Honest Party Actions.}
            In each round, a list of events (specified in Step~\ref{step:interfacecalls}) is executed as follows
            \begin{enumerate}
                \item The event specifies an honest party $\party$, an interface of the protocol, and a set of inputs.
                \item The event is triggered by calling the interface with the respective inputs on party $\party$, provided $\party$ is active.
            \end{enumerate}
            Furthermore, active honest parties process incoming messages and perform actions based on them.
            Also, active honest parties continue processing actions from previous rounds if the protocol specifies that those actions span multiple rounds;

            typically, the latter happens due to waiting for some other processing to complete or waiting for incoming messages.
            We stress that, unless specified otherwise, this means that waiting does \emph{not} block parties from processing new event or handling other messages.
        \item \label{step:leaving}\emph{Leaving.}
            At the end of each round $\tau$, honest parties leave according to $\JoinLeaveSchedule(\tau)$.
        \item \emph{Adversarial Actions.}
            In each round, $\adversary$ can take the following actions, after observing the messages that honest parties send:
            \begin{enumerate}
                \item Create malicious parties.
                \item Determine the behavior (e.g., messages) of all malicious parties in this round;
                \item\label{step:interfacecalls} Determine the events for the next round (cf.\ Step~\ref{step:events}); this means
                \begin{enumerate}
                    \item Determine which interface calls (excluding $\iInit$ and $\iJoin$) are executed on honest parties and their inputs.
                    \item Send messages originating from any malicious party to any honest party (thereby adding message processing events).
                    \item Determine the order in which interface calls and messages are processed by active honest parties (including continuations from previous rounds, if processing spans multiple rounds).
                    \item\label{step:determine_extra_bootstrap_nodes}Determine the additional bootstrap nodes used by the parties that join in the next round. The adversary can query $\JoinLeaveSchedule$, so it knows which honest parties will join in the future.
                \end{enumerate}
            \end{enumerate}
        Recall that we assume authenticated channels and we assume a fixed (but unbounded) set of malicious parties, so creating a malicious party really means creating communication channels to honest parties.
        Also, since the adversary is byzantine, there is no need for its parties to join the protocol via $\Pi.\iJoin$; it just sends appropriate messages.
    \end{myenumerate}
    We assume that the experiment does not terminate, but we only require that our security properties hold for a certain lifetime $\ProtocolLifetime \in \NN$.
    \end{definition}
\else
We make the experiment explicit in \cref{sec:appendix:ccs:experiment}, \cref{def:experiment}.
\fi
\begin{remark}[Setup]\label{rmk:weaker_setup}
In our model, we assume that the initial parties are all honest. This is mostly for simplicity because we want to model that our protocol starts with a defined state.
The assumption can be weakened to only require that the set $\HonestInitialParties$ of initial honest parties know each other.
In this case, every honest initial party $\party\in\HonestInitialParties$ is initialized with $\prot.\iInit(\HonestInitialParties\union\AdversarialInitialParties{\party},\AuxRole_1,\ldots)$ where $\AdversarialInitialParties{\party}$ are adversarially chosen parties.
We remark that the protocol we define in \cref{sec:construction} is also secure in this more general model, provided we instantiate it with an appropriate subnet protocol, such as the concrete subnet protocol in \cref{sec:appendix:simplesubnetprotocol}.
\end{remark}

Our robustness notions defined later (e.g., \cref{def:rda:robustness}) and our security proofs will heavily make use of certain events that may happen during a run of this experiment.
For ease of reference, we provide a table on Page \pageref{table:eventsoverview} of all such events defined throughout the document.
\ifnum\ccs=0
Let us start by defining some events related to activity of parties that we will need very often, and state a trivial proposition that follows from the definition and our assumptions that parties never re-join.
\else
Let us start by defining some events related to activity of parties that we will need very often.
\fi
\begin{definition}[Activity Events]
    \label{def:EventActive}
Let $0\leq \tau \leq \tau'$ be time slots and let $\party$ be an honest party.
Consider a run of our experiment from \cref{def:experiment} for some protocol and some join-leave schedule $\JoinLeaveSchedule$.
Then, we define the following events.
\begin{myitemize}
    \item Event $\EventActive{\party}{\tau}$: This event occurs if at time $\tau$, honest party $\party$ is active.%
            \ifnum\ccs=0\footnote{%
            In our experiment that we introduced in \cref{def:experiment}, an honest party $\party$ always joins at the \emph{beginning} of a round $\tau_1$ (or is one of the initial parties, in which case $\tau_1 = 0$) and leaves at the \emph{end} of a round $\tau_2$ (or never leaves, in which case $\tau_2 = \infty$).
            This means that $\party$ is active precisely from time $\tau_1$ until $\tau_2$.
            Since $\party$ will not join/leave in the middle of a round, we are guaranteed that interfaces can be called and all messages are processed by $\party$ in any round during which $\party$ is active.
            }\fi
    \item Event $\EventActiveDuration{\party}{\tau}{\tau'}$: This event occurs if honest party $\party$ is active (at least) from time $\tau$ to $\tau'$ (both inclusive).

    \item Event $\EventFullyJoinedDuration{\party}{\tau}{\tau'}$: This event occurs if the honest party $\party$ is active (at least) from $\tau$ to $\tau'$ and either is one of the initial parties or its $\iJoin$-call has finished by time slot $\tau$.%
    \ifnum\ccs=0\footnote{%
            Recall that in our experiment \cref{def:experiment}, the $\iJoin$-call will finish at the \emph{beginning} of some round $\tau_0$, even if that $\iJoin$-call takes multiple rounds. For that round $\tau_0$ itself, we consider $\party$ already fully joined, so $\EventFullyJoined{\party}{\tau_0}$ holds.
            }\fi
    \item Event $\EventFullyJoined{\party}{\tau}$: This is short for $\EventFullyJoinedDuration{\party}{\tau}{\tau}$.
\end{myitemize}
\end{definition}
\ifnum\ccs=0
\begin{proposition}[Activity is an Interval]
    Let $\party$ be an honest party and $0\leq \tau_0 \leq \tau_1$ be points in time.
    Then, for any point in time $\tau$ with $\tau_0 \leq \tau \leq \tau_1$, we have
    \begin{align*}
        \EventActive{\party}{\tau_0} &\land \EventActive{\party}{\tau_1} \implies \EventActive{\party}{\tau}, \text{ and}\\
        \EventFullyJoined{\party}{\tau_0}& \land \EventActive{\party}{\tau_1} \implies \EventFullyJoined{\party}{\tau}\enspace.
    \end{align*}
    Consequently, we have
    \begin{align*}
        \EventActive{\party}{\tau_0} &\land \EventActive{\party}{\tau_1} \iff \EventActiveDuration{\party}{\tau_0}{\tau_1}, \text{ and}\\
        \EventFullyJoined{\party}{\tau_0} &\land \EventActive{\party}{\tau_1} \iff \EventFullyJoinedDuration{\party}{\tau_0}{\tau_1}\enspace.
    \end{align*}
\end{proposition}
\noindent
We will use these properties without explicitly referring to this proposition.
\begin{proof}
This follows from the definitions and the fact that parties never re-join.
\end{proof}
\fi

\section{Robust Distributed Arrays}
\label{sec:rda}
Now, we present the definition of robust distributed arrays motivated by data availability sampling.
\subsection{Background: Data Availability Sampling}
To motivate our definition of robust distributed arrays, we recall the basic setting of data availability sampling (DAS).
Robust distributed arrays are meant to implement the networking layer of DAS.

\superparagraph{Data Availability Sampling.}
\label{sec:background:das}
DAS~\cite{FC:ASBK21,cryptoeprint:2023/1079} allows clients to confirm whether data (such as a transaction list) is available on a peer-to-peer network.
We can think of this process as follows: \begin{myenumerate}
    \item A proposer holding the data redundantly encodes the data and computes a succinct commitment of this encoding.
    \item The proposer distributes this encoding on a peer-to-peer network, such that it is stored across various nodes.
    \item Clients download the concise commitment.
    \item To verify the availability of the associated data on the network, clients randomly request specific portions of the encoded data from the network and check consistency with the commitment. If enough of these random checks are successful, the clients can infer that the data is indeed available.
\end{myenumerate}
Hall-Andersen, Simkin, and Wagner~\cite{cryptoeprint:2023/1079} have defined a formal cryptographic framework for DAS and have analyzed several constructions.
Importantly, their model and analysis concerns the encoding and verification part of DAS, but \emph{not how the network is implemented}.
That is, in their model they treat the network as a black-box by providing an oracle to the clients:
\begin{myitemize}
    \item \emph{Correctness Case.} When modeling the requirement that all clients accept if the encoding is indeed available and provided by an honest proposer, clients are assumed to have oracle access to the encoding. That is, a client can input $i$ to its oracle and obtains the $i$th symbol of the encoding.
    \item \emph{Adversarial Case.} When defining security properties, the oracle is instead under full control of the adversary. In particular, the adversary can decide for each query what to return, and it can make its decisions adaptive, even taking into account which client makes the query.
\end{myitemize}
The natural next question after the work of Hall-Andersen, Simkin, and Wagner is how the network would be implemented in reality, and what kind of properties it should achieve.

\superparagraph{Basic Network Interface.}
From the description above it is clear that any implementation of the network should provide two services: \begin{myitemize}
    \item \emph{Storing The Encoding.} The proposer has to have a way to distribute the encoding somehow in the network.
    \item \emph{Making Queries.} The clients need to be able to specify indices $i$ and get back the $i$th symbol of the encoding from the network.
\end{myitemize}
We will make this interface more precise when giving our formal definition, but one may already notice that the (non-distributed) data-structure that provides this kind of interface is simply an \emph{array}.

\superparagraph{What Can't Go Wrong.}
As Hall-Andersen, Simkin, and Wagner show, constructions of DAS remain \emph{secure} even in case the adversary is in full control of the oracle, i.e., in full control of the network. Therefore, \emph{no matter how} the network is implemented, it is guaranteed that (enough) clients only accept if data is indeed available.
One may be tempted to think that this means we are done, and that we could just use any network implementation, but as we see next, this thought is misleading.

\superparagraph{What Can Go Wrong.}
It turns out that the model of Hall-Andersen, Simkin, and Wagner alone does not give us \emph{liveness}.
More specifically, consider the case where we have an honest proposer that encodes the data and makes it available on the network. But consider the case that a certain portion of peers on the network acts maliciously. For example, these peers may delete the symbols of the encoding that they receive, forward incorrect symbols, or try to prevent honest peers from receiving correct data.
Then, a bad implementation of the network would make clients think that data is not available (e.g., because they do not get responses for their queries), even though the proposer made it available initially.
This is what our robustness definition should rule out.

Another relevant scenario that arises is that some encoded symbols are initially unavailable but later recovered by honest parties through the properties of the underlying erasure code.
In this case, these honest parties should be able to reintroduce the recovered symbols to the network, ensuring that other honest peers can access them. This mechanism is crucial for eventually reaching agreement on data availability.

\subsection{Defining Robust Distributed Arrays}
We now turn to a formal definition of robust distributed arrays, using the network and adversarial model that we have introduced in \cref{sec:prelims:model}.
First, we define the syntax of such protocols.
As already outlined above, we require two interfaces.
Namely, to allow parties to distribute encodings, we require an interface $\iStore$, which takes an index $i$ and a symbol $x$, and intuitively should set the $i$th position of the array to $x$.
Parties can read from the array by specifying the position $i$ in a call to an interface $\iGet$, which intuitively should return the $i$th symbol of the array.
To allow for a collection of arrays, we also add a \emph{handle} $h$ as input to the interfaces, with the intuition that different $h$'s access different arrays.
In the setting of DAS, we think of $h$ as being the concise commitment that all clients download.
As distributed arrays could have other applications, we refer to the encoding as a \emph{file}.
\begin{definition}[Distributed Array]
    A distributed array protocol with file space $\SymbolSpace^\FileLength$ and handle space $\HandleSpace$ is a protocol $\prot$ offering the interfaces $\iStore$, $\iGet$, with the following syntax: \begin{myitemize}
        \item Interface $\iStore(h,i,x)$: This interface can be called by any party $\party$. It takes as input a handle $h \in \HandleSpace$, an index $i \in [\FileLength]$, and a symbol $x \in \SymbolSpace$.
        It does not return anything to $\party$.
        \item Interface $\iGet(h,i)$: This interface can be called by any party $\party$. It takes as input a handle $h \in \HandleSpace$ and an index $i \in [\FileLength]$. It returns a symbol $x \in \SymbolSpace$ to $\party$.
        \item Interfaces $\iJoin$ and $\iInit$ as required by \cref{sec:prelims:model}.
    \end{myitemize}
\end{definition}
We have yet to provide a \emph{semantic} definition of how $\iStore$ and $\iGet$ should behave.
To establish a meaningful definition, first consider the simplified scenario in which all parties behave honestly and what would constitute the correctness property of a (non-distributed) array: if a symbol $x$ is stored at position $i$ and this position is later accessed, the value $x$ should be returned, assuming that $x$ has not been overwritten in the interim.
To simplify our lives, we will disregard this overwriting assumption and instead focus on position-symbol pairs satisfying a predicate $\rdspred(h, i, x)$, where $\rdspred$ models the \emph{verification} of responses in DAS. Crucially, we assume that $\rdspred$ is position-binding (see \cref{remark:positionbinding}).

Our goal is to define a robustness property that mirrors the correctness property of traditional, non-distributed arrays, but within a distributed and adversarial context. Specifically, we want to ensure that if an honest party calls $\iStore(h, i, x)$ and another honest party subsequently calls $\iGet(h, i)$, the latter retrieves $x$. Naturally, this requirement must be adapted to accommodate the distributed and adversarial nature of our setting.

First, the fact that we are in a distributed setting necessitates several changes.
For instance, a symbol stored via $\iStore$ might not be immediately accessible through $\iGet$, as time is required for the symbol to propagate through the system. To address this, we introduce two delay parameters: \begin{myitemize}
    \item $\GetDelay$: the delay between invoking $\iGet$ and receiving a response.
    \item $\StoreGetDelay$: the delay required after an $\iStore$ call before the stored symbol can be reliably accessed.
\end{myitemize}
Also, we may need to make some assumptions on how (many) parties join and when, so we parameterize the definition by a set of admissible schedules $\AdmissibleSchedules$.

Second, in the presence of adversaries, we cannot expect correctness to hold universally.
Instead, we allow for some tolerance of adversarial influence:
\begin{myitemize}
    \item For a small $\RelPosCorrThresh$-fraction of indices $\Corrupted_{\party}^\tau \subseteq [\FileLength]$, correctness may fail. These corrupted indices may depend on the specific parties $\party$ calling the interface and the current time $\tau$.
    \item All of this is then required to hold with high probability $1-\delta$ over the protocol's lifetime, $\ProtocolLifetime$.
\end{myitemize}
We now present the formal robustness definition.
To summarize, it states that for most indices $i$, if a party $\PartyStoring$ stored the symbol $x$ at position $i$ at time $\tauStore$, and after some minimum delay $\StoreGetDelay$ a party $\PartyGetting$ calls $\iGet(h,i)$, then this call returns $x$ after delay $\GetDelay$.
\begin{definition}[Robustness]
    Let $\prot$ be a distributed array protocol with file space $\SymbolSpace^\FileLength$, handle space $\HandleSpace$, and predicate $\rdspred\colon \HandleSpace \times [\FileLength] \times \SymbolSpace \to \bool$.
    Consider the experiment specified by our network and adversarial model for an adversary $\adversary$ and a join-leave schedule $\JoinLeaveSchedule$.
    In this experiment, consider the following events:
    \begin{myitemize}
        \item Event $\EventStored{\PartyStoring}{\tauStore}{h}{i}{x}$:
            This event occurs if at time $\tauStore$, honest party $\PartyStoring$ calls $\iStore(h,i,x)$.
        \item Event $\EventCalledGet{\PartyGetting}{\tauGet}{h}{i}$:
            This event occurs if at time $\tauGet$, honest party $\PartyGetting$ calls $\iGet(h,i)$.
        \item Event $\EventGotResult{\PartyGetting}{\tauGet}{h}{i}{\Delta}{x}$:
            \ifnum\ccs=0
            This event occurs if $\EventCalledGet{\PartyGetting}{\tauGet}{h}{i}$ occurs and at some point $\tau \in [\tauGet,\tauGet+\Delta]$, the call to $\iGet$ returns $x$.
            \else
            This event occurs if we have both $\EventCalledGet{\PartyGetting}{\tauGet}{h}{i}$ and at some point $\tau \in [\tauGet,\tauGet+\Delta]$, the call to $\iGet$ returns $x$.
            \fi
    \end{myitemize}
    Then, we say that $\prot$ is $(\rdspred,\RelPosCorrThresh,\RobustnessError,\ProtocolLifetime,\StoreGetDelay,\GetDelay,\AdmissibleSchedules)$-robust, if for every adversary $\adversary$ and $\JoinLeaveSchedule \in \AdmissibleSchedules$, the following holds except with probability at most $\RobustnessError$:

    \smallskip\noindent There is a family of subsets $\Corrupted_{\party}^\tau \subseteq [\FileLength]$, indexed by honest parties and time, such that: \begin{itemize}
        \item \emph{Bounded Size.}
            For every honest party $\party$ and any point in time $0 \leq \tau \leq \ProtocolLifetime$, we have $|\Corrupted_\party^\tau| \leq \RelPosCorrThresh \FileLength$.
        \item \emph{Correctness if not Corrupted.}
            For all honest parties $\PartyStoring,\PartyGetting$,
            all points in time $0 \leq \tauStore,\tauGet \leq \ProtocolLifetime$ with $\tauGet \geq \tauStore + \StoreGetDelay$,
            all handles $h \in \HandleSpace$, symbols $x \in \SymbolSpace$, and indices $i \in [\FileLength] \setminus (\Corrupted_{\PartyStoring}^{\tauStore} \union \Corrupted_{\PartyGetting}^{\tauGet})$, we have
        \ifnum\ccs=0
        \begin{align*}
            \left(\begin{array}{rl}
                        &\EventFullyJoined{\PartyStoring}{\tauStore}\\
                \land   &\EventStored{\PartyStoring}{\tauStore}{h}{i}{x} \\
                \land   &\rdspred(h,i,x) \\
                \land   &\EventFullyJoinedDuration{\PartyGetting}{\tauGet}{(\tauGet+\GetDelay)}\\
                \land   &\EventCalledGet{\PartyGetting}{\tauGet}{h}{i}
            \end{array} \right)
             \implies \EventGotResult{\PartyGetting}{\tauGet}{h}{i}{\GetDelay}{x}.
        \end{align*}
        \else
        \begin{align*}
            &\left(\begin{array}{rl}
                        &\EventFullyJoined{\PartyStoring}{\tauStore}\\
                \land   &\EventStored{\PartyStoring}{\tauStore}{h}{i}{x} \\
                \land   &\rdspred(h,i,x) \\
                \land   &\EventFullyJoinedDuration{\PartyGetting}{\tauGet}{(\tauGet+\GetDelay)}\\
                \land   &\EventCalledGet{\PartyGetting}{\tauGet}{h}{i}
            \end{array} \right)
             \\
             \implies &\EventGotResult{\PartyGetting}{\tauGet}{h}{i}{\GetDelay}{x}.
        \end{align*}
        \fi
    \end{itemize}
    \label{def:rda:robustness}
\end{definition}
\begin{remark}[Position-Binding]
    In the context of DAS, the predicate will always be \emph{position-binding} (see, e.g., \cite{cryptoeprint:2023/1079}). That is, no efficient adversary can output a handle $h \in \HandleSpace$, a position $i \in [\FileLength]$, and two symbols $x,x' \in \SymbolSpace$ with $x \neq x'$ and $\rdspred(\handle,i,x) = \rdspred(\handle,i,x') = 1$.
    \label{remark:positionbinding}
\end{remark}
\begin{remark}[No Secrecy]
    We do not require any secrecy from the distributed array protocol. The adversary is allowed to immediately learn any symbol that is stored, and any query that is made.
\end{remark}
\begin{remark}[Agreement on the Handle]
    The robust distributed array protocol does not specify how parties agree on the handle $h$ (the commitment in DAS). One can assume that another protocol (e.g., a broadcast channel, realized via the blockchain) takes care of that.
\end{remark}

\superparagraph{Efficiency Metrics.}
We aim to minimize the following metrics:
\begin{myitemize}
    \item \emph{Latency and Round Complexity.} What is the latency of the $\iStore$ and $\iGet$ interfaces? Specifically, what are the delays $\StoreGetDelay$ and $\GetDelay$, and how many communication rounds are required when these interfaces are invoked?
    \item \emph{Communication Complexity.} How much data is exchanged among honest parties during protocol execution?
    \item \emph{Bandwidth per Node.} What is the data transmission and reception rate for an honest party over a fixed time period?
    \item \emph{Number of Connections.} How many peer connections must a party maintain? While memory constraints are less of a concern, maintaining a high number of connections increases bandwidth demands.
    \item \emph{Data per Node.} What portion of the data must each node store?
\end{myitemize}

\section{Our Construction}
\label{sec:construction}
In this section, we present our construction and analyze it.
\ifnum\ccs=0
We start with an informal overview.
Then, we formally define a necessary building block. Finally, we specify our protocol and analyze it.
\fi

\subsection{Overview}
In this section, we explain the intuition behind our protocol, assuming a fixed handle $h \in \HandleSpace$ for simplicity.
We start with a simple protocol and then introduce our protocol as a generalization that improves on it in terms of bandwidth requirements. This first solution is (almost) identical to the PANDAS proposal~\cite{ethresearPANDASPractical} (see \cref{sec:relwork:das}).

\superparagraph{Starting Point: A Clique Network.}
Robust distributed hash tables typically rely on honest-majority assumptions, and each query generates a large polylogarithmic number of messages and incurring polylogarithmic latency.
Our first goal is to overcome these limitations.
To achieve this, we split the file from $\SymbolSpace^\FileLength$ to be stored into $k_2$ chunks\footnote{It will be clear later why we denote the number of chunks by $k_2$.} of equal size, i.e., each chunk containing $\FileLength / k_2$ symbols.
When a party joins, it is assigned responsibility to store a random chunk.
We make the following assumptions: (1) when an honest party $\party$ joins, all other honest parties learn about $\party$ within a bounded time; (2) one can publicly determine for which chunk a given party $\party$ is responsible, e.g., by assigning the chunks via a random oracle.
With these assumptions, a simple protocol implementing distributed arrays with constant latency exists:
to store or get a symbol at index $i$, one first identifies the chunk $j \in [k_2]$ in which this symbol resides.
Then, one identifies all parties responsible for chunk $j$ and contacts them (in parallel).
In the case of get, one takes the response satisfying the predicate\footnote{Recall that we assume that the predicate is position-binding, cf.\ \cref{remark:positionbinding}.}.
By our two assumptions, it must be that all honest parties that are responsible for chunk $j$ are contacted.
Therefore, as long as \emph{at least one} honest party is responsible for the chunk, the query succeeds.
For any fixed chunk, the probability that no honest party is responsible is at most $(1-1/k_2)^{N}$, where $N$ denotes (a lower bound on) the number of honest parties. If we assume $N \geq \Omega(\secpar k_2)$, we get a correct response with overwhelming probability in $\secpar$. Notably, the number of malicious parties does not matter at all, and so we have eliminated the need for an honest-majority assumption. Robustness is determined solely by the absolute number of honest parties, while the presence of adversarial parties only affects efficiency.

The downside of this protocol is that every party needs to know the identities of \emph{all} other parties. Maintaining this global party list requires significant bandwidth. Consequently, our next goal is to remove this requirement while preserving the protocol's efficiency and robustness.

\superparagraph{Splitting Parties into Rows.}
Our first idea towards a solution is to have $k_1$ \emph{rows}, each row copying the approach from before.
For now, assume these rows are fully independent networks and that each party would be part of one row.
In particular, we can think of each party as being in one cell of a $k_1\times k_2$ matrix, where each column is responsible for storing one chunk, and nodes within rows form a clique.
If a party in row $i$ wants to access a symbol in a chunk $j \in [k_2]$, it would now ask all parties in cell $(i,j)$ for that symbol.
If we can assign parties to rows in a balanced way, then each of these cliques will contain only an $1/k_1$ fraction of all parties, and so we reduced the bandwidth requirements.
Assuming for now that this is exact, then fixing a row and a chunk as before, the probability that no honest party is responsible for that chunk in that row is at most $(1-1/k_2)^{N/k_1}$. This means that our requirement on $N$ slightly increases to $N \geq \Omega(\secpar k_1 k_2)$.
Of course, in an open distributed protocol, there is no omniscient coordinator that could ensure that every row contains exactly $N/k_1$ honest parties.
To address this, we adopt a randomized approach: parties are assigned to \emph{random} rows. Within their rows, they are then assigned to random cells as before.
Intuitively, with a sufficiently large number of parties, each row is likely to contain a comparable number of honest parties.

\superparagraph{How to Store Data.}
So far, it is not clear how parties in different rows actually get the data in the first place.
Specifically, assume a party $\party$ in row $i$ wants to store a symbol that resides in chunk $j \in [k_2]$.
Then, it can of course send it to every party in cell $(i,j)$, but because rows are completely independent, there is no way for that data to end up in cell $(i',j)$ for some $i' \neq i$.
This means parties in other rows cannot access this data later.
To solve this issue, we introduce \emph{column subnets}, i.e., we make it so that columns also form cliques.
Then, every party in cell $(i,j)$ that receives the data from $\party$ would use this column subnet to further distribute the data to every other node in the column, thereby propagating it into other rows.
We visualize the process of storing and retrieving data in \cref{fig:tikz_grid_store_get}.
To define the full protocol, we need to specify more details, e.g., how to sync data during joining, and what kind of assumptions we make on these column and row subnets.

\subsection{Subnet Discovery}
Before presenting the precise description of our protocol, we introduce an abstract subprotocol that we will use for implementing the row and column subnets.
As outlined in our overview, we aim to maintain the intuitive invariant that in each row and column all nodes are fully connected.
To model this, we introduce the definition of a \emph{subnet discovery protocol}.
This protocol provides two key functionalities: (a) a node can join a subnet through interaction with an existing member, and (b) a node can look up the list of all members of the subnet.
\begin{definition}[Subnet Discovery Protocol]
    A subnet discovery protocol (for at most $\numberofsubnets$ subnets) is a protocol $\prot$ offering the interfaces $\iCreateSubnet$, $\iJoinSubnet$, and $\iGetPeers$ with the following syntax: \begin{itemize}
        \item Interface $\iCreateSubnet(\subnetid, \mathbf{P})$: This interface can be called by any party $\party$.
        It takes a subnet identifier $\subnetid \in [\numberofsubnets]$ and a set $\mathbf{P}$ of parties.
        It does not return any output.
        \item Interface $\iJoinSubnet(\subnetid, \party')$:
        This interface can be called by any party $\party$. It takes as input a subnet identifier $\subnetid \in [\numberofsubnets]$ and a party $\party'$. It does not return anything to $\party$.
        \item Interface $\iGetPeers(\subnetid)$:
        This interface can be called by any party $\party$. It takes as input a subnet identifier $\subnetid \in [\numberofsubnets]$. It returns a set $\mathbf{P}$ of parties to $\party$.
        \item Interfaces $\iInit(\mathbf{P},\perp,\ldots,\perp)$ and $\iJoin(\emptyset,\perp)$ as required by \cref{sec:prelims:model}.
            We make the simplifying assumption that $\iJoin$ uses no bootstrap nodes (hence the $\emptyset$) and we use no auxiliary data $\AuxRole$ (hence the $\perp$'s).\ifnum\ccs=0 See \cref{rmk:ConcurrentSecurity}.\fi
    \end{itemize}

    \label{def:subnetdiscovery}
\end{definition}
\ifnum\ccs=0
\begin{remark}[Multiple subnets]
\label{rmk:ConcurrentSecurity}
A natural way to realize a subnet discovery protocol for $\numberofsubnets$ subnets is to run $\numberofsubnets$ many concurrent instances (parameterized by $\subnetid$) of a simplified subnet protocol for just one subnet each.
However, our definition given here is more general and allows us to state the required robustness for a \emph{concurrent} execution.
Still, it is helpful to think of protocols constructed that way and we will for simplicity only consider subnet discovery protocols $\prot$ for which both
$\prot.\iInit$ and $\prot.\iJoin$ complete within a single round, $\prot.\iJoin$ does not use bootstrap nodes and we do not have extra input $\AuxRole$ for $\prot.\iInit$ and $\prot.\iJoin$.
\end{remark}
\else
\vspace{-3ex plus 3ex} 
\fi%
We now turn to defining robustness for subnet protocols.
Essentially, our definition says that if two parties both join a subnet, then they learn about each other eventually.
As for the robustness definition of robust distributed arrays, we allow a failure probability $\RobustnessError$ and consider a protocol lifetime $\ProtocolLifetime$ and a set of admissible join-leave schedules $\AdmissibleSchedules$.
The definition also uses a delay parameter $\SubnetDelay$ that models the time needed to join the subnets.
\begin{definition}[Robustness]
    \label{def:subnetdiscovery:robustness}
    Let $\prot$ be a subnet discovery protocol.
    Consider the experiment specified by our network and adversarial model for an adversary $\adversary$ and a join-leave schedule $\JoinLeaveSchedule$. In this experiment, consider the following events:
    \begin{itemize}
        \item Event $\EventCreatedSubnet{\subnetid}{\mathbf{P}}$: This event occurs, if at time $\tau = 0$, all parties $\party \in \mathbf{P}$ are honest and call the interface $\iCreateSubnet(\subnetid, \mathbf{P})$. Furthermore, no other calls to $\iCreateSubnet$ with this $\subnetid$ are made by honest parties at any point in time, and no calls to $\iJoinSubnet$ with this $\subnetid$ are made by honest parties at time $\tau = 0$.
        \item Event $\EventStaysInSubnetDelay{\subnetid}{\party}{\tau_0}{\tau_1}$: This event occurs for a honest party $\party$ and time slots $0\leq \tau_0 \leq \tau_1$ if $\party$ is active (at least) from time slots $\tau_0$ to $\tau_1$ (both inclusive) and either one of the following holds:\begin{itemize}
            \item $\EventCreatedSubnet{\subnetid}{\mathbf{P}}$ such that $\party \in \mathbf{P}$. \quad -- or --
            \item At some time $\tau' \leq \tau_0 - \SubnetDelay$, $\party$ calls $\iJoinSubnet(\subnetid, \party')$ for an honest party $\party'$ with $\EventStaysInSubnetDelay{\subnetid}{\party'}{\tau'}{(\tau'+\SubnetDelay)}$.
            \end{itemize}
        \item Event $\EventIsInSubnetDelay{\subnetid}{\party}{\tau}$: This event is a shorthand for $\EventStaysInSubnetDelay{\subnetid}{\party}{\tau}{\tau}$.
        \item Event $\EventCalledGetPeers{\subnetid}{\party}{\tau}$: This event occurs, if at time $\tau$, honest party $\party$ calls $\iGetPeers(\subnetid)$.
        \item Event $\EventGotPeer{\subnetid}{\party}{\party'}{\tau}$: This event occurs, if the event $\EventCalledGetPeers{\subnetid}{\party}{\tau}$ occurs and the call $\iGetPeers(\subnetid)$ returns $\mathbf{P}$ at time $\tau$ such that $\party' \in \mathbf{P}$.
    \end{itemize}

    With these events at hand, we say that $\prot$ is $(\RobustnessError,\ProtocolLifetime,\SubnetDelay,\AdmissibleSchedules)$-robust, if for every adversary $\adversary$ and $\JoinLeaveSchedule \in \AdmissibleSchedules$, the following holds except with probability at most $\RobustnessError$:
    For every point in time $\tau \in  [\ProtocolLifetime]$, every subnet identifier $\subnetid \in [\numberofsubnets]$, and every pair of honest parties $\party,\party'$, we have:
    \begin{align*}
        \left(
            \begin{array}{rl}
                        &\EventIsInSubnetDelay{\subnetid}{\party}{\tau}\\
                \land   &\EventIsInSubnetDelay{\subnetid}{\party'}{\tau}\\
                \land   &\EventCalledGetPeers{\subnetid}{\party}{\tau}
            \end{array}
        \right) \implies \EventGotPeer{\subnetid}{\party}{\party'}{\tau}\enspace.
    \end{align*}
\end{definition}
To summarize, the definition states that if a subnet has been created properly at the onset of the protocol ($\tau = 0$), and an honest party $\party$ is in that subnet, then $\iGetPeers$ returns all other honest parties in that subnet.
We show how to implement a basic subnet discovery protocol that satisfies this notion in \cref{sec:appendix:simplesubnetprotocol}.
\ifnum\ccs=0
\begin{proposition}[One Stays in Subnets]
    Let $\prot$ be a subnet discovery protocol for $\numberofsubnets$ subnets.
    Let $\party$ be an honest party, let $\tau \geq 0$ be a point in time, and let $\subnetid \in [\numberofsubnets]$. Let $\SubnetDelay \geq 0$ be any parameter.
    Then, we have: 
    \[
        \left(\EventIsInSubnetDelay{\subnetid}{\party}{\tau} \land \EventActive{\party}{\tau+1}\right) \implies \EventIsInSubnetDelay{\subnetid}{\party}{\tau+1}\enspace.
    \]
\end{proposition}
\fi

\subsection{Construction}
Let $k_1,k_2 \in \NN$ and assume that we have a subnet discovery protocol $\subnetprot$ for $\numberofsubnets = k_1 + k_2$ subnets.
Let $\SubnetDelay\geq 2$ be a delay parameter for the subnet discovery protocol and let $\SyncDelay \geq 2$ be some synchronization delay parameter modeling the time needed to synchronize the stored data (see \cref{rmk:SyncDelay}).
Given that, we now specify a distributed array protocol $\ourprot$ with arbitrary file space $\SymbolSpace^\FileLength$ and arbitrary handle space $\HandleSpace$.
For simplicity, we will assume that $k_2$ divides $m$.
We will make use of a hash function $\Hash\colon\boolstar\rightarrow\bool$ to assign parties to cells, which is modeled as a random oracle.
In the following, we describe our protocol verbally.
We give the full specification of $\ourprot$ as pseudocode in \cref{fig:ourprotocol:initialization,fig:ourprotocol:joining,fig:ourprotocol:storeget}.

\superparagraph{Overall Structure.}
The protocol will be presented by talking about $k_1$ rows and $k_2$ columns.
For a file $f \in \SymbolSpace^\FileLength$, the first $\FileLength/k_2$ symbols $f_1,\dots,f_{\FileLength/k_2}$ reside in column $c = 1$, the second $\FileLength/k_2$ symbols $f_{\FileLength/k_2+1},\dots,f_{2\FileLength/k_2}$ reside in column $c = 2$, and so on.
That is, we say that the $i$th symbol $f_i$ (for $i \in [\FileLength]$) \emph{resides in column} $c$, where $(c-1)\FileLength/k_2 < i \leq c\FileLength/k_2$.
Each party $\party$ will \emph{reside in a cell} $(r,c) = \AlgCell(\party) \in [k_1] \times [k_2]$, i.e., in row $r$ and column $c$.
This cell is determined via the random oracle, i.e., $\AlgCell(\party) = \Hash(\party)$.
With this, each other party can determine whether $\party$ resides in cell $(r,c)$. 
Our protocol utilizes the subnet discovery protocol $\subnetprot$ to maintain a subnet for each row and for each column. We use functions $\AlgGetRowSid$, $\AlgGetColSid$ to obtain a unique identifier for each subnet from the row resp.\ column index.
Whenever $r$ resp. $c$ denotes the index of a row resp.\ column, we will use the shorthand notation $\subnetid_r\coloneqq \AlgGetRowSid(r)$ resp.\ $\subnetid_c\coloneqq \AlgGetColSid(c)$ for the identifier.
It will always be clear from context whether such an index refers to a row or to a column.
Only a subset of nodes is utilized as bootstrap nodes. Nodes know whether they are bootstrap nodes via their initial input $\AuxRole\in\{0,1\}$, with $\AuxRole = 1$ indicating a bootstrap node.
These nodes need to perform more work and will be in every row subnet.
Other than that, they will act exactly as normal nodes, e.g., they have their single dedicated cell assigned via $\AlgCell$.


\superparagraph{Setup and Initial Topology.}
Let $\mathbf{P}=[\party_1,\ldots,\party_{\ell}]$ be the list of initial honest nodes specified by the join-leave-schedule $\JoinLeaveSchedule$.
For each such honest node $\party_i\in\mathbf{P}$, we have a corresponding bit $\AuxRole_i\in\{0,1\}$.
Our experiment will call $\ourprot.\iInit(\mathbf{P},\AuxRole_1,\ldots\AuxRole_{\ell})$ on each $\party\in\mathbf{P}$, which performs the following actions:
\begin{myitemize}
    \item Call $\subnetprot.\iInit(\mathbf{P},\perp,\ldots,\perp)$ to run the coordinated setup of the subnet discovery protocol.
    Also, inherit any message handlers of $\subnetprot$.
    By this, we mean that $\ourprot$ should henceforth process all messages intended for $\subnetprot$ by forwarding them to $\subnetprot$.
    \item Let $\RowParties_r \subseteq \mathbf{P}$ be the subset of all parties in row $r$, i.e., the set of all $\party_i$'s with $\AlgRow(\party_i) = r$ or $\AuxRole_i = 1$.
    \item Let $\ColParties_c \subseteq \mathbf{P}$ be the subset of all parties in column $c$, i.e., the set of all $\party_i$'s with $\AlgCol(\party_i) = c$.
    \item \ifnum\ccs=0
            Create each row subnet $r \in [k_1]$ by calling $\subnetprot.\iCreateSubnet(\subnetid_r,\allowbreak \RowParties_r)$, provided $\self \in \RowParties_r$.
          \else
            Create each row subnet for $r \in [k_1]$ by calling the interface $\subnetprot.\iCreateSubnet(\subnetid_r,\allowbreak \RowParties_r)$, provided $\self \in \RowParties_r$.
          \fi
          The latter is true either for only $r=\AlgRow(\self)$ or for all $r\in [k_1]$, depending on $\AuxRole$.
    \item Party $\self$ creates its corresponding column subnet by calling $\subnetprot.\iCreateSubnet(\subnetid_c, \ColParties_c)$ for the column $c=\AlgCol(\self)$.
    \item Party $\self$ initializes an empty map $\SymbolsMap[\cdot,\cdot]$.
    \item Party $\self$ stores $\IsBootstrapNode \coloneqq \AuxRole_i$, where $\self = \party_i$.
\end{myitemize}
Note that in our robustness proof, we will need to ensure that the $\RowParties_r$ and $\ColParties_c$ are non-empty and we require at least one initial bootstrap node if we ever want another node to join.

\superparagraph{Store.}
Suppose interface $\iStore(h,i,x)$ is invoked on party $\self$.
Then, the following steps are taken, but only if $\rdspred(h,i,x) = 1$: \begin{myenumerate}
    \item Let $c \coloneq \AlgGetColForSymbol(i) \in [k_2]$ denote the column in which symbol $i$ resides, and let $r\in[k_1]$ be the row in which $\self$ resides.
    \item Party $\self$ identifies the set $\mathbf{P}$ of all parties in cell $(r,c)$ it is aware of:
        \begin{myenumerate}
            \item Get all parties $\mathbf{P}'$ in row $r$ via $\subnetprot.\iGetPeers(\subnetid_r)$.
            \item Consider only the parties $\mathbf{P} \subseteq \mathbf{P}'$ in column $c$, which can be done via $\Hash$.
        \end{myenumerate}
    \item To each party $\party \in \mathbf{P}$, party $\self$ sends a message $(\msgStore,h,i,x)$.
    \item Each party $\party'$ that receives $(\msgStore,h,i,x)$ stores the symbol as $\SymbolsMap[h,i] \coloneqq x$, provided $\rdspred(h,i,x) = 1$. 
    \item Each such party $\party'$ then identifies the set of all parties in column $c$ and sends $(\msgStoreFwd,h,i,x)$ to them.
    \item \ifnum\ccs=0
            Each party that receives $(\msgStoreFwd,h,i,x)$ stores $\SymbolsMap[h,i] \coloneqq x$, provided $\rdspred(h,i,x) = 1$.
          \else
            Each party that receives $(\msgStoreFwd,h,i,x)$ stores\\ $\SymbolsMap[h,i] \coloneqq x$, provided $\rdspred(h,i,x) = 1$.
          \fi
\end{myenumerate}


\superparagraph{Get.}
Suppose interface $\iGet(h,i)$ is invoked on party $\self$ at time $\tau$.
The following steps are taken:
\begin{myenumerate}
    \item Let $c \in [k_2]$ denote the column in which symbol $i$ resides, and let $r\in[k_1]$ be the row in which $\self$ resides.
    \item Party $\self$ identifies the set $\mathbf{P}$ of all parties in cell $(r,c)$ it is aware of, as in $\iStore$.
    \item To each party $\party \in \mathbf{P}$, party $\self$ sends a message $(\msgGet,h,i)$.
    \item \ifnum\ccs=0
            Each party that receives $(\msgGet,h,i)$ from $\self$ checks if $\SymbolsMap[h,i] = \bot$.
            If not, it responds with a message $(\msgGetRsp,h,i,x)$ for $x \coloneqq \SymbolsMap[h,i]$.
          \else
            Each party that receives $(\msgGet,h,i)$ from $\self$ checks whether $\SymbolsMap[h,i] = \bot$ holds.
            If not, it responds with a message $(\msgGetRsp,h,i,x)$ for $x \coloneqq \SymbolsMap[h,i]$.
          \fi
    \item When $\self$ receives the first response $(\msgGetRsp,h,i,x)$ such that $\rdspred(h,i,x) = 1$, it outputs $x$.
    \item If $\self$ does not receive such a response by time slot $\tau+2$, it outputs $\bot$ in time slot $\tau+3$.
\end{myenumerate}

\superparagraph{Join.}
Suppose a party $\self$ joins the protocol at time $\tau$ via parties $\{\party_1,\dots,\party_t\}$ acting as bootstrap nodes. Let $\AuxRole\in\{0,1\}$ be the provided bit designating whether $\self$ may be a bootstrap node itself for parties that join in the future.
The following steps are taken, where $(r,c) \coloneqq \AlgCell(\self)$:
\begin{myenumerate}
    \item Call $\subnetprot.\iJoin(\emptyset,\perp)$ and inherit $\subnetprot$'s message handlers.
    \item Party $\self$ initializes an empty map $\SymbolsMap[\cdot,\cdot]$ and memorizes $\IsBootstrapNode \coloneqq \AuxRole$.
    \item Party $\self$ starts to join (in parallel) the appropriate row subnets via the bootstrap nodes. In detail:
    \begin{myenumerate}
     \item If $\AuxRole = 1$, call $\subnetprot.\iJoinSubnet(\subnetid_{r'},\party_i)$ for each row $r'\in[k_1]$ and each bootstrap node $\party_i$.
     \item Otherwise, call $\subnetprot.\iJoinSubnet(\subnetid_r,\party_i)$ for each bootstrap node $\party_i$.
    \end{myenumerate}
    \item Set $\ColParties_c\coloneqq\emptyset$. To join its column subnet, party $\self$ first sends $\msgJoin$ to each of its bootstrap nodes $\party_i$ in parallel. Each bootstrap node reacts as follows:
    \begin{myenumerate}
        \item \ifnum\ccs=0
                Lookup all nodes via the row subnets, i.e., call $\subnetprot.\iGetPeers(\subnetid_r)$ for each row $r \in [k_1]$ and let $\mathbf{P}$ be the union of the results.
              \else
                Lookup all nodes via row subnets: call $\subnetprot.\iGetPeers(\subnetid_r)$ for each row $r \in [k_1]$ and let $\mathbf{P}$ be the union of the results.
              \fi
        \item Determine $c \coloneqq \AlgCol(\party)$, where $\party$ is the sender of the received $\msgJoin$ message.
        \item Let $\ColParties_c\subseteq \mathbf{P}$ be the set of parties $\party'\in\mathbf{P}$ in column $c$.
        \item Send $(\msgJoinRsp,\ColParties_c)$ to $\party$.
    \end{myenumerate}
    \item \ifnum\ccs=0
            Party $\self$ waits, expecting back the message $(\msgJoinRsp,\mathbf{P}^{(i)}_c)$ from each $\party_i$ within two time slots.
          \else
            Party $\self$ waits, expecting the message $(\msgJoinRsp,\smash{\mathbf{P}^{(i)}_c})$ from each $\party_i$ within two time slots.
          \fi
    \item Upon receiving each such message within two time slots:
    \begin{myenumerate}
        \item Call $\subnetprot.\iJoinSubnet(\subnetid_c, \party')$ for every $\party'\in\mathbf{P}^{(i)}_c\setminus\ColParties_c$ at time $\tau + 2$.
        \item $\ColParties_c \coloneqq \ColParties_c \union \mathbf{P}^{(i)}_c$.
    \end{myenumerate}

    \item Wait for $\SubnetDelay$ additional time until $\tau + 2 + \SubnetDelay$.
    \item At time slot $\tau + 2 + \SubnetDelay$,  $\self$ starts to synchronize data as follows:
    \begin{myenumerate}
        \item $\party$ updates $\ColParties_c =\subnetprot.\iGetPeers(\subnetid_c)$.
        \item $\self$ sends $\msgSync$ to every $\party' \in \ColParties_c$. 
        \item $\self$ terminates\footnote{%
        The significance of formally terminating $\iJoin$ here before we actually synched the data is that $\EventFullyJoined{\self}{\tau + 2 +\SubnetDelay}$ holds.
        This means that starting from round $\tau+2+\SubnetDelay$, we may already call $\iStore$ and $\iGet$ and enjoy the guarantees of the robustness definition, even though synchronization of the data has not yet finished.}
        $\iJoin$ now in time slot $\tau + 2 + \SubnetDelay$. Note that $\self$ will still handle $\msgSyncRsp$-answers as described below, even though $\iJoin$ has terminated.
    \end{myenumerate}
    \item When a party $\party'$ receives $\msgSync$, it does the following
    \begin{myenumerate}
    \item $\party'$ sets
        \begin{equation*}
            \mathbf{S} \coloneqq \{(h,i,\smash{\SymbolsMap^{\party'}}[h,i]) \mid \smash{\SymbolsMap^{\party'}}[h,i] \neq \bot\}\enspace,
        \end{equation*}
        where $\smash{\SymbolsMap^{\party'}}$ refers to the $\SymbolsMap$ map stored by $\smash{\party'}$.
    \item $\party'$ responds with $(\msgSyncRsp,\mathbf{S})$ \emph{after some arbitrary delay} $T \leq \SyncDelay-2$. See \cref{rmk:SyncDelay} for an explanation of the delay. Note that $\mathbf{S}$ is determined at the time of receiving $\msgSync$.
    \end{myenumerate}
    \item For each message $(\msgSyncRsp,\mathbf{S})$ that $\self$ obtains, it stores $\SymbolsMap[h,i] = x$ for each $(h,i,x) \in \mathbf{S}$ with $\rdspred(h,i,x) = 1$.
\end{myenumerate}

To summarize, each normal node only joins its row and column subnets, and each bootstrap node ($\AuxRole_i = 1$) additionally joins \emph{all} row subnets.
After time $\SubnetDelay$, the new node will receive relevant $\msgStore$ and $\msgStoreFwd$ messages and learn about new data; we can call $\iStore$ and $\iGet$ from this point on.
To learn data that was stored prior to this point in our column, we send a $\msgSync$ message to other nodes in our column to synchronize that data.

\begin{remark}[Synchronization Delay]\label{rmk:SyncDelay}
As part of $\iJoin$, a newly joining node $\self$ asks other nodes $\party'$ in the same column for any data that was stored in the past by sending a $\msgSync$-query.
In practice, the amount of data that $\party'$ needs to send back to $\self$ here might be very large.
In our network model, we have assumed that \emph{any} message arrives by the end of the next round, which is unrealistic in the case of large data.
To accommodate for that fact in a realistic way, we have artificially introduced a delay $\SyncDelay$ here.
We will provide more discussion on $\SyncDelay$ in \cref{sec:extensions}.
\end{remark}

%

\subsection{Robustness}

We now show that our construction $\ourprot$ satisfies the robustness notion defined in \cref{def:rda:robustness}.
%
\def \REMARKABOUTACTIVITY{%
First, we observe that $\ourprot$ uses the subnet protocol $\subnetprot$ is a way that precisely simulates a run of the experiment specified in \cref{def:experiment}.
Furthermore, in this simulated experiment, honest parties are active in $\subnetprot$ iff they are active in $\ourprot$. Parties are initial parties in $\subnetprot$ if and only if they are initial parties in $\ourprot$. 
The event $\EventFullyJoinedNoArgs$, which refers to $\iJoin$ having terminated, always means that $\ourprot.\iJoin$ has terminated.
This allows us to meaningfully use the events from \cref{def:EventActive} and \cref{def:subnetdiscovery:robustness} and we will talk about active parties without the need to distinguish between $\ourprot$ and $\subnetprot$.
}
\ifnum\ccs=0
\REMARKABOUTACTIVITY
\fi
Before stating the precise theorem, we give an intuition for the conditions we require for our theorem to hold:
Roughly speaking, we wish that the underlying subnet protocol is robust
and there are honest parties in each row, column and cell.

\superparagraph{Overlap times.}
A closer look shows that just requiring there to be at least one active honest party assigned to each row, column and cell is not quite enough:
we require there to be sufficient overlap time $\OverlapTime$ between potentially leaving old honest parties and newly arriving parties.
We model this overlap by asking that active honest parties $\party$ exist that remain active for at least $\OverlapTime$ more slots.
This way, if a new party joins, its lifetime overlaps at least $\OverlapTime$ time slots with $\party$.
The overlap time is needed both for joining the subnets and for actually handing over the data.
\ifnum\ccs=0
The old party needs to have been around for a sufficient time in order to have properly joined the subnet itself and finished retrieving data from its peers.
Then the old party needs to remain online for roughly the same amount of time to help new parties join their subnets and retrieve data from the old party.
For that reason, the required overlap is roughly twice that of $\SubnetDelay$ resp.\ $\SyncDelay$.
\fi

\superparagraph{Failure recovery.}
Our conditions for rows, columns and cells differ: assume that in one of our \emph{subnets}, there is no honest party with sufficient overlap: any future party asked to join that subnet will be ``tainted'' by joining solely via adversarial or honest ``tainted'' parties.
Consequently, this subnet will not provide any guarantees anymore and our protocol cannot recover from this. By contrast, having no honest party in a given bad \emph{cell} is a \emph{transient} failure case that our protocol can recover from.
Such failures will give rise to sets $\Corrupted_{\party}^\tau$ as afforded by \cref{def:rda:robustness} and we merely need to rule out that too many cells are simultaneously bad at any given time.
\ifnum\ccs=1
See \cref{def:our_corruption_sets} for our concrete definition of $\Corrupted_{\party}^\tau$.
\fi

\ifnum\ccs=0\medskip\fi
\superparagraph{Organization.}
\ifnum\ccs=0
Our robustness analysis is organized as follows: we first define several good events that capture the conditions under which we will show robustness.
Then we will define our class of admissible join-leave schedules.
With that definition at hand, we are able to formally state \cref{theorem:ourprot:analysis:maintheorem}, which asserts that our construction satisfies our robustness notion.
The proof of that theorem has two main technical parts:
in the first part, we show that construction satisfies what we want from the robustness notion (i.e., that we can retrieve data that was previously stored), provided certain good events occur. This first part does not use probabilities.
Then, in the second part, we show that the good events from the first part actually occur with sufficiently high probability, provided that our join-leave schedule is admissible.
Combined, these parts then prove \cref{theorem:ourprot:analysis:maintheorem}.
\else
Due to space limitations, we will only provide a definition of our set of admissible schedules and the main theorem here. The actual proof is then deferred to \cref{section:appendix_proof}.
\fi


To simplify notation, we consider $\SubnetDelay\geq 2$ and $\SyncDelay\geq 2$ fixed constants in the analysis and often do not write the dependency of events on these constants explicitly. For instance,
we write $\EventIsInSubnet{\subnetid}{\party}{\tau}$ and $\EventStaysInSubnet{\subnetid}{\party}{\tau_0}{\tau_1}$ instead of $\EventIsInSubnetDelay{\subnetid}{\party}{\tau}$ and $\EventStaysInSubnetDelay{\subnetid}{\party}{\tau_0}{\tau_1}$.

\ifnum\ccs=0
\subsubsection{Good Events and Admissible Schedules}
\else
\fi

\def \DEFINEGOODEVENTS{%
We start by defining some good events that capture the conditions under which we will show that our protocol does not fail.

\begin{definition}[Good Events]\label{def:GoodEvents}
We define the following good events:
\begin{myitemize}
    \item Event $\EventSubnetprotGood{T}$: This event occurs if the property we want from the robustness definition of $\subnetprot$ holds until time $T$, i.e.,
        for every time $0\leq \tau \leq T$, pairs of honest parties $\party,\party'$, subnet identifier $\subnetid$, the following holds:
        if $\EventIsInSubnet{\subnetid}{\party}{\tau}$, $\EventIsInSubnet{\subnetid}{\party'}{\tau}$, $\EventCalledGetPeers{\subnetid}{\party}{\tau}$, then
        $\EventGotPeer{\subnetid}{\party}{\party'}{\tau}$.
    \item Event $\EventColumnGoodUntil{c}{T}{\Delta}$:
        This event occurs for column $c\in [k_2]$ if for every $0\leq \tau \leq T$, there exists an honest party $\party$ 
        with $\AlgCol(\party) = c$ and $\EventActiveDuration{\party}{\tau}{\tau+\Delta}$.
    \item Event $\EventGoodCell{r}{c}{\tau_1}{\tau_2}$:
        This event occurs for row $r \in [k_1]$, column $c \in [k_2]$ and times $\tau_1\leq \tau_2$ if there exists an honest party $\party$ with $\AlgCell(\party)=(r,c)$ and 
        $\EventActiveDuration{\party}{\max\{0,\tau_1-\SubnetDelay-2\}}{\tau_2}$.
\end{myitemize}
We remark that $\EventColumnGoodUntil{c}{T}{\Delta}$ captures the condition about column subnets we explained above, with overlap time $\Delta$.
We use the event $\EventGoodCell{r}{c}{\tau_1}{\tau_2}$ to express the conditions for individual cells.
The shift of $\tau_1$ by $\SubnetDelay+2$ in the latter definition is done in order to imply that $\AlgGetPeersInCell(r,c)$ will find an honest node in the subnets if called at a time between $\tau_1$ and $\tau_2$. See \cref{cor:GoodCellImpliesFindInCell} below.
Observe that we defined no event for the row subnets: for the rows, we require that there exists honest nodes with sufficient overlap \emph{among the bootstrap nodes used to join} rather than among all nodes.
This is more conveniently expressed as a condition on the join-leave schedule $\JoinLeaveSchedule$, defined below as \emph{$\JoinLeaveSchedule$ uses good bootstrap nodes} in \cref{def:ourprot:analysis:admissibleschedules}.
\end{definition}
}
\ifnum\ccs=0
\DEFINEGOODEVENTS
\fi

\begin{definition}[Admissible Schedules]
\label{def:ourprot:analysis:admissibleschedules}
Let $\JoinLeaveSchedule$ be a join-leave schedule.
We define the following terminology:
 \begin{myitemize}
    \item We say that $\JoinLeaveSchedule$ \emph{guarantees $N$ honest parties with overlap $\OverlapTime$}, if for all time slots $\tau$, there are at least $N$ honest parties $\party$ which are active (at least) from $\tau$ to $\tau + \OverlapTime$ (both inclusive).
    \item We say that $\JoinLeaveSchedule$ \emph{respects bootstrap nodes} if the following holds:
    \begin{myitemize}
        \item The auxiliary data $\AuxRole$ provided to parties in $\iJoin$ or $\iInit$ is in $\{0,1\}$.
        We will denote parties for which $\AuxRole = 1$ as \emph{prospective bootstrap nodes}.
        \item Whenever $\JoinLeaveSchedule$ schedules a party $\party$ to join via honest bootstrap nodes $\{\party_1,\ldots,\party_t\}$, each $\party_i$ is a prospective bootstrap node.
    \end{myitemize}
    \item We say that $\JoinLeaveSchedule$ \emph{uses good bootstrap nodes} if for every honest party joining at time $\tau$ via bootstrap nodes $\{\party_1,\ldots,\party_t\}$, there exists an honest bootstrap node $\party_i$ for which the following both hold:
    \begin{myitemize}
     \item Either $\tau \geq \SubnetDelay$ and $\party_i$ is active at time $\tau-\SubnetDelay$ or $\tau < \SubnetDelay$ and $\party_i$ is one of the initial parties.
     \item $\party_i$ is active at time $\tau + \SubnetDelay$.
    \end{myitemize}
 \end{myitemize}
Finally, we define the admissible join-leave schedules for $\ourprot$:

Let $N$ be a positive integer, $\AdmissibleSchedulesSubnet$ be a set of admissible schedules for the underlying subnet protocol and $\OverlapTime \geq \max\{2\SubnetDelay + 2, 2\SyncDelay + \SubnetDelay + 2\}$.
We define the set of $(N,\OverlapTime,\AdmissibleSchedulesSubnet)$-admissible schedules for $\ourprot$ as the set of join-leave schedules which respect bootstrap nodes, use good bootstrap nodes, guarantee $N$ honest parties with overlap $\OverlapTime$ and which are in $\AdmissibleSchedulesSubnet$.
\end{definition}

\medskip\noindent
With this definition at hand, we now formally state the robustness property that we will prove.
\ifnum\ccs=1
Due to space limitations, the proof is found in \cref{section:appendix_proof}.
\fi
Recall that we assume that $k_2$ divides $m$.
\begin{theorem}[Robustness of $\ourprot$]
    Let $\rdspred\colon \HandleSpace \times [\FileLength] \times \SymbolSpace \rightarrow \bool$ be position-binding, cf.~\cref{remark:positionbinding}.
    Assume that $\subnetprot$ is $(\RobustnessErrorSubnet,\ProtocolLifetimeSubnet,\SubnetDelay,\AdmissibleSchedulesSubnet)$-robust according to \cref{def:subnetdiscovery:robustness}.
    Let $N$ be a positive integer and $\OverlapTime \geq \OverlapTimeMin$, where $\OverlapTimeMin = \max\{2\SubnetDelay + 2, 2\SyncDelay + \SubnetDelay + 2\}$.
    Then, $\ourprot$ is $(\rdspred,\RelPosCorrThresh,\RobustnessError,\ProtocolLifetime,\StoreGetDelay,\allowbreak\GetDelay,\AdmissibleSchedules)$-robust according to \cref{def:rda:robustness}, where \begin{myitemize}
        \item \emph{Admissible Schedules.} The set $\AdmissibleSchedules$ is defined as the set of $(N,\allowbreak\OverlapTime,\allowbreak\AdmissibleSchedulesSubnet)$-admissible schedules in \cref{def:ourprot:analysis:admissibleschedules}.
        \item \emph{Lifetime.} We have $\ProtocolLifetime = \ProtocolLifetimeSubnet - 1$.
        \item \emph{Store-Get Delay.} We have $\StoreGetDelay = 2$.
        \item \emph{Get Delay.} We have $\GetDelay = 2$.
        \item \emph{Corrupted Symbol Fraction.} We may choose any $0 < \RelPosCorrThresh < 1$ (but $\RobustnessError$ depends on it).
        \ifnum\ccs=0
        \item \emph{Corruption Sets.} We have corruption sets $\Corrupted_\party^\tau$ as follows: For an honest party $\party$ and time slot $\tau$, set
        \[
         \Corrupted_\party^\tau = \{i\in [\FileLength] \mid \neg \EventGoodCell{r}{c}{\max\{0,\tau-\SyncDelay\}}{(\tau+1)} \text{ for } c=\AlgGetColForSymbol(i)\}
        \]
        where $r = \AlgRow(\party)$.
        \fi
        \item \emph{Error.} We have an error probability $\RobustnessError$ bounded by
        \[
            \RobustnessError \leq \RobustnessErrorSubnet + \Bigl\lceil \frac{\ProtocolLifetime+2}{\OverlapTime-\OverlapTimeMin+1}\Bigr\rceil
                \cdot
            \Bigl(k_1 2^{\binaryEntropy(\varepsilon)k_2}e^{-\varepsilon\frac{N}{k_1}} + k_2 e^{-\frac{N}{k_2}}\Bigr)\enspace,
        \]
        where $\binaryEntropy(\varepsilon) = -\varepsilon \log_2(\varepsilon) - (1-\varepsilon)\log_2(1-\varepsilon)$ is the binary entropy function.
    \end{myitemize}
    \label{theorem:ourprot:analysis:maintheorem}
\end{theorem}
Let us give some explanation for the individual constituents of the error bound in the above~\cref{theorem:ourprot:analysis:maintheorem}. We refer to the proof of the theorem for more details.
Clearly, the $\RobustnessErrorSubnet$-term comes from the possibility of the subnet protocol failing.
The $\lceil(\ProtocolLifetime+2)/(\OverlapTime-\OverlapTimeMin+1)\rceil$-term comes from a union bound over the lifetime $\ProtocolLifetime$;
the denominator $\OverlapTime-\OverlapTimeMin+1$ in that expression is an optimization that uses the fact that the failure events for different times are not independent.
The expression $k_1 2^{\binaryEntropy(\varepsilon)k_2}e^{-\varepsilon N / k_1}$ accounts for the bad event that there are too many failing (i.e.\ without an active honest party with sufficient overlap) cells in any row.
The expression $k_2 e^{-{N}/{k_2}}$ accounts for the bad event that in some column, there is no honest party with sufficient overlap.



\subsection{Efficiency}
\label{sec:efficiency}
We now discuss the efficiency of our protocol.
Note that while the number of adversarial parties has no impact on the robustness of our protocol, efficiency depends on the total number of parties, and therefore on the number of adversarial parties.
Throughout this section, we denote the (upper bound on) the number of simultaneously active honest parties by $\NumberOfHonestPartiesMax$ and the (upper bound on) the number of total simultaneously active\footnote{%
Recall that we defined \emph{active} for honest parties by them having joined. Since malicious parties do not call $\iJoin$ or $\iInit$, \emph{active} for them shall mean that honest parties may keep a connection to it.
}
parties by $\NumberOfPartiesMax \geq \NumberOfHonestPartiesMax$.
We assume that the subnet discovery protocol is instantiated as in \cref{sec:appendix:simplesubnetprotocol}, with the optimizations described there.


\superparagraph{Latency and Round Complexity.}
When an honest party calls $\iGet$, the latency is at most $3$. Precisely, a call to $\iGet$ made at time $\tau$ returns either a valid $x$ at time $\tau+2$ or $\bot$ at time $\tau+3$.
For interface $\iStore$, the user calling $\iStore$ at $\tau$ does not expect a response and the call returns immediately.
By the end of round $\tau+2$, all messages caused by this call have been received.
For $\iJoin$ called at time $\tau$, the call terminates at $\tau+2+\SubnetDelay$.
Synchronization of data may take longer, but we emphasize that as soon as the call is terminated, the party can already participate in the protocol via $\iGet$ and $\iStore$.

\superparagraph{Data Per Node.}
We first examine how much data a node needs to store.
For simplicity, assume that exactly one file $(x_1,\dots,x_\FileLength) \in \SymbolSpace^\FileLength$ is stored.
Recall that each party resides in exactly one column and that nodes in a column are responsible for storing $\FileLength/k_2$ symbols.
Therefore, the data storage requirement for each node is $\FileLength/k_2$.

\superparagraph{Connections Per Node.}
Now, let us consider how many connections to peers a node needs to keep.
We consider a fixed point in time $\tau$.
Our goal is to derive a bound for the expected number of connections that a fixed node $\party$ has to maintain, taking the randomness over the random oracle.
We also assume that $\party$ is not a bootstrap node ($\AuxRole = 0$). When joining, $\party$ needs to connect to each of its supplied bootstrap nodes, but those connections can be dropped after two rounds and we ignore those.
Let $\AlgCell(\party)=(r,c)$. Let $\mathsf{Peers}$ be the number of peers $\party$ has.
As argued above, we may assume that $\party$ will only possibly keep a connection to another $\party'$ if $\AlgRow(\party') = r$ or $\AlgCol(\party') = c$ and $\party'$ is active.
$\AlgRow(\party') = r \lor \AlgRow(\party') = c$ happens with probability (over the choice of random oracle) ${(k_1 + k_2 - 1)}/({k_1 k_2})$.
Since there are at most $\NumberOfPartiesMax-1$ such parties $\party'$, we can bound the expected number of connections as
\[
 \Expect{|\mathsf{Peers}|} \leq (\NumberOfPartiesMax-1)\smash{\frac{k_1+k_2-1}{k_1 k_2}} \leq  \smash{\frac{\NumberOfPartiesMax-1}{k_1}} + \smash{\frac{\NumberOfPartiesMax-1}{k_2}}\enspace.
\]

\superparagraph{Communication Complexity.}
For communication complexity, we consider the expected total communication that is caused by honest parties calling protocol interfaces.
We do only count the outgoing data of honest parties, as one can never limit the amount of data that dishonest parties send.
For simplicity, we denote the maximum size of a single message (with prefix $\msgGet,\msgGetRsp$ or $\msgStore$, $\msgStoreFwd$, $\msgJoin$) by $\MessageSize$.
Note that $\MessageSize$ depends on the parameters $|\HandleSpace|$, $\FileLength$, and $|\SymbolSpace|$.
For messages of the form $(\msgJoinRsp,\mathbf{P}^{(i)}_c)$, we treat them as $|\mathbf{P}^{(i)}_c|$ messages of size $\MessageSize$.

First, consider an honest party $\party$ calling $\iGet(h,i)$.
Here, $\party$ identifies a cell $(r,c)$, where column $c$ is responsible for storing the $i$th symbol, then sends a $\msgGet$-message to all nodes in that cell. The nodes respond with a $\msgGetRsp$-message. We use $|\cdot|$ to denote the length of strings, and denote by $\mathsf{NodesInCell}_{r,c}$ (resp. $\mathsf{HonNodesInCell}_{r,c}$) the random variable modeling the number of nodes (resp. honest nodes) in cell $(r,c)$.
Then, the expected (over the choice of the random oracle) communication complexity is bounded by
\ifnum\ccs=0
\begin{align*}
    \Expect{\mathsf{NodesInCell}_{r,c}} \cdot \MessageSize + \Expect{\mathsf{HonNodesInCell}_{r,c}} \cdot \MessageSize \leq \left(\frac{\NumberOfPartiesMax}{k_1k_2} + \frac{\NumberOfHonestPartiesMax}{k_1k_2}\right)\cdot \MessageSize,
\end{align*}
\else
\begin{align*}
    &\Expect{\mathsf{NodesInCell}_{r,c}} \cdot \MessageSize + \Expect{\mathsf{HonNodesInCell}_{r,c}} \cdot \MessageSize\\ {}\leq{}& \left(\frac{\NumberOfPartiesMax}{k_1k_2} + \frac{\NumberOfHonestPartiesMax}{k_1k_2}\right)\cdot \MessageSize,
\end{align*}
\fi
where the first summand accounts for $\msgGet$-messages and the second for $\msgGetRsp$-responses of honest parties.

Next, we consider an honest party $\party$ calling $\iStore(h,i,x)$.
Again, $\party$ identifies all parties in an appropriate cell $(r,c)$ and then sends a $\msgStore$-message. All parties in the cell then send $\msgStoreFwd$-messages to all parties in column $c$.
Denoting the number of parties in column $c$ by ${\NodesInCol_{c}}$, the following random variable models the communication complexity:
\begin{align*}
    &\mathsf{NodesInCell}_{r,c} \cdot \MessageSize +{\mathsf{HonNodesInCell}_{r,c}} \cdot  \NodesInCol_{c} \cdot \MessageSize.
\end{align*}
To get an upper bound on the expectation, we use \cref{cor:expectationofproduct} with $Y = \mathsf{HonNodesInCell}_{r,c}$, $X = \NodesInCol_{c}$, and $p = 1/k_2$. Then, as soon as $\NumberOfPartiesMax \geq 3 k_2 \ln k_2$, the expected communication becomes
\begin{equation*}
    \left(\frac{\NumberOfPartiesMax}{k_1k_2} + \frac{3\NumberOfHonestPartiesMax \NumberOfPartiesMax }{k_1k_2^2} \right)\cdot \MessageSize.
\end{equation*}
Finally, we consider an honest party $\party$ with $\AuxRole = 0$ joining the network.
We ignore communication complexity from data synchronization, as this depends on the size of the data.
Assume that it joins via $t$ bootstrap nodes $\{\party_1,\ldots,\party_t\}$.
As we only count outgoing messages of honest parties, the worst case is when these $t$ bootstrap nodes are all honest.
Let $\NumberOfBootstrapNodes$ denote the total number of currently active bootstrap nodes.

Denote by $\CommunicationJoinRow(S)$ resp.\ $\CommunicationJoinCol(S)$ the communication complexity of joining a row resp.\ column subnet containing $S$ nodes,
where for the row subnets $\subnetid_r$ we do not count (bootstrap) nodes $\party'$ with $\AlgRow(\party') \neq r$.
With our protocol from \cref{sec:appendix:simplesubnetprotocol}, we have
\vspace{-1.5ex plus 1ex}
\begin{align*}
    \CommunicationJoinCol(S) \leq \MessageSize( 2 + 2S).
\end{align*}
Using the optimizations that we have outlined in the end of \cref{sec:appendix:simplesubnetprotocol}, the \msgJoinSubnetPullRsp-messages to parties do not contain bootstrap nodes, and so we have
\ifnum\ccs=1
\begin{align*}
    \CommunicationJoinRow(S) \leq \MessageSize( 2 + 2S + \NumberOfBootstrapNodes).
\end{align*}
\else
\begin{align*}
    \CommunicationJoinRow(S) \leq \MessageSize( 2 + 2(S + \NumberOfBootstrapNodes) - \NumberOfBootstrapNodes) = \MessageSize( 2 + 2S + \NumberOfBootstrapNodes).
\end{align*}
\fi

Denoting the number of parties in row $r$ by ${\NodesInRow_{r}}$, we can bound the communication complexity induced by an honest party with $\AuxRole=0$ joining by
\ifnum\ccs=1
\begin{align*}
    & t \cdot \CommunicationJoinRow({\NodesInRow_{r}})
      + t \left(\MessageSize + \NodesInCol_{c} \cdot \MessageSize\right)\\
    & {}+ \NodesInCol_{c}\cdot \CommunicationJoinCol\left(\NodesInCol_{c}\right) \\
    \leq &\MessageSize \bigl( 3t + 2t\cdot {\NodesInRow_{r}} + (t+2)\cdot \NodesInCol_{c} \\
         &\phantom{\MessageSize \bigl(}  + 2\cdot \NodesInCol_{c}^2 + t\NumberOfBootstrapNodes\bigr).
\end{align*}
\else
\begin{align*}
 &t \cdot \CommunicationJoinRow({\NodesInRow_{r}})
      + t \left(\MessageSize + \NodesInCol_{c} \cdot \MessageSize\right)
   + \NodesInCol_{c}\cdot \CommunicationJoinCol\left(\NodesInCol_{c}\right) \\
    \leq~~&\MessageSize \left( 3t + 2t\cdot {\NodesInRow_{r}} + (t+2)\cdot \NodesInCol_{c}   + 2\cdot \NodesInCol_{c}^2 + t\NumberOfBootstrapNodes\right).
\end{align*}
\fi
We get the expected value of this using linearity of expectation and \cref{lemma:secondmomentofbinom} for the quadratic term.
Namely, the expected communication is bounded by
\begin{align*}
    \left(3t + t\NumberOfBootstrapNodes + \frac{t\smash{\NumberOfPartiesMax}}{k_1} + \frac{(t+4) \smash{\NumberOfPartiesMax} k_2 - 2\smash{\NumberOfPartiesMax} + \smash{(\NumberOfPartiesMax)^2}}{\smash{k_2^2}}\right)\cdot \MessageSize.
\end{align*}

\section{Extensions and Practical Considerations}
\label{sec:extensions}
In practical deployments of our protocol, various modifications and extensions may be beneficial.
This section informally discusses several such enhancements.

\superparagraph{Virtual Nodes.}
The robustness of our protocol relies on probabilistic guarantees that depend heavily on the number of honest nodes.
We can artificially increase this number by making each node simulate $K \geq 1$ \emph{virtual nodes}.
This enables us to apply our analysis to a larger effective set of honest nodes, increasing their count from $N$ to $N' = N \cdot K$.
A key advantage of this approach is reduced variance in the probabilistic processes involved, leading to improved protocol parameters such as smaller $\RobustnessError$ and $\RelPosCorrThresh$.
Of course, this comes at the cost of increased storage requirements per node, as each node participates in multiple columns.
This tradeoff can be beneficial for system designers who can afford additional storage but cannot increase $N$. The use of virtual nodes to improve variance is a well-established technique, e.g., used in consistent hashing~\cite{STOC:KLLPLL97}.

\superparagraph{Sending Leave Messages.}
An honest party may additionally send leave messages before leaving to all of its peers. In this way, peers can free memory and do not need to keep connections open. Importantly, such leave messages must be authenticated, as otherwise an adversary could use them to break connections between two honest peers.

\superparagraph{Efficient Data Synchronization.}
Our protocol requires nodes in the same column to synchronize stored data, particularly when a new node joins or when new data is added. For simplicity, we have described the naive approach where each node transmits all of its stored data to its peers.
However, more efficient methods exist: Set reconciliation techniques~\cite{DBLP:journals/tit/MinskyTZ03,ICALP:BelKucWal24,DBLP:conf/sigcomm/YangGA24} can significantly reduce communication overhead in such settings.
These methods can be integrated into our protocol.

\superparagraph{Synchronization Time.}
The synchronization delay, $\SyncDelay$, influences the admissibility condition of our join-leave schedules (\cref{def:ourprot:analysis:admissibleschedules}).
Further, while our theoretical description of our protocol required that parties have explicit knowledge of $\SyncDelay$ when responding to $\msgSync$-messages, parties executing the protocol in practice do not need to know the exact value of $\SyncDelay$. Instead, they should respond to $\msgSync$-messages as quickly as possible.
Consequently, if $\SyncDelay$ increases, slightly stronger honesty assumptions are required, but the protocol’s code remains unchanged.
This property is particularly relevant in practice, as synchronization time may increase over time with the growth of stored data.
Hence, we can think of $\SyncDelay$ as depending on the time, which also means the admissibility condition changes over time.

\superparagraph{Lazy Cell Messages.}
Many operations in our protocol involve querying \emph{all} nodes in a cell in parallel.
For example, in interface $\iGet$, a node requests data from all nodes in a cell simultaneously.
A more bandwidth-efficient alternative is to query nodes sequentially, stopping once data satisfying the predicate is retrieved, or to query in small batches, proceeding incrementally. A practical heuristic is to prioritize nodes based on past responsiveness, contacting the most responsive nodes first. This approach introduces a bandwidth-latency tradeoff that can be adjusted to suit different system requirements.

\superparagraph{Less Powerful Bootstrap Nodes.}
In our protocol description, we assumed that parties take on either the role of a normal node or a bootstrap node.
We defined bootstrap nodes to be present in every row.
In this way, when a party wishes to join a cell $(r,c)$, it can join its row subnet $r$ via the bootstrap node.
Further, it can join its column subnet $c$ through nodes in column $c$ obtained from the bootstrap node. Consequently, bootstrap nodes are more powerful and have higher bandwidth requirements.

If, in practice, one wishes to avoid relying on such powerful bootstrap nodes, a slight modification to the join process can address this issue.
Specifically, suppose a party $\party$ joins the protocol through a (normal) node\footnote{The reader may recall that a party can join via multiple nodes, but for simplicity, we describe the process assuming it joins via a single honest node.} $\party'$.
Assume that $\party$ intends to join cell $(r,c)$ while $\party'$ resides in cell $(r',c')$.
Then, $\party$ can request from $\party'$ a list of all nodes in cell $(r,c')$ and use them to join its row subnet $r$. Similarly, $\party$ can request from $\party'$ a list of all nodes in cell $(r',c)$ and use them to join its column subnet $c$.
This approach works as long as $\party'$ is honest and the cells $(r,c')$ and $(r',c)$ each contain at least one honest party.




\ifnum\ccs=1
\ifnum\ccs=1
    \newcommand{\estfigWidth}{10cm}
    \newcommand{\estfigHeight}{4cm}
\else
    \newcommand{\estfigWidth}{9cm}
    \newcommand{\estfigHeight}{4cm}
\fi

\begin{figure*}[htbp] 
    \centering

    \pgfplotscreateplotcyclelist{my cycle list}{
        {blue, mark=*},
        {red, mark=square*},
        {black, mark=triangle*},
        {orange, mark=diamond*},
    }

    \begin{minipage}{0.48\textwidth}
        \resizebox{\textwidth}{!}{
        \begin{tikzpicture}
            \begin{axis}[
                width=\estfigWidth,
                height=\estfigHeight,
                xlabel={$k_2$},
                ylabel={$\max k_1$},
                grid=major,
                legend style={at={(0.5,1.1)}, anchor=south, legend columns=2},
                mark repeat=60,
                cycle list name=my cycle list,
                xmode=log, 
                ymode=log, 
            ]

            \addplot table [x=k2, y=k1, col sep=comma] {csvdata/estimates_data_10_100000.csv}; 
            \addplot table [x=k2, y=k1, col sep=comma] {csvdata/estimates_data_10_10000.csv}; 
            \addplot table [x=k2, y=k1, col sep=comma] {csvdata/estimates_data_10_5000.csv}; 
            \addplot table [x=k2, y=k1, col sep=comma] {csvdata/estimates_data_10_1000.csv}; 
            \end{axis}
        \end{tikzpicture}
        }
    \end{minipage}
    \begin{minipage}{0.48\textwidth}
        \resizebox{\textwidth}{!}{
        \begin{tikzpicture}
            \begin{axis}[
                width=\estfigWidth,
                height=\estfigHeight,
                xlabel={$k_2$},
                ylabel={$\iJoin$},
                grid=major,
                cycle list name=my cycle list,
                mark repeat=60,
                xmode=log, 
                ymode=log, 
            ]
            \addplot table [x=k2, y=join_complexity, col sep=comma] {csvdata/estimates_data_10_100000.csv};
            \addplot table [x=k2, y=join_complexity, col sep=comma] {csvdata/estimates_data_10_10000.csv};
            \addplot table [x=k2, y=join_complexity, col sep=comma] {csvdata/estimates_data_10_5000.csv};
            \addplot table [x=k2, y=join_complexity, col sep=comma] {csvdata/estimates_data_10_1000.csv};
            \end{axis}
        \end{tikzpicture}
        }
    \end{minipage}

    \begin{minipage}{0.48\textwidth}
        \resizebox{\textwidth}{!}{
        \begin{tikzpicture}
            \begin{axis}[
                width=\estfigWidth,
                height=\estfigHeight,
                xlabel={$k_2$},
                ylabel={$\iGet$},
                grid=major,
                cycle list name=my cycle list,
                mark repeat=60,
                xmode=log, 
                ymode=log, 
            ]
            \addplot table [x=k2, y=get_complexity, col sep=comma] {csvdata/estimates_data_10_100000.csv};
            \addplot table [x=k2, y=get_complexity, col sep=comma] {csvdata/estimates_data_10_10000.csv};
            \addplot table [x=k2, y=get_complexity, col sep=comma] {csvdata/estimates_data_10_5000.csv};
            \addplot table [x=k2, y=get_complexity, col sep=comma] {csvdata/estimates_data_10_1000.csv};
            \end{axis}
        \end{tikzpicture}
        }
    \end{minipage}
    \begin{minipage}{0.48\textwidth}
        \resizebox{\textwidth}{!}{
        \begin{tikzpicture}
            \begin{axis}[
                width=\estfigWidth,
                height=\estfigHeight,
                xlabel={$k_2$},
                ylabel={$\iStore$},
                grid=major,
                cycle list name=my cycle list,
                mark repeat=60,
                xmode=log, 
                ymode=log, 
            ]
            \addplot table [x=k2, y=store_complexity, col sep=comma] {csvdata/estimates_data_10_100000.csv};
            \addplot table [x=k2, y=store_complexity, col sep=comma] {csvdata/estimates_data_10_10000.csv};
            \addplot table [x=k2, y=store_complexity, col sep=comma] {csvdata/estimates_data_10_5000.csv};
            \addplot table [x=k2, y=store_complexity, col sep=comma] {csvdata/estimates_data_10_1000.csv};
            \end{axis}
        \end{tikzpicture}
        }
    \end{minipage}
    \vspace{-1em}
    \begingroup
    \ifnum\ccs=1
        \def \myresizesize\0.7\textwidth
    \else
        \def \myresizesize\textwidth
    \fi
    \resizebox{\myresizesize}{}{%
    \begin{tikzpicture}
        \begin{axis}[
            hide axis,
            scale only axis,
            height=0pt,
            width=0pt,
            legend columns=4,
            legend style={
                at={(0.5,0.5)},
                anchor=center,
                column sep=1ex
            },
            cycle list name=my cycle list,
        ]
        \addplot coordinates {(-1e10,-1e10)}; \addlegendentry{$N = 100{,}000$, $\RelPosCorrThresh=0.1$};
        \addplot coordinates {(-1e10,-1e10)}; \addlegendentry{$N = 10{,}000$, $\RelPosCorrThresh=0.1$};
        \addplot coordinates {(-1e10,-1e10)}; \addlegendentry{$N = 5{,}000$, $\RelPosCorrThresh=0.1$};
        \addplot coordinates {(-1e10,-1e10)}; \addlegendentry{$N = 1{,}000$, $\RelPosCorrThresh=0.1$};
        \end{axis}
    \end{tikzpicture}
    }
    \endgroup
    \vspace{0.05em}
    \caption{Trade-off between $k_1$ (number of rows) and $k_2$ (number of columns) and resulting complexities for $\RelPosCorrThresh = 0.1$. The complexities show the expected total message complexity when an honest party calls the respective interfaces. All plots assume $\RobustnessError \leq 10^{-9}$, and the given $N$. Scales are logarithmic.}
    \label{fig:estimatesA}
\end{figure*}
\newcommand{\figWidth}{10cm}
\ifnum\ccs=1
    \newcommand{\figHeight}{4cm}
\else
    \newcommand{\figHeight}{6cm}
\fi
\begin{figure*}[htbp]  
    \centering
    \begin{minipage}{0.48\textwidth}
        \resizebox{\textwidth}{!}{
        \begin{tikzpicture}
            \begin{axis}[
                width=\figWidth,
                height=\figHeight,
                xlabel={Time Step $\tau$},
                ylabel={$\max_\party |\Corrupted_\party^\tau| / \FileLength$},
                title={Corrupted symbols over time},
                grid=major,
                legend pos=north east,
                mark repeat=30
            ]
            \addplot table [x=Time_Step, y=Corruption_Rows_1, col sep=comma] {csvdata/simulation_data.csv}; \addlegendentry{$k_1=1$};
            \addplot table [x=Time_Step, y=Corruption_Rows_5, col sep=comma] {csvdata/simulation_data.csv}; \addlegendentry{$k_1=5$};
            \addplot table [x=Time_Step, y=Corruption_Rows_10, col sep=comma] {csvdata/simulation_data.csv}; \addlegendentry{$k_1=10$};
            \addplot table [x=Time_Step, y=Corruption_Rows_25, col sep=comma] {csvdata/simulation_data.csv}; \addlegendentry{$k_1=25$};
            \end{axis}
        \end{tikzpicture}
        }
    \end{minipage}
    \quad
    \begin{minipage}{0.48\textwidth}
        \resizebox{\textwidth}{!}{
        \begin{tikzpicture}
            \begin{axis}[
                width=\figWidth,
                height=\figHeight,
                xlabel={Time Step $\tau$},
                ylabel={$\max_\party |\mathsf{Peers}_\party^\tau|$},
                title={Number of peers over time},
                grid=major,
                legend pos=north east,
                mark repeat=30
            ]
            \addplot table [x=Time_Step, y=Connections_Rows_1, col sep=comma] {csvdata/simulation_data.csv}; \addlegendentry{$k_1=1$};
            \addplot table [x=Time_Step, y=Connections_Rows_5, col sep=comma] {csvdata/simulation_data.csv}; \addlegendentry{$k_1=5$};
            \addplot table [x=Time_Step, y=Connections_Rows_10, col sep=comma] {csvdata/simulation_data.csv}; \addlegendentry{$k_1=10$};
            \addplot table [x=Time_Step, y=Connections_Rows_25, col sep=comma] {csvdata/simulation_data.csv}; \addlegendentry{$k_1=25$};
            \end{axis}
        \end{tikzpicture}
        }
    \end{minipage}
    \caption{Simulation results for $k_2 = 100$ columns and different number of rows $k_1$. We simulated a join-leave schedule that leads to $2{,}500$ honest parties (after a warmup phase) and each party staying for $50$ time steps.}
    \label{fig:simulationplots}
\end{figure*}
\fi
\section{Benchmarks}
\label{sec:benches}
\ifnum\anonymous=1
In this section, we derive concrete numbers from our robustness analysis and present some simulation results\footnote{Code for our plots:~\url{https://anonymous.4open.science/r/rda-experiments-DE09}.}.
\else
In this section, we derive concrete numbers from our robustness analysis and present some simulation results. The code for generating our plots can be found at
\begin{center}
    \url{https://github.com/b-wagn/rda-experiments}.
\end{center}
\fi

\subsection{Concrete Numbers from Analytical Results}
The robustness theorem (\cref{theorem:ourprot:analysis:maintheorem}) shows a relation between the error $\RobustnessError$, the maximum fraction of corrupted symbols $\RelPosCorrThresh$, the minimum number of honest parties $N$, and the number of rows and columns $k_1$ and $k_2$. Also, by \cref{sec:efficiency}, there is a relation between $k_1$ and $k_2$ and the efficiency of our protocol. We show these dependencies in \cref{fig:estimatesA,fig:estimatesB}, and explain the plots next.

\ifnum\ccs=0
\ifnum\ccs=1
    \newcommand{\estfigWidth}{10cm}
    \newcommand{\estfigHeight}{4cm}
\else
    \newcommand{\estfigWidth}{9cm}
    \newcommand{\estfigHeight}{4cm}
\fi

\begin{figure*}[htbp] 
    \centering

    \pgfplotscreateplotcyclelist{my cycle list}{
        {blue, mark=*},
        {red, mark=square*},
        {black, mark=triangle*},
        {orange, mark=diamond*},
    }

    \begin{minipage}{0.48\textwidth}
        \resizebox{\textwidth}{!}{
        \begin{tikzpicture}
            \begin{axis}[
                width=\estfigWidth,
                height=\estfigHeight,
                xlabel={$k_2$},
                ylabel={$\max k_1$},
                grid=major,
                legend style={at={(0.5,1.1)}, anchor=south, legend columns=2},
                mark repeat=60,
                cycle list name=my cycle list,
                xmode=log, 
                ymode=log, 
            ]

            \addplot table [x=k2, y=k1, col sep=comma] {csvdata/estimates_data_10_100000.csv}; 
            \addplot table [x=k2, y=k1, col sep=comma] {csvdata/estimates_data_10_10000.csv}; 
            \addplot table [x=k2, y=k1, col sep=comma] {csvdata/estimates_data_10_5000.csv}; 
            \addplot table [x=k2, y=k1, col sep=comma] {csvdata/estimates_data_10_1000.csv}; 
            \end{axis}
        \end{tikzpicture}
        }
    \end{minipage}
    \begin{minipage}{0.48\textwidth}
        \resizebox{\textwidth}{!}{
        \begin{tikzpicture}
            \begin{axis}[
                width=\estfigWidth,
                height=\estfigHeight,
                xlabel={$k_2$},
                ylabel={$\iJoin$},
                grid=major,
                cycle list name=my cycle list,
                mark repeat=60,
                xmode=log, 
                ymode=log, 
            ]
            \addplot table [x=k2, y=join_complexity, col sep=comma] {csvdata/estimates_data_10_100000.csv};
            \addplot table [x=k2, y=join_complexity, col sep=comma] {csvdata/estimates_data_10_10000.csv};
            \addplot table [x=k2, y=join_complexity, col sep=comma] {csvdata/estimates_data_10_5000.csv};
            \addplot table [x=k2, y=join_complexity, col sep=comma] {csvdata/estimates_data_10_1000.csv};
            \end{axis}
        \end{tikzpicture}
        }
    \end{minipage}

    \begin{minipage}{0.48\textwidth}
        \resizebox{\textwidth}{!}{
        \begin{tikzpicture}
            \begin{axis}[
                width=\estfigWidth,
                height=\estfigHeight,
                xlabel={$k_2$},
                ylabel={$\iGet$},
                grid=major,
                cycle list name=my cycle list,
                mark repeat=60,
                xmode=log, 
                ymode=log, 
            ]
            \addplot table [x=k2, y=get_complexity, col sep=comma] {csvdata/estimates_data_10_100000.csv};
            \addplot table [x=k2, y=get_complexity, col sep=comma] {csvdata/estimates_data_10_10000.csv};
            \addplot table [x=k2, y=get_complexity, col sep=comma] {csvdata/estimates_data_10_5000.csv};
            \addplot table [x=k2, y=get_complexity, col sep=comma] {csvdata/estimates_data_10_1000.csv};
            \end{axis}
        \end{tikzpicture}
        }
    \end{minipage}
    \begin{minipage}{0.48\textwidth}
        \resizebox{\textwidth}{!}{
        \begin{tikzpicture}
            \begin{axis}[
                width=\estfigWidth,
                height=\estfigHeight,
                xlabel={$k_2$},
                ylabel={$\iStore$},
                grid=major,
                cycle list name=my cycle list,
                mark repeat=60,
                xmode=log, 
                ymode=log, 
            ]
            \addplot table [x=k2, y=store_complexity, col sep=comma] {csvdata/estimates_data_10_100000.csv};
            \addplot table [x=k2, y=store_complexity, col sep=comma] {csvdata/estimates_data_10_10000.csv};
            \addplot table [x=k2, y=store_complexity, col sep=comma] {csvdata/estimates_data_10_5000.csv};
            \addplot table [x=k2, y=store_complexity, col sep=comma] {csvdata/estimates_data_10_1000.csv};
            \end{axis}
        \end{tikzpicture}
        }
    \end{minipage}
    \vspace{-1em}
    \begingroup
    \ifnum\ccs=1
        \def \myresizesize\0.7\textwidth
    \else
        \def \myresizesize\textwidth
    \fi
    \resizebox{\myresizesize}{}{%
    \begin{tikzpicture}
        \begin{axis}[
            hide axis,
            scale only axis,
            height=0pt,
            width=0pt,
            legend columns=4,
            legend style={
                at={(0.5,0.5)},
                anchor=center,
                column sep=1ex
            },
            cycle list name=my cycle list,
        ]
        \addplot coordinates {(-1e10,-1e10)}; \addlegendentry{$N = 100{,}000$, $\RelPosCorrThresh=0.1$};
        \addplot coordinates {(-1e10,-1e10)}; \addlegendentry{$N = 10{,}000$, $\RelPosCorrThresh=0.1$};
        \addplot coordinates {(-1e10,-1e10)}; \addlegendentry{$N = 5{,}000$, $\RelPosCorrThresh=0.1$};
        \addplot coordinates {(-1e10,-1e10)}; \addlegendentry{$N = 1{,}000$, $\RelPosCorrThresh=0.1$};
        \end{axis}
    \end{tikzpicture}
    }
    \endgroup
    \vspace{0.05em}
    \caption{Trade-off between $k_1$ (number of rows) and $k_2$ (number of columns) and resulting complexities for $\RelPosCorrThresh = 0.1$. The complexities show the expected total message complexity when an honest party calls the respective interfaces. All plots assume $\RobustnessError \leq 10^{-9}$, and the given $N$. Scales are logarithmic.}
    \label{fig:estimatesA}
\end{figure*}
\fi

\ifnum\ccs=0
\begin{figure*}[htbp]
    \centering

    \pgfplotscreateplotcyclelist{my cycle list}{
        {black, mark=x},
        {green!50!black, mark=square*},
        {blue, mark=diamond*},
    }

    \begin{minipage}{0.48\textwidth}
        \resizebox{\textwidth}{!}{
        \begin{tikzpicture}
            \begin{axis}[
                width=\estfigWidth,
                height=\estfigHeight,
                xlabel={$k_2$},
                ylabel={$\max k_1$},
                grid=major,
                legend style={at={(0.5,1.1)}, anchor=south, legend columns=2},
                mark repeat=60,
                cycle list name=my cycle list,
                xmode=log, 
                ymode=log, 
            ]

            \addplot table [x=k2, y=k1, col sep=comma] {csvdata/estimates_data_5_100000.csv}; 
            \addplot table [x=k2, y=k1, col sep=comma] {csvdata/estimates_data_5_10000.csv}; 
            \addplot table [x=k2, y=k1, col sep=comma] {csvdata/estimates_data_5_5000.csv}; 
            \end{axis}
        \end{tikzpicture}
        }
    \end{minipage}
    \begin{minipage}{0.48\textwidth}
        \resizebox{\textwidth}{!}{
        \begin{tikzpicture}
            \begin{axis}[
                width=\estfigWidth,
                height=\estfigHeight,
                xlabel={$k_2$},
                ylabel={$\iJoin$},
                grid=major,
                cycle list name=my cycle list,
                mark repeat=60,
                xmode=log, 
                ymode=log, 
            ]
            \addplot table [x=k2, y=join_complexity, col sep=comma] {csvdata/estimates_data_5_100000.csv};
            \addplot table [x=k2, y=join_complexity, col sep=comma] {csvdata/estimates_data_5_10000.csv};
            \addplot table [x=k2, y=join_complexity, col sep=comma] {csvdata/estimates_data_5_5000.csv};
            \end{axis}
        \end{tikzpicture}
        }
    \end{minipage}

    \begin{minipage}{0.48\textwidth}
        \resizebox{\textwidth}{!}{
        \begin{tikzpicture}
            \begin{axis}[
                width=\estfigWidth,
                height=\estfigHeight,
                xlabel={$k_2$},
                ylabel={$\iGet$},
                grid=major,
                cycle list name=my cycle list,
                mark repeat=60,
                xmode=log, 
                ymode=log, 
            ]
            \addplot table [x=k2, y=get_complexity, col sep=comma] {csvdata/estimates_data_5_100000.csv};
            \addplot table [x=k2, y=get_complexity, col sep=comma] {csvdata/estimates_data_5_10000.csv};
            \addplot table [x=k2, y=get_complexity, col sep=comma] {csvdata/estimates_data_5_5000.csv};
            \end{axis}
        \end{tikzpicture}
        }
    \end{minipage}
    \begin{minipage}{0.48\textwidth}
        \resizebox{\textwidth}{!}{
        \begin{tikzpicture}
            \begin{axis}[
                width=\estfigWidth,
                height=\estfigHeight,
                xlabel={$k_2$},
                ylabel={$\iStore$},
                grid=major,
                cycle list name=my cycle list,
                mark repeat=60,
                xmode=log, 
                ymode=log, 
            ]
            \addplot table [x=k2, y=store_complexity, col sep=comma] {csvdata/estimates_data_5_100000.csv};
            \addplot table [x=k2, y=store_complexity, col sep=comma] {csvdata/estimates_data_5_10000.csv};
            \addplot table [x=k2, y=store_complexity, col sep=comma] {csvdata/estimates_data_5_5000.csv};
            \end{axis}
        \end{tikzpicture}
        }
    \end{minipage}
    \vspace{-1em}
    \begin{tikzpicture}
        \begin{axis}[
            hide axis,
            scale only axis,
            height=0pt,
            width=0pt,
            legend columns=4,
            legend style={
                at={(0.5,0.5)},
                anchor=center,
                column sep=1ex
            },
            cycle list name=my cycle list,
        ]
        \addplot coordinates {(-1e10,-1e10)}; \addlegendentry{$N = 100{,}000$, $\RelPosCorrThresh=0.05$};
        \addplot coordinates {(-1e10,-1e10)}; \addlegendentry{$N = 10{,}000$, $\RelPosCorrThresh=0.05$};
        \addplot coordinates {(-1e10,-1e10)}; \addlegendentry{$N = 5{,}000$, $\RelPosCorrThresh=0.05$};
        \end{axis}
    \end{tikzpicture}
    \vspace{0.5em}
    \caption{Trade-off between $k_1$ (number of rows) and $k_2$ (number of columns) and resulting complexities for $\RelPosCorrThresh = 0.05$. The complexities show the expected total message complexity when an honest party calls the respective interfaces. All plots assume $\RobustnessError \leq 10^{-9}$ and the given $N$. Scales are logarithmic.}
    \label{fig:estimatesB}
\end{figure*}
\fi

\superparagraph{Setting and Assumptions.}
The subnet discovery protocol in \cref{sec:appendix:simplesubnetprotocol} is perfectly robust, i.e., we can assume $\RobustnessErrorSubnet = 0$ in \cref{theorem:ourprot:analysis:maintheorem}.
We assume that one round in the synchronous model takes $4$ seconds and target a protocol lifetime of $10$ years (which determines the value for $\ProtocolLifetime$). We set $\OverlapTimeMin = 450$, which means we assume that data can be synchronized in at most roughly $15$ minutes. We set $\OverlapTime$ by assuming the average honest nodes stays online for $6$ hours.
We target an error probability of at most $\RobustnessError \leq 10^{-9}$.
For these values, we then take the following approach to arrive at the plots in \cref{fig:estimatesA,fig:estimatesB}: we consider a variety of combinations of $N$ and $\RelPosCorrThresh$, and for each of those, we do the following: (1) we generate a trade-off curve between $k_1$ and $k_2$. This curve shows for each $k_2$ (number of columns) the maximum $k_1$ (number of rows) for which the bound in \cref{theorem:ourprot:analysis:maintheorem} yields the desired $\RobustnessError$; (2) we then generate plots showing the expected communication complexity of $\iGet$ and $\iJoin$ (as analyzed in \cref{sec:efficiency}) depending on $k_2$, viewing $k_1$ as a function of $k_2$ using (1).
For (2), we assume $\MessageSize = 1$, i.e., we plot message complexity. For the complexity of $\iJoin$, we assume $t = 50$ bootstrap nodes, and $\NumberOfBootstrapNodes = 100$. We assume that $\NumberOfHonestPartiesMax = 2N$ and $\NumberOfPartiesMax = 5N$.

\superparagraph{Interpretation.}
We can interpret these plots by first observing that a given number of honest nodes $N$ can only support a limited number of cells, assuming that (most) of the cells have to contain one honest node. This creates the tradeoff between $k_1$ and $k_2$ as shown in the plots; the total number of cells is $k_1 \cdot k_2$ (top left graph). Intuitively, one might think that the plot would represent the hyperbola $k_1 \cdot k_2 \approx N$.
However, with a little thought experiment we can see why the fact that we fix an error threshold (for example, $\RelPosCorrThresh = 0.1$) creates a deviation from the hyperbola:
\begin{myitemize}
    \item If we choose $k_1=1$, then the bound $\RelPosCorrThresh$ applies to the total number of cells that are corrupted, i.e., contain no honest party.
    \item However, choosing a larger number of rows $k_1$ means that the bound has to apply to every individual row, and not just to all $k_1 \cdot k_2$ cells globally.
\end{myitemize}
The result is that the number of cells $k_1 \cdot k_2$ is maximized for $k_1 = 1$. This is relevant for operations for which the cost is inversely proportional to $k_1 \cdot k_2$, such as $\iGet$.

To choose a particular number for $k_1$ and $k_2$ with the goal of maximizing the efficiency of the protocol, i.e., the total message complexity, we look at the complexity of $\iGet$ and $\iStore$ for storing data as well as the complexity of $\iJoin$, which can be seen as the cost of maintaining the networking topology.

\begin{myitemize}
    \item $\iGet$ and $\iStore$ are always minimized by choosing $k_1 = 1$ and maximizing $k_2$. This is a result of the above observation that this maximizes the number of cells $k_1 \cdot k_2$.
    \item $\iJoin$ complexity has a minimum where $k_1 \approx k_2$, representing the point where the cost of joining rows and columns is balanced.
\end{myitemize}

So, depending on the relative complexity of the two operations, we can choose $k_1$ and $k_2$ to minimize the total message complexity. In applications where storing and retrieving data is very common and therefore dominates complexity, we probably want to choose $k_1$ to be small while maximizing $k_2$. However, if $\iJoin$ is expensive (for example, because there are a lot of nodes, or joining and leaving is frequent), we should choose a parameter set closer to the minimum of $\iJoin$ complexity, i.e., $k_1 \approx k_2$.

\subsection{Results from Simulation}
We complement our analytical results with simulations of our protocol (excluding extensions).

\superparagraph{Simulation Setup.}
We fix a specific join-leave schedule and values for the parameters $k_1$ and $k_2$, and simulate the assignment of honest parties to cells. At each time step $\tau$, we compute the number of connections a party $\party$ has, denoted by $|\mathsf{Peers}|_\party^\tau$, assuming that a party is immediately connected to all parties in the same row and column.
We also evaluate the fraction of corrupted symbols for each party, denoted as $|\Corrupted_\party^\tau| / \FileLength$, where a symbol in column $c$ is considered corrupted for a party in row $r$ if the cell $(r,c)$ contains no honest party.
This simulation is intentionally simplified: it does not account for timing constraints or overlap times. Nonetheless, it captures the key characteristics of our protocol. Importantly, due to the protocol's design, we do not need to model the behavior of dishonest parties.
For our plots, we fix the number of columns to $k_2 = 100$ and vary the number of rows $k_1$. The simulated join-leave schedule begins with $20$ initial parties. Then follows a warm-up phase in which $2{,}480$ additional parties join, one per round, until a total of $2{,}500$ parties are active. After this phase, we introduce churn: in each subsequent step, $50$ parties leave (in FIFO order) and $50$ new parties join. Consequently, each party remains in the system for exactly $2{,}500 / 50 = 50$ time steps.

\superparagraph{Results.}
Our simulation results are presented in \cref{fig:simulationplots}, where we plot $\max_\party |\Corrupted_\party^\tau| / \FileLength$ and $\max_\party |\mathsf{Peers}_\party^\tau|$ as functions of time $\tau$.
One key observation is that after the warm-up phase, the fraction of corrupted symbols stabilizes. Notably, errors do not accumulate over time.
We further observe that the fraction of corrupted symbols decreases as the number of rows $k_1$ is reduced. This is expected, as fewer rows imply fewer total cells, increasing the probability that a cell contains an honest party. In contrast, the number of peers (and thus bandwidth usage) increases with fewer rows. Hence, our design demonstrates a trade-off between robustness (fewer corrupted symbols) and bandwidth (number of peer connections).

\ifnum\ccs=0
\newcommand{\figWidth}{10cm}
\ifnum\ccs=1
    \newcommand{\figHeight}{4cm}
\else
    \newcommand{\figHeight}{6cm}
\fi
\begin{figure*}[htbp]  
    \centering
    \begin{minipage}{0.48\textwidth}
        \resizebox{\textwidth}{!}{
        \begin{tikzpicture}
            \begin{axis}[
                width=\figWidth,
                height=\figHeight,
                xlabel={Time Step $\tau$},
                ylabel={$\max_\party |\Corrupted_\party^\tau| / \FileLength$},
                title={Corrupted symbols over time},
                grid=major,
                legend pos=north east,
                mark repeat=30
            ]
            \addplot table [x=Time_Step, y=Corruption_Rows_1, col sep=comma] {csvdata/simulation_data.csv}; \addlegendentry{$k_1=1$};
            \addplot table [x=Time_Step, y=Corruption_Rows_5, col sep=comma] {csvdata/simulation_data.csv}; \addlegendentry{$k_1=5$};
            \addplot table [x=Time_Step, y=Corruption_Rows_10, col sep=comma] {csvdata/simulation_data.csv}; \addlegendentry{$k_1=10$};
            \addplot table [x=Time_Step, y=Corruption_Rows_25, col sep=comma] {csvdata/simulation_data.csv}; \addlegendentry{$k_1=25$};
            \end{axis}
        \end{tikzpicture}
        }
    \end{minipage}
    \quad
    \begin{minipage}{0.48\textwidth}
        \resizebox{\textwidth}{!}{
        \begin{tikzpicture}
            \begin{axis}[
                width=\figWidth,
                height=\figHeight,
                xlabel={Time Step $\tau$},
                ylabel={$\max_\party |\mathsf{Peers}_\party^\tau|$},
                title={Number of peers over time},
                grid=major,
                legend pos=north east,
                mark repeat=30
            ]
            \addplot table [x=Time_Step, y=Connections_Rows_1, col sep=comma] {csvdata/simulation_data.csv}; \addlegendentry{$k_1=1$};
            \addplot table [x=Time_Step, y=Connections_Rows_5, col sep=comma] {csvdata/simulation_data.csv}; \addlegendentry{$k_1=5$};
            \addplot table [x=Time_Step, y=Connections_Rows_10, col sep=comma] {csvdata/simulation_data.csv}; \addlegendentry{$k_1=10$};
            \addplot table [x=Time_Step, y=Connections_Rows_25, col sep=comma] {csvdata/simulation_data.csv}; \addlegendentry{$k_1=25$};
            \end{axis}
        \end{tikzpicture}
        }
    \end{minipage}
    \caption{Simulation results for $k_2 = 100$ columns and different number of rows $k_1$. We simulated a join-leave schedule that leads to $2{,}500$ honest parties (after a warmup phase) and each party staying for $50$ time steps.}
    \label{fig:simulationplots}
\end{figure*}
\fi


\clearpage
\bibliographystyle{ACM-Reference-Format}
\bibliography{export/abbrev3,export/crypto,add}

\appendix
\section*{Appendices}
\section{More Related Work}
\label{sec:relwork}
In this section, we provide a more extensive overview of related work.
To summarize, our work is the first to rigorously address the networking aspect of DAS.
Our robust distributed arrays share similarities with distributed hash tables but differ significantly in their design goals. Namely, in contrast to solutions based on distributed hash tables, our construction achieves simplicity, low query latency, and provable security even under dishonest majority.

\subsection{Data Availability Sampling}
\label{sec:relwork:das}
Our primary motivation for introducing and studying robust distributed arrays is \emph{data availability sampling} (DAS), first introduced by Al-Bassam et al.~\cite{FC:ASBK21}.
To provide context for our definition, we review the fundamentals of DAS in \cref{sec:background:das}.
Here, we discuss related work in the context of DAS.

\superparagraph{Cryptography for DAS.}
The cryptographic and encoding aspects of DAS were first formally defined as a cryptographic primitive in~\cite{cryptoeprint:2023/1079}.
Since then, these cryptographic aspects have been explored further~\cite{C:HalSimWag24,cryptoeprint:2024/1362,cryptoeprint:2025/034,evans2025accidental}.
Notably, all of these works assume an abstract, black-box networking layer.
Prior to our work, there was no formal definition of the networking requirements or a construction that provably satisfies them.

\superparagraph{Networking for DAS.}
While there has been no rigorous formalization of networking in the context of DAS, several recent studies have explored potential approaches and identified key challenges.
Kr{\'o}l et al.~\cite{DASp2p} have (informally) identified functionality and efficiency requirements for a DAS networking layer, such as permissionlessness and low-latency storage requests.
They have also discussed why classical networking solutions, including distributed hash tables (DHTs) like Kademlia~\cite{Kademlia}, are not directly applicable to DAS without significant modifications.
Cortes-Goicoechea et al.~\cite{kademliaDASlimits} specifically analyze the suitability of Kademlia for DAS in terms of efficiency, though not security.
Using a custom DHT simulator, they evaluate Kademlia's performance and conclude that while query efficiency is reasonable, the workload associated with storing the data is prohibitive for DAS applications.
A notable proposal for DAS networking is PANDAS~\cite{ethresearPANDASPractical}.
This proposal assumes that each node maintains an (almost) complete list of all other nodes through continuous background discovery.
Nodes are deterministically assigned to specific chunks of encoded data using a hash function. In combination, this ensures that data can be retrieved in a single hop with minimal latency.
Conceptually, PANDAS can be viewed as a special case of our proposed solution with exactly one row.
However, there are differences: PANDAS assumes that node-to-chunk mappings change with each slot, whereas we propose a stable assignment where nodes remain within their designated cells until they leave the network.
Most importantly, however, PANDAS lacks a formal security definition or proof.

\subsection{Distributed Hash Tables}
\label{sec:relwork:dhts}
In many ways, distributed arrays bear similarities to distributed hash tables (DHTs).
Therefore, we discuss DHTs and how they relate to distributed arrays and our design goals.

\superparagraph{Distributed Hash Tables (DHTs).}
A \emph{hash table} is a classical data structure that implements a \emph{key-value store}, dynamically and efficiently maintaining a partial mapping $T\colon \mathcal{U} \to \mathcal{V}$ between a large universe of keys $\mathcal{U}$ and a domain of values $\mathcal{V}$.
It supports an \emph{insert} operation, which assigns a value $v \in \mathcal{V}$ to a key $k \in \mathcal{U}$ (i.e., $T(k) := v$), and a \emph{query} operation, which retrieves $T(k) \in \mathcal{V} \cup \{\bot\}$ given a key $k \in \mathcal{U}$.
A \emph{distributed hash table} extends this concept by distributing key-value pairs across multiple nodes, ensuring efficient storage and retrieval by contacting only a few nodes.
DHTs function effectively even when nodes dynamically join or leave and operate with only a local view of the network topology.
Prominent examples of DHTs include CAN~\cite{candht}, Pastry~\cite{Pastry}, Chord~\cite{Chord}, and Kademlia~\cite{Kademlia}, which leverage ideas from consistent hashing~\cite{STOC:KLLPLL97}.

\superparagraph{Flat Key Space vs. Indices.}
A key distinction between DHTs and our distributed arrays lies in the structure of the key space.
In a DHT, key-value pairs can have arbitrary keys from a large universe $\mathcal{U}$, with no structural constraints (e.g., $\mathcal{U} = \bool^{256}$).
The actual set of inserted keys forms only a small subset of $\mathcal{U}$.
By contrast, in our distributed arrays, parties query indices $i \in [\FileLength]$.
That is, keys are always small integers with a predefined structure.

\superparagraph{Orthogonal Design Goals.}
While DHTs could theoretically implement distributed arrays, this approach is unsuitable for our context.
First, classical DHTs~\cite{candht,Pastry,Chord,Kademlia} lack provable robustness against adversarial behavior (though some robust variants exist, see below).
Second, DHTs prioritize minimizing the number of stored peers, at the cost of (poly-)logarithmic query latencies.
This is because when these solutions have been developed over two decades ago, memory was expensive, as already discussed in \cite{DASp2p}.
In contrast, in our modern setting we can assume each node having sufficient memory to store \emph{almost all} other peers, though without requiring full peer storage to avoid high bandwidth costs.
Given the importance of minimizing query latency in DAS, our design goals fundamentally differ from those of DHTs.
Although recent work~\cite{FC:ZhaVen23} optimizes Kademlia's concrete latency, it still lacks formal security analysis.

\superparagraph{Kelips DHT.}
An exception to the typical $\Omega(\mathsf{polylog}(N))$ query latency in DHTs is Kelips~\cite{kelips}, which achieves $\Theta(1)$ query latency at the expense of increased memory usage and background communication.
This design choice aligns closely with our goals.
Moreover, both our construction and Kelips employ a similar conceptual framework:
Kelips randomly partitions nodes into \emph{affinity groups} -- akin to our construction's rows -- with each node maintaining an (almost complete) list of peers within its affinity group and partial connections to other groups.
In this way, Kelips enables queries in \emph{two} hops, whereas our approach achieves \emph{one}-hop queries by leveraging additional connections along columns (cf. \cref{fig:tikz_grid}).
Furthermore, unlike Kelips, our construction provides a formal security proof and rigorous definition.

\superparagraph{Adversarially Robust DHTs.}
To add robustness in presence of byzantine adversaries, Baumgart and Mies introduced S/Kademlia~\cite{skademlia}, a variant of the well-known Kademlia DHT.
Their work systematically analyzes multiple attack vectors against Kademlia and integrates a set of ad-hoc countermeasures into the protocol.
While these modifications improve practical robustness, they do not provide formal guarantees of security.

A significant line of research aims to establish a more rigorous understanding of security in DHTs, ultimately converging on \emph{quorum-based} solutions.
Castro et al.~\cite{DBLP:conf/osdi/CastroDGRW02} have observed in 2002 that while network overlays such as DHTs offer resilience against \emph{random crash failures}, they lack guarantees against coordinated malicious failures. Their robustness definition informally asserts that any honest query for a specific key should reach all honest nodes responsible for storing that key.
Around the same time, Fiat and Saia~\cite{SODA:FiaSai02} proposed a content-addressable network resilient to \emph{malicious node deletions}.
Their approach groups nodes into \emph{supernodes} (or \emph{quorums}), which collectively act as vertices in a butterfly network, facilitating message routing through this structure.
A simpler construction by Naor and Wieder~\cite{DBLP:conf/iptps/NaorW03} also provides robustness against random crash failures and makes use of overlapping quorums induced by a continuous graph.

Beyond random crash failures, Fiat, Saia, and Young introduced the \emph{byzantine join adversary} (also known as \emph{adaptive join-leave attacks}) and adapted Chord to be resilient against it~\cite{DBLP:conf/esa/FiatSY05}. This adversary can adaptively insert malicious nodes into the network, allowing them to choose their identifiers and thus their positions within the DHT topology.
Fiat, Saia, and Young's approach replaces each Chord node with a quorum (they call it a \emph{swarm}) of logarithmically many nodes.
As long as each quorum maintains a (strong) \emph{honest majority}, the protocol remains secure.
Several subsequent works have extended this approach~\cite{DBLP:conf/icalp/AwerbuchS04,STOC:Scheideler05,DBLP:conf/spaa/AwerbuchS06,DBLP:conf/opodis/AwerbuchS06,DBLP:conf/iptps/AwerbuchS07,DBLP:journals/sigops/SenF12,DBLP:journals/ton/YoungKGK13,DBLP:conf/podc/GuerraouiHK13} and studied how this local honest majority condition can be achieved.
However, such quorum-based techniques exhibit three key limitations making them unsuitable for use in DAS: (1) achieving a (strong) honest majority within quorums necessitates a (strong) honest majority among all nodes; (2) quorum management requires complex sub-protocols, such as distributed coin flips~\cite{DBLP:conf/opodis/AwerbuchS06} or distributed key generation~\cite{DBLP:journals/ton/YoungKGK13}; (3) similar to classical DHTs, these protocols have (poly-)logarithmic latency, and they introduce additional polylogarithmic overheads, particularly in terms of message complexity.
Jaiyeola et al.~\cite{DBLP:conf/ipps/JaiyeolaPSYZ18} proposed a method that reduces quorum sizes (and consequently overhead) to $\Theta(\log \log N)$.
To achieve this, they need to limit the adversary's join rate and relax the robustness requirements to hold only for a large fraction of the nodes.

\subsection{Further Related Work}
\label{sec:relwork:further}
Here, we discuss further related work.

\superparagraph{Distributed Data Structures.}
We have already discussed distributed hash tables as a fundamental example of a distributed data structure.
There is a rich literature on distributed data structures (e.g., \cite{distsplaytree,distlinkedlist,distsuffarray}). We refer the reader to these works for further references.
In the following, we highlight two examples that share similarities with our robust distributed arrays:

Skip graphs, introduced by Aspnes and Shah~\cite{SODA:AspSha03}, serve as a distributed counterpart to balanced search trees, offering efficient range and successor queries -- capabilities that traditional distributed hash tables lack.
Moreover, their construction is resilient to a substantial fraction of nodes crashing.
Over time, skip graphs have been extensively analyzed and extended~\cite{SODA:GooNelSun06,DBLP:journals/comcom/BeltranMS08,PODC:JRSST09,DBLP:journals/dc/AspnesW09,DBLP:conf/icdcs/HuqG17}.

To the best of our knowledge, only one prior work explicitly considers distributed arrays~\cite{distarray}.
Similar to skip graphs, their motivation is the observation that distributed hash tables do not maintain locality in a logical array structure.
Specifically, if an array $a[0], \dots, a[n-1]$ is stored in a DHT using key-value pairs $(i, a[i])$, then consecutive elements $a[i]$ and $a[i+1]$ are unlikely to be stored near each other. To mitigate this, their approach modifies Chord~\cite{Chord} by employing a reverse-bit order mapping. As a result, this enables more efficient sequential access, sorted array searches, and more.
Compared to our work, the key distinctions are: (1) their approach does not define or guarantee any robustness, and (2) our focus is not primarily on locality features such as efficiency of sequential access. Still, our constructions supports some limited form of locality: symbols stored within the same column are stored on the same node.

\superparagraph{Random Node Sampling.}
One potential approach to constructing networking protocols for DAS is the use of peer discovery protocols, a method widely considered by practitioners~\cite{ethresearPeerDASSimpler,ethresearNumberPeers}.
In this approach, each node employs a peer discovery protocol to establish connections with a sufficiently large number, $L$, of peers.
When a node needs to query a specific position, it relies on the assumption that at least one of these $L$ peers both stores the required position and behaves honestly.
Existing analyses of this approach~\cite{ethresearNumberPeers} assume that peers are sampled uniformly at random and independently from the overall peer set. Under this assumption, if there are sufficiently many honest nodes and $L$ is large enough, most positions will be adequately covered.
In essence, these analyses implicitly assume that the peer discovery protocol securely implements a \emph{random node sampling functionality}. However, peer discovery protocols used in practice do not necessarily fulfill this requirement.

There is existing literature on byzantine-resistant random node sampling~\cite{brahms,DBLP:conf/opodis/AnceaumeBG10,basalt,honeybee}.
However, some of these works lack formal analysis, while for others, it remains unclear whether they truly instantiate the required functionality -- particularly regarding whether the sampled nodes are independent.

\clearpage
\section{Mathematical Tools}
\label{sec:appendix:mathtools}

\begin{lemma}[Second Moment of Binomial Distribution, Any Textbook]
    Let $X = \sum_{i=1}^n X_i$ be the sum of independent random variables $X_i \in \bool$, where $\Prob{X_i=1} = p$ for all $i \in [n]$. Then, we have $\Expect{X^2} = np(1-p) + n^2 p^2$.
    \label{lemma:secondmomentofbinom}
\end{lemma}
\begin{lemma}[Chernoff Bound, taken from \cite{DBLP:books/daglib/0012859}]
    Let $X_1,\dots,X_\ell$ be independent Poisson trials. Let $X = \sum_{i \in [\ell]} X_i$ and $\mu = \Expect{X}$.
    Then, for $0 < \delta \leq 1$, we have \[
        \Prob{X \geq (1+\delta)\mu} \leq e^{-\mu\delta^2/3}.
    \]
    \label{lemma:chernoff:mitzenmacher}
\end{lemma}
\begin{corollary}
    Let $X,Y$ be random variables, not necessarily independent. Assume that $X$ follows a binomial distribution with $n$ trials and parameter $p$. Then, we have \[
        \Expect{XY} \leq \Expect{Y}(n e^{-np/3} + 2 np).
    \]
    In particular, if $n \geq -3\ln(p)/p$, then  \[
        \Expect{XY} \leq 3np \cdot \Expect{Y}.
    \]
    \label{cor:expectationofproduct}
\end{corollary}
\begin{proof}
    This follows directly by using conditional expectation and applying \cref{lemma:chernoff:mitzenmacher} to random variable $X$ with $\delta = 1$ and $\mu = np$.
\end{proof}
\noindent
The following two facts are standard:
\begin{lemma}[Exponential Decay]\label{lemma:exponential_decay}
We have $(1-x)^N < e^{-xN}$ for all real-valued $1>x\geq 0$ and $N>0$.
\end{lemma}
\begin{proof}
Taking $\ln$ on both sides, this is equivalent to $N\ln(1-x) < -xN$, which follows from the fact that $\ln$ is convex.
\end{proof}
\begin{lemma}[Binomial Coefficients and Entropy]\label{lemma:binomial_entropy}
For any $n >0$ and $0\leq k \leq n$, we have $\binom{n}{k} \leq 2^{\binaryEntropy(k/n)n}$, where
$\binaryEntropy(p) = -p\log_2 p - (1-p)\log_2(1-p)$ is the binary entropy function.
\end{lemma}
\begin{proof}
Consider the probability space, where we unformly choose exactly $k$ different items out of $\{1,\ldots,n\}$. Let $X_i$ be the indicator variable indicating whether $i$ was chosen.
Denote the Shannon-entropy of (the distribution of) a random variable X by $\Shannon(X)$. By definition, $\Shannon(X_i) = \binaryEntropy(k/n)$.
Then using subadditivity of entropy, we get
\[
    \log_2\left(\binom{n}{k}\right) = \Shannon(X_1,\ldots,X_n) \leq \Shannon(X_1) + \ldots + \Shannon(X_n) = n\cdot \binaryEntropy(k/n) \enspace,
\]
which shows the claim.
\end{proof}




\clearpage
\section{Postponed Experiment Definition}
\label{sec:appendix:ccs:experiment}

\clearpage

\section{Subnet Discovery Protocol}
\label{sec:appendix:simplesubnetprotocol}
In this section, we present a simple protocol that satisfies \cref{def:subnetdiscovery,def:subnetdiscovery:robustness}.
That is, we give a subnet discovery protocol $\subnetprot$ for $\numberofsubnets$ subnets.
The protocol is specified in \cref{fig:subnetprotocol}.

\superparagraph{Protocol Overview.}
We give an informal overview of the subnet protocol. For details, we refer to \cref{fig:subnetprotocol}.
The main idea is simple: if a party $\party$ wants to join a subnet $\subnetid$ via a party $\hat{\party}$, it first announces itself to that party.
Then, $\hat{\party}$ forwards this information to all parties in the subnet that it knows.
In this way, if $\party$ joins at time $\tau$, all parties that $\hat{\party}$ knows will learn about $\party$ at latest at the beginning of round $\tau+3$.
Further, we need to establish that $\party$ learns from $\hat{\party}$ about all other parties in the subnet.
In a naive approach, $\hat{\party}$ would immediately respond with all peers in the subnet that it knows.
However, this approach fails if multiple parties join simultaneously (either all via $\hat{\party}$ or via different parties).
To solve this, $\party$ must give $\hat{\party}$ enough time to learn about all parties that joined around time $\tau$.
It turns out that waiting until $\tau+4$ before pulling the list of peers from $\hat{\party}$ is enough.

\begin{figure*}[hbt!]
    \centering
    \begin{tikzpicture}

        \def\timelineHeight{0.4}
        \def\timeslotWidth{1.2}
        \def\nTimeslotsMinusOne{7}
        \def\nTimeslots{8}

        \def\yP{3}
        \def\yPhat{2.3}
        \def\yPhatPrime{1.6}
        \def\yPPrime{0.9}

        \foreach \y/\label in {\yP/$\party$, \yPhat/$\hat{\party}$, \yPhatPrime/$\hat{\party}'$, \yPPrime/$\party'$} {
            \draw[thick, fill=gray!20] (0,\y) rectangle (\nTimeslots*\timeslotWidth, \y+\timelineHeight);
            \node[left] at (-0.2, \y+0.2) {\label};
        }

        \foreach \x in {0,1,...,\nTimeslotsMinusOne} {
            \draw[dashed] (\x*\timeslotWidth, \yP+\timelineHeight) -- (\x*\timeslotWidth, \yPPrime);
            \node[above,font=\scriptsize] at (\x*\timeslotWidth+0.6, \yP+0.7) {$\tau_{\mathsf{join}}+\x$};
        }

        \draw[->, thick] (0.6*\timeslotWidth, \yP+0.2)
            to[out=-40, in=120] (1.6*\timeslotWidth, \yPhat+0.2);

        \draw[->, thick] (1.68*\timeslotWidth, \yPhat+0.2)
            to (2.19*\timeslotWidth, \yPhat-0.14);
        \draw[->, thick] (1.68*\timeslotWidth, \yPhat+0.2)
            to (2*\timeslotWidth, \yPhat-0.2);
        \draw[->, thick] (1.68*\timeslotWidth, \yPhat+0.2)
            to (1.75*\timeslotWidth, \yPhat-0.19);

        \draw[->, thick] (4.1*\timeslotWidth, \yP+0.2)
            to (5.3*\timeslotWidth, \yPhat+0.2);
        \draw[->, thick] (5.6*\timeslotWidth, \yPhat+0.2)
            to (6.6*\timeslotWidth, \yP+0.2);

        \draw[->, thick] (2.6*\timeslotWidth, \yPPrime+0.2)
            to[out=40, in=-120] (3.6*\timeslotWidth, \yPhatPrime+0.2);

        \draw[->, thick] (3.68*\timeslotWidth, \yPhatPrime+0.2)
            to (4.19*\timeslotWidth, \yPhatPrime+0.34);
        \draw[->, thick] (3.68*\timeslotWidth, \yPhatPrime+0.2)
            to (4*\timeslotWidth, \yPhatPrime+0.21);
        \draw[->, thick] (3.68*\timeslotWidth, \yPhatPrime+0.2)
            to[out=20, in=-140] (4.6*\timeslotWidth, \yPhat+0.2); 

        \draw[->, thick] (6.6*\timeslotWidth, \yPPrime+0.2)
            to (7.6*\timeslotWidth, \yPhatPrime+0.2);

        \draw[->, thick] (7.68*\timeslotWidth, \yPhatPrime+0.2)
            to (8.2*\timeslotWidth, \yPPrime+0.2);
    \end{tikzpicture}
    \caption{Visualization of joining a subnet in our simple subnet discovery protocol $\subnetprot$. Party $\party$ joins the subnet at time $\tau_{\mathsf{join}}$ via party $\hat{\party}$. It sends a message to $\hat{\party}$, and then $\hat{\party}$ announces to other parties in the subnet that $\party$ joined. At $\tau_{\mathsf{join}}+4$, party $\party$ pulls all peers in the subnet from $\hat{\party}$. By that time, $\hat{\party}$ will have learned about $\party'$ from $\hat{\party}'$, assuming that $\hat{\party}$ was one of the peers of $\hat{\party}'$ at time $\tau_{\mathsf{join}}+3$. $\party'$ joined the subnet via $\hat{\party}'$ at time $\tau_{\mathsf{join}}+2$.}
    \label{fig:tikz_grid_subnprot}
\end{figure*}

\begin{figure*}[!htb]
    \resizebox{\textwidth}{!}{
    \centering
    \nicoresetlinenr
    \noindent\fbox{\parbox{1.1\textwidth}{
        \centering
        \begin{minipage}[t]{0.5\textwidth}%
            \underline{\textbf{Interface} $\iInit(\mathbf{P}, \AuxRole_1,\ldots,\AuxRole_{\ell})$}
            \begin{nicodemus}
                \smallskip
                \item[]\textit{// nothing to do here}
            \end{nicodemus}
            \medskip\noindent
            \underline{\textbf{Interface} $\iJoin(\{\party_1,\ldots,\party_t\},\AuxRole)$}
            \begin{nicodemus}
                \smallskip
                \item[]\textit{// initialize all subnets locally}
                \item $\pcfor \subnetid \in [\numberofsubnets]\colon~\subnetmap[\subnetid] := \emptyset$
            \end{nicodemus}
            \medskip\noindent
            \underline{\textbf{Interface} $\iCreateSubnet(\subnetid, \mathbf{P})$}
            \begin{nicodemus}
                \item $\subnetmap[\subnetid] \coloneqq  \subnetmap[\subnetid] \union \mathbf{P}$
            \end{nicodemus}
            \medskip\noindent
            \underline{\textbf{Interface} $\iJoinSubnet(\subnetid, \hat{\party})$}
            \begin{nicodemus}
                \item $\subnetmap[\subnetid] \coloneqq \subnetmap[\subnetid]  \union \{\self,\hat{\party}\}$
                \item \textbf{send} $(\msgJoinSubnet,\subnetid)$ \textbf{to} $\hat{\party}$ (\textbf{say, in round $\tau$})
                \item \textbf{in round $\tau+4$: send} $(\msgJoinSubnetPull,\subnetid)$ \textbf{to} $\hat{\party}$
                \item \textbf{receive} $(\msgJoinSubnetPullRsp,\subnetid,\mathbf{P})$ \textbf{from} $\hat{\party}$
                \item $\subnetmap[\subnetid] \coloneqq  \subnetmap[\subnetid] \union \mathbf{P}$
            \end{nicodemus}
        \end{minipage}\qquad
        \begin{minipage}[t]{0.5\textwidth}%
            \underline{\textbf{On Message} $(\msgJoinSubnet,\subnetid)$ \textbf{from} $\party$}
            \begin{nicodemus}
                \item $\mathbf{P} \coloneqq \subnetmap[\subnetid]$
                \item $\subnetmap[\subnetid] \coloneqq \subnetmap[\subnetid]  \union \{\party\}$
                \smallskip
                \item[]\textit{// forward the party (in parallel)}
                \item $\pcfor \party'' \in \mathbf{P}\colon$ \textbf{send} $(\msgJoinSubnetFwd,\subnetid,\party)$  \textbf{to} $\party''$
            \end{nicodemus}
            \medskip\noindent
            \underline{\textbf{On Message} $(\msgJoinSubnetPull,\subnetid)$ \textbf{from} $\party$}
            \begin{nicodemus}
                \item $\mathbf{P} \coloneqq \subnetmap[\subnetid]$
                \smallskip
                \item[]\textit{// respond to the party}
                \item \textbf{send} $(\msgJoinSubnetPullRsp,\subnetid,\mathbf{P})$ \textbf{to} $\party$
            \end{nicodemus}
            \medskip\noindent
            \underline{\textbf{On Message} $(\msgJoinSubnetFwd,\subnetid,\party)$ \textbf{from} $\hat{\party}$}
            \begin{nicodemus}
                \item $\subnetmap[\subnetid]  \coloneqq \subnetmap[\subnetid] \union \{\party\}$
            \end{nicodemus}
            \medskip\noindent
            \underline{\textbf{Interface} $\iGetPeers(\subnetid)$}
            \begin{nicodemus}
                \item $\pcreturn \subnetmap[\subnetid]$
            \end{nicodemus}
        \end{minipage}
    }}}
    \caption{A subnet discovery protocol $\subnetprot$ for $\numberofsubnets$ subnets. The party that is executing the instructions is referred to as $\self$.}
    \label{fig:subnetprotocol}
\end{figure*}

\superparagraph{Robustness Analysis.}
In the next definition, we introduce a shorthand notation to denote that a party $\party$ knows another party $\party'$ as part of a subnet $\subnetid$ at time $\tau$.
Note that the notation represents an event in the experiment of running the protocol.
\begin{definition}[Peer Notation]
    Consider $\subnetprot$ as defined in \cref{fig:subnetprotocol}.
    Let $\party,\party'$ be parties, let $\tau \geq 0$ be a point in time, and let $\subnetid \in [\numberofsubnets]$.
    Then, we write $\partyknows{\party}{\party'}{\tau}{\subnetid}$ to denote the event that at time $\tau$, $\party$ is active and honest, and the set $\subnetmap[\subnetid]$ locally stored at party $\party$ contains\footnote{We mean it is contained in the set over the \emph{entire} round.} $\party'$.
    \label{def:simplesubnetprot:peernotation}
\end{definition}
With this notation, we can now make several observations.
The first observation is that if a party knows another party at some point, it will also know this party in the future.
This holds because parties never remove elements from sets $\subnetmap[\subnetid]$.
The second observation also follows trivially from our definition of the protocol and the definition of the event.
It relates the creation of new subnets and our peer notation.
The third trivial observation that we make is how our notation relates to interface $\iGetPeers$.
Our fourth observation is about the mechanics of joining a subnet. It can easily be seen from the definition of the protocol.
\begin{proposition}[Peer Preservation]
    Let $\party,\party'$ be honest parties, let $\tau \geq 0$ be a point in time, and let $\subnetid \in [\numberofsubnets]$.
    Conditioned on the event that at time $\tau$ and $\tau+1$, $\party$ is active, we have (with probability $1$):
    \[
        {\partyknows{\party}{\party'}{\tau}{\subnetid} \implies \partyknows{\party}{\party'}{\tau+1}{\subnetid}}.
    \]
    \label{prop:simplesubnetprot:peerpreservation}
\end{proposition}
\begin{proposition}[Subnet Creation]
    Let $\party,\party'$ be honest parties active at time $\tau = 0$ and $\mathbf{P}$ be a set with $\party, \party' \in \mathbf{P}$. Let $\subnetid \in [\numberofsubnets]$.
    Then, conditioned on event $\EventCreatedSubnet{\subnetid}{\mathbf{P}}$, we have (with probability $1$): \[
        \partyknows{\party}{\party'}{0}{\subnetid} \land \partyknows{\party'}{\party}{0}{\subnetid}.
    \]
    \label{prop:simplesubnetprot:initial}
\end{proposition}
\begin{proposition}[Get Peers]
    Let $\party,\party'$ be honest parties, let $\tau \geq 0$ be a point in time at which $\party$ is active, and let $\subnetid \in [\numberofsubnets]$. Then, with probability $1$, we have:
    \[
        \left(\partyknows{\party}{\party'}{\tau}{\subnetid} \land \EventCalledGetPeers{\subnetid}{\party}{\tau}\right) \implies \EventGotPeer{\subnetid}{\party}{\party'}{\tau}.
    \]
    \label{prop:simplesubnetprot:callgetpeers}
\end{proposition}
\begin{proposition}[Join Timings]
    Let $\party,\party',\hat{\party}$ be honest parties, let ${\tau} \geq 0$ be a point in time, and let $\subnetid \in [\numberofsubnets]$.
    Condition on $\party$ calling $\iJoinSubnet(\subnetid, \hat{\party})$ at time ${\tau}$ and $\hat{\party}$ being active at time $\tau$.
    Then, with probability $1$, we have the following two:
    \begin{align*}
        \EventActive{\party}{\tau+7} \land \partyknows{\hat{\party}}{\party'}{{\tau}+5}{\subnetid} &\implies \partyknows{{\party}}{\party'}{{\tau}+7}{\subnetid}, \text{ and } \\
        \EventActiveDuration{\party'}{(\tau+2)}{(\tau+3)} \land \partyknows{\hat{\party}}{\party'}{{\tau}+1}{\subnetid} &\implies \partyknows{{\party}'}{\party}{{\tau}+3}{\subnetid}.
    \end{align*}
    \label{proposition:simplesubnetprot:joinmechanics}
\end{proposition}
We will now come to our central lemma, which we then use to show robustness. Before giving the lemma, we introduce some useful terminology, namely, what it means that a node joins a subnet properly. This is exactly as in the definition of event $\EventIsInSubnetNoArgs$.
\begin{definition}[Joining Subnets Properly]
    Let $\SubnetDelay = 7$. Let $\subnetid \in [\numberofsubnets]$.
    Let $\party$ be an honest party and $\tau \geq 0$ be a point in time.
    Then, we say that \emph{party $\party$ joins the subnet $\subnetid$ properly at time $\tau$} if (a) $\tau = 0$ and $\EventCreatedSubnet{\subnetid}{\mathbf{P}}$ such that $\party \in \mathbf{P}$ or (b) $\party$ calls $\iJoinSubnet(\subnetid, \hat{\party})$ at time $\tau$ for an honest party $\hat{\party}$ with $\EventStaysInSubnetDelay{\subnetid}{\hat{\party}}{\tau}{(\tau+\SubnetDelay)}$. In case (b) we also say that $\party$ joins properly \emph{via} $\hat{\party}$.
\end{definition}
\begin{lemma}[Peer Propagation]
    Let $\party,\party'$ be honest parties and let $\subnetid \in [\numberofsubnets]$. Assume that \begin{itemize}
        \item No honest party calls interface $\iJoinSubnet$ with this $\subnetid$ in rounds $0$ to $7$ (inclusive), and
        \item Party $\party$ joins the subnet $\subnetid$ properly at time $\tau_{\mathsf{join}}$, and
        \item party $\party'$ joins the subnet $\subnetid$ properly at time $\tau_{\mathsf{join}}'$ and is active at a point $\tau \geq \tau_{\mathsf{join}}'$.
    \end{itemize}
    Further, assume that at least one of the following three conditions hold:
    \begin{enumerate}[label=(\alph*)]
        \item\label{item:simplesubnetprot:centrallemma:conda} $\tau \geq \tau_{\mathsf{join}}+3 \land \tau_{\mathsf{join}}' \leq \tau_{\mathsf{join}}-3$, or
        \item\label{item:simplesubnetprot:centrallemma:condb} $\tau \geq \tau_{\mathsf{join}}+3 \land \tau_{\mathsf{join}}' \leq \tau - 7$, or
        \item\label{item:simplesubnetprot:centrallemma:condc} $\tau_{\mathsf{join}} = \tau_{\mathsf{join}}' = 0$.
    \end{enumerate}
    Then, we have $\partyknows{{\party}'}{\party}{{\tau}}{\subnetid}$.
    \label{lemma:simplesubnetprot:centrallemma}
\end{lemma}
\begin{proof}
    Before we start with the actual proof, we make the following useful observation, which follows from the assumption that no honest party joins in rounds $1$ to $7$:
    If a party $\party$ joins the subnet $\subnetid$ properly at time $\tau$, then \begin{myitemize}
        \item Either, $\party$ is one of the initial parties, i.e., $\tau = 0$ and event $\EventCreatedSubnet{\subnetid}{\mathbf{P}}$ occurs such that $\party \in \mathbf{P}$, or
        \item $\party$ calls $\iJoinSubnet(\subnetid, \hat{\party})$ at time $\tau$ for an honest party $\hat{\party}$ with $\EventStaysInSubnetDelay{\subnetid}{\hat{\party}}{\tau}{(\tau+\SubnetDelay)}$, \emph{where $\hat{\party}$ joined the subnet at some time $\leq \tau - 7$.}
    \end{myitemize}
    With this in mind, we will now prove the statement via induction on pairs $(\tau,\tau_{\mathsf{join}})$, with lexicographic ordering. Essentially, \ref{item:simplesubnetprot:centrallemma:condc} serves as our base case, and for the induction step we consider cases \ref{item:simplesubnetprot:centrallemma:conda} and \ref{item:simplesubnetprot:centrallemma:condb} separately.

    \medskip\noindent\emph{Assuming \ref{item:simplesubnetprot:centrallemma:condc} holds:} If $\tau_{\mathsf{join}} = \tau_{\mathsf{join}}' = 0$, then we know that event $\EventCreatedSubnet{\subnetid}{\mathbf{P}}$ such that $\party,\party' \in \mathbf{P}$ holds, and therefore (via \cref{prop:simplesubnetprot:initial,prop:simplesubnetprot:peerpreservation}), we get that $\partyknows{{\party}'}{\party}{{\tau}}{\subnetid}$.

    \medskip\noindent\emph{Induction Step, assuming \ref{item:simplesubnetprot:centrallemma:conda} holds:}
    Our induction hypothesis is that the lemma holds for all pairs $(\tilde{\tau},\tilde{\tau}_{\mathsf{join}})$ preceding $(\tau,\tau_{\mathsf{join}})$ in the lexicographic ordering.
    Assume that \ref{item:simplesubnetprot:centrallemma:conda} holds. Note that in case \ref{item:simplesubnetprot:centrallemma:conda} we cannot have $\tau_{\mathsf{join}} = 0$, as otherwise $\tau'_{\mathsf{join}} < 0$.
    So, we can assume that $\tau_{\mathsf{join}} > 0$.
    By the definition of joining a subnet properly and our observation above, we therefore know that $\party$ calls $\iJoinSubnet(\subnetid, \hat{\party})$ at time $\tau_{\mathsf{join}}$ for an honest party $\hat{\party}$ with $\EventStaysInSubnetDelay{\subnetid}{\hat{\party}}{\tau_{\mathsf{join}}}{(\tau_{\mathsf{join}}+\SubnetDelay)}$, where $\hat{\party}$ joined the subnet properly at some time $\hat{\tau}_{\mathsf{join}}\leq \tau_{\mathsf{join}} - 7$.
    Our claim is that $\partyknows{\hat{\party}}{\party'}{{\tau}_{\mathsf{join}}+1}{\subnetid}$, which then implies  $\partyknows{{\party}'}{\party}{{\tau}_{\mathsf{join}}+3}{\subnetid}$ via \cref{proposition:simplesubnetprot:joinmechanics}, and then $\partyknows{{\party}'}{\party}{{\tau}}{\subnetid}$ because of $\tau \geq \tau_{\mathsf{join}}+3$ and \cref{prop:simplesubnetprot:peerpreservation}.

    We can show this claim using the induction hypothesis applied to
    \footnote{Values with a tilde indicate that we use the induction hypothesis on these, e.g., $\tilde{\party}$ takes the role of $\party$ in the lemma when applying the hypothesis.}
    \begin{align*}
        \tilde{\party} := \party',~\tilde{\party}' := \hat{\party},~\tilde{\tau}_{\mathsf{join}} := \tau_{\mathsf{join}}',~\tilde{\tau}_{\mathsf{join}}' := \hat{\tau}_{\mathsf{join}},~\tilde{\tau} := {\tau}_{\mathsf{join}}+1.
    \end{align*}
    To see why we can apply the induction hypothesis, first observe that $\tilde{\party}' = \hat{\party}$ is active at $\tilde{\tau} = {\tau}_{\mathsf{join}}+1$, because $\EventStaysInSubnetDelay{\subnetid}{\hat{\party}}{\tau_{\mathsf{join}}}{(\tau_{\mathsf{join}}+\SubnetDelay)}$.
    Second, observe that $$(\tilde{\tau},\tilde{\tau}_{\mathsf{join}}) = ({\tau}_{\mathsf{join}}+1,\tau_{\mathsf{join}}') \leq ({\tau}-2,\tau_{\mathsf{join}}-3) < (\tau,\tau_{\mathsf{join}}).$$
    Third, one can see that condition \ref{item:simplesubnetprot:centrallemma:condb} is met:
    \begin{align*}
        &\tilde{\tau} = {\tau}_{\mathsf{join}}+1  > {\tau}_{\mathsf{join}} - 3 + 3  \geq \tau_{\mathsf{join}}' + 3 = \tilde{\tau}_{\mathsf{join}} +3, \\
        \text{and } & \tilde{\tau}_{\mathsf{join}}' = \hat{\tau}_{\mathsf{join}} \leq \tau_{\mathsf{join}} - 7
        \leq ({\tau}_{\mathsf{join}} +1) - 7 = \tilde{\tau} - 7
    \end{align*}

    \medskip\noindent\emph{Induction Step, assuming \ref{item:simplesubnetprot:centrallemma:condb} holds:}
    Again, our induction hypothesis is that the lemma holds for all pairs $(\tilde{\tau},\tilde{\tau}_{\mathsf{join}})$ preceding $(\tau,\tau_{\mathsf{join}})$ in the lexicographic ordering.
    Now, we assume that \ref{item:simplesubnetprot:centrallemma:condb} holds.
    First, note that if $\tau = \tau_{\mathsf{join}} + 3$, then one can see that the conditions in \ref{item:simplesubnetprot:centrallemma:conda} are met, and we have already proven this case.
    So, we can assume that $\tau \geq \tau_{\mathsf{join}} + 4$. Now, if $\tau_{\mathsf{join}}' \leq \tau - 8$, then \ref{item:simplesubnetprot:centrallemma:condb} holds for the same set of parties and $\tilde{\tau} := \tau-1$, so we can apply induction directly. Therefore, we can make the following assumptions: \[
        \tau \geq \tau_{\mathsf{join}} + 4 \land \tau_{\mathsf{join}}' = \tau - 7.
    \]
    Now, consider the case in which $\tau_{\mathsf{join}}' = 0$ and thus $\tau = 7$.
    Then, if $\tau_{\mathsf{join}}=0$ as well, we are in case \ref{item:simplesubnetprot:centrallemma:condc} and we are done. If $\tau_{\mathsf{join}}>0$, then $\party$ joined properly via $\hat{\party}$ at $\tau_{\mathsf{join}} \leq \tau - 3 = 4$. But we assume that there is no join in rounds $1$ to $7$. Therefore, we can extend the assumptions that we make:
    \[
        \tau \geq \tau_{\mathsf{join}} + 4 \land \tau_{\mathsf{join}}' = \tau - 7 \land \tau_{\mathsf{join}}' > 0.
    \]
    This means that party $\party'$ joins the subnet $\subnetid$ properly at time $\tau_{\mathsf{join}}'$ via a party $\hat{\party}'$. Denote the time at which $\hat{\party}$ joined the subnet by $\hat{\tau}_{\mathsf{join}}'$. In the following, we consider two cases:

    \smallskip\noindent\emph{... Case b1: $\tau = \tau_{\mathsf{join}} + 4$.}
    We have already established that $\tau_{\mathsf{join}}' = \tau - 7$. In combination, this means that \[
        \tau = \tau_{\mathsf{join}} + 4 \geq \tau_{\mathsf{join}} + 3 \land \tau_{\mathsf{join}}' = \tau - 7 = (\tau - 4) - 3 = \tau_{\mathsf{join}} - 3.
    \]
    Therefore, this is a special case of \ref{item:simplesubnetprot:centrallemma:conda}, which we have already proven.

    \smallskip\noindent\emph{... Case b2: $\tau \geq \tau_{\mathsf{join}} + 5$.}
    In this case, we use the induction hypothesis to show that $\partyknows{\hat{\party}'}{\party}{{\tau}-2 = \tau_{\mathsf{join}}' +5}{\subnetid}$. Note that this implies our desired statement $\partyknows{{\party}'}{\party}{{\tau}}{\subnetid}$ via \cref{proposition:simplesubnetprot:joinmechanics}.
    Concretely, we apply the induction hypothesis to \begin{align*}
        \tilde{\party} := \party,~\tilde{\party}' := \hat{\party}',~\tilde{\tau}_{\mathsf{join}} := \tau_{\mathsf{join}},
        ~\tilde{\tau}_{\mathsf{join}}' := \hat{\tau}_{\mathsf{join}}',~\tilde{\tau} := {\tau}-2.
    \end{align*}
    Again, let us make explicit why we can use the induction hypothesis:
    For $\tilde{\party}' = \hat{\party}'$ we know that $\EventStaysInSubnetDelay{\subnetid}{\hat{\party}'}{{\tau}_{\mathsf{join}}'}{({\tau}_{\mathsf{join}}'+\SubnetDelay)}$. As $\tau - 2 = {\tau}_{\mathsf{join}}'+5$, this in particular means that $\hat{\party}'$ is active at $\tilde{\tau} = {\tau}-2.$
    Further, it is clear that $(\tilde{\tau},\tilde{\tau}_{\mathsf{join}}) = (\tau-2,\tau_{\mathsf{join}})$ lexicographically precedes $(\tau,\tau_{\mathsf{join}})$. Finally, the conditions of \ref{item:simplesubnetprot:centrallemma:condb} are satisfied as \begin{align*}
        &\tilde{\tau} = {\tau}-2 \geq \tau_{\mathsf{join}} + 5 - 2 = \tilde{\tau}_{\mathsf{join}} + 3, \\
        \text{and } & \tilde{\tau}_{\mathsf{join}}' = \hat{\tau}_{\mathsf{join}}' \leq {\tau}_{\mathsf{join}}' - 7
         = \tau - 7 - 7 \leq  \tau - 2 - 7 = \tilde{\tau} -7.
    \end{align*}
\end{proof}

\begin{theorem}[Robustness of $\subnetprot$]
    Consider the subnet discovery protocol $\subnetprot$ for $\numberofsubnets$ subnets as specified in \cref{fig:subnetprotocol}.
    Let $\AdmissibleSchedules$ be an arbitrary class of join-leave schedules. Let $\ProtocolLifetime \in \NN$ be arbitrary, and let $\RobustnessError = 0$ and $\SubnetDelay = 7$.
    Then, $\subnetprot$ is $(\RobustnessError,\ProtocolLifetime,\SubnetDelay,\AdmissibleSchedules)$-robust.
\end{theorem}
\begin{proof}
    Fix any point in time $\tau \in  [\ProtocolLifetime]$, any subnet identifier $\subnetid \in [\numberofsubnets]$, and any pair of honest parties $\party,\party'$. We need to show that
    \begin{align*}
        \left(
            \begin{array}{rl}
                &\EventIsInSubnetDelay{\subnetid}{\party}{\tau} \\
                \land &\EventIsInSubnetDelay{\subnetid}{\party'}{\tau} \\
                \land &\EventCalledGetPeers{\subnetid}{\party}{\tau}
            \end{array}
        \right) \implies \EventGotPeer{\subnetid}{\party}{\party'}{\tau}.
    \end{align*}
    To do so, we first argue that \begin{equation}
        \left(\EventIsInSubnetDelay{\subnetid}{\party}{\tau}
        \land \EventIsInSubnetDelay{\subnetid}{\party'}{\tau}\right)
        \stackrel{\text{\cref{lemma:simplesubnetprot:centrallemma}}}{\implies} \partyknows{\party}{\party'}{\tau}{\subnetid}.
        \label{eq:simplesubnetprot:proofofrobustness:1}
    \end{equation}
    This is because we can apply \cref{lemma:simplesubnetprot:centrallemma} as follows:
    First, one can assume without loss of generality that no honest party calls $\iJoinSubnet$ with $\subnetid$ in rounds $1$ to $7$\footnote{
        To see why this is without loss of generality, first assume that \cref{eq:simplesubnetprot:proofofrobustness:1} holds for some run of the protocol in which no honest party joins a subnet in rounds $1$ to $7$.
        Then, one can easily see that it also holds for the modified run in which rounds $1$ to $7$ are entirely skipped. This is because the only messages sent by the protocol are due to joining a subnet. As every run with arbitrary join-leave schedule can be seen as such a truncation of a run in which no honest party joins a subnet in rounds $1$ to $7$, we get that \cref{lemma:simplesubnetprot:centrallemma} holds for every run, as long as we can show that it holds for such runs with no joins from rounds $1$ to $7$.
    }. Also, if an honest party calls $\iJoinSubnet$ with $\subnetid$ at round $0$, then our robustness definition does not make any guarantees, and so there is nothing to show.

    Now, if $\EventIsInSubnetDelay{\subnetid}{\party}{\tau}$, then we know that party $\party$ joined the subnet $\subnetid$ properly at time $\tau_{\mathsf{join}} \leq \tau-\SubnetDelay \leq \tau - 3$. Similarly, if $\EventIsInSubnetDelay{\subnetid}{\party'}{\tau}$, then party $\party$ joined the subnet $\subnetid$ properly at time $\tau_{\mathsf{join}}' \leq \tau-\SubnetDelay  = \tau - 7$. So, all conditions of \cref{lemma:simplesubnetprot:centrallemma} (in particular, \cref{item:simplesubnetprot:centrallemma:condb}) are satisfied, which means that $\partyknows{{\party}'}{\party}{{\tau}}{\subnetid}$. Similarly, we get $\partyknows{{\party}}{\party'}{{\tau}}{\subnetid}$.
    Then, using \cref{prop:simplesubnetprot:callgetpeers}, we get the desired implication.
\end{proof}

\superparagraph{Optimizations.}
We make some optimizations when using this subnet discovery protocol to instantiate our protocol $\ourprot$:
first, observe that honest parties only ever join row subnets via their bootstrap nodes.
Thus, nodes with $\AuxRole = 0$ can ignore $\msgJoinSubnet$ and $\msgJoinSubnetPull$ messages for row subnets, as those are supposed to only be sent to bootstrap nodes.
In fact, an honest node $\party$ with $\AuxRole = 0$ and $\AlgRow(\party)=r$ only uses their local list of row subnet peers for calls $\mathbf{P} \coloneqq \AlgGetPeersInCell(r,c)$ for some $c$ and then sends messages to those parties in $\mathbf{P}$.
For that reason, $\party$ only needs to keep track of row subnet peers $\party'$ with $\AlgRow(\party') = r$.
Note that $\party$ will {learn about} their row subnet peers exclusively via $\msgJoinSubnetFwd$ and $\msgJoinSubnetPullRsp$-messages from parties $\hat\party$ with $\AuxRole = 1$.
We shall assume that those $\msgJoinSubnetFwd$ and $\msgJoinSubnetPullRsp$ messages by $\hat\party$ do not contain those peers that $\party$ would not store anyway.
Furthermore, we assume that every honest party $\party$ ignores column subnet messages originating from (necessarily malicious) parties $\party'$ with $\AlgCol(\party')\neq\AlgCol(\party)$.

\clearpage
\ifnum\ccs=1
\section{Proof of~\cref{theorem:ourprot:analysis:maintheorem}}\label{section:appendix_proof}

In this section, we will prove \cref{theorem:ourprot:analysis:maintheorem}.
Our proof is organized as follows: we first define several good events that capture the conditions under which we will show robustness.
Then, the actual proof of the main theorem has two main technical parts:
in the first part, we show that construction satisfies what we want from the robustness notion (i.e., that we can retrieve data that was previously stored), provided certain good events occur.
This first part does not use probabilities.
Then, in the second part, we show that the good events from the first part actually occur with sufficiently high probability, provided that our join-leave schedule is admissible.
Combined, these parts then prove \cref{theorem:ourprot:analysis:maintheorem}.
\else

Below, we prove \cref{theorem:ourprot:analysis:maintheorem} by a series of lemmas, using the strategy outlined at the beginning of the section.
\fi

\ifnum\ccs=1
\subsection{Good Events}

\REMARKABOUTACTIVITY

\DEFINEGOODEVENTS

\begin{definition}[Corruption Sets]\label{def:our_corruption_sets}
Our robustness definition allows for corruption sets $\Corrupted_\party^\tau\subseteq [m]$, where the property we want from robustness is allowed to fail, provided those sets are small enough.
cf.~\cref{def:subnetdiscovery:robustness}. For our concrete \cref{theorem:ourprot:analysis:maintheorem}, we set our corruption sets as follows:

For an honest party $\party$ and time slot $\tau$, set
        \begin{align*}
         \Corrupted_\party^\tau = \bigl\{&i\in [\FileLength] \mid \neg \EventGoodCell{r}{c}{\max\{0,\tau-\SyncDelay\}}{(\tau+1)}\\
         &\text{ for } c=\AlgGetColForSymbol(i)\bigr\}
        \end{align*}
where $r = \AlgRow(\party)$.
\end{definition}

\fi
%
%
\ifnum\ccs=0
\subsubsection{Our Protocol Works, Provided Good Events Occur}
\else
\subsection{Our Protocol Works, Provided Good Events Occur}
\fi

In the first part of our proof, we will show that our protocol $\ourprot$ works (i.e.\ we can successfully query data that was previously stored), provided certain good event occur.
For this, we will show a series of technical lemmas, roughly showing (in order) that under appropriate conditions
\begin{myenumerate}
\item Joining the subnets works.
\item After calling $\ourprot.\iStore$, the data is stored by nodes in the appropriate column subnet.
\item Data stored in a column subnet is retained, even if parties leave and join.
\item Calling $\ourprot.\iGet$ will successfully retrieve data stored in the column.
\item We can get data that we stored (i.e.\ the property we want) by combining the above.
\end{myenumerate}

\superparagraph{Joining the subnets works.}
Let us start by analyzing the joining of row and column subnets.

\begin{lemma}[Joining the rows and columns]\label{lem:join_rows_works}
Let $\JoinLeaveSchedule$ be a join-leave schedule that respects bootstrap nodes and uses good bootstrap nodes as in \cref{def:ourprot:analysis:admissibleschedules}.
Consider a run of our protocol $\ourprot$. Assume $\EventColumnGoodUntil{c}{T}{2\SubnetDelay+2}$ and $\EventSubnetprotGood{T+1}$ hold for some lifetime $T$.
Let $\party$ be an honest party joining (via $\iJoin$ or $\iInit$) our protocol $\ourprot$ at time $\tauJoin \geq 0$ with $\AuxRole\in\{0,1\}$.
Let $(r,c) = \AlgCell(\party)$.
Let $\tauActive$ be some time where $\party$ is active. Then the following hold if $\tauJoin > 0$ (i.e. $\party$ is not one of the initial parties):
\begin{myitemize}
\item \ifnum\ccs=0 The event\fi $\EventIsInSubnet{\subnetid_r}{\party}{\tauActive}$ occurs if and only if $\tauActive \geq \tauJoin + \SubnetDelay$.
\item If $\AuxRole=1$, then for all row subnets $\subnetid_{r'}$, we have\ifnum\ccs=1:\\\fi $\EventIsInSubnet{\subnetid_{r'}}{\party}{\tauActive}$ if and only if $\tauActive\geq \tauJoin + \SubnetDelay$.
\item \ifnum\ccs=0 If $\tauJoin \leq T$, then $\EventIsInSubnet{\subnetid_c}{\party}{\tauActive}$ occurs if and only if $\tauActive\geq \tauJoin+\SubnetDelay + 2$.%
      \else If $\tauJoin \leq T$, then:\\ $\EventIsInSubnet{\subnetid_c}{\party}{\tauActive}$ if and only if $\tauActive\geq \tauJoin+\SubnetDelay + 2$.%
      \fi
\end{myitemize}
In case $\tauJoin = 0$ (i.e.\ $\party$ is an initial party), we can strengthen this to
\begin{myitemize}
\item $\EventIsInSubnet{\subnetid_c}{\party}{\tauActive}$ and $\EventIsInSubnet{\subnetid_r}{\party}{\tauActive}$ hold \ifnum\ccs=0 unconditionally\fi.
\item If $\AuxRole = 1$, then\ifnum\ccs=1:\\\fi $\EventIsInSubnet{\subnetid_{r'}}{\party}{\tauActive}$ holds for all row subnets $\subnetid_{r'}$.
\end{myitemize}
\end{lemma}

\medskip\noindent
Before we prove \cref{lem:join_rows_works}, we state three corollaries:
\ifnum\ccs=0
fully joined nodes are in their appropriate subnets,
$\EventColumnGoodUntilNoArgs$ guarantees honest parties with $\EventIsInSubnetNoArgs$ (rather than just active),
and $\EventGoodCellNoArgs$ means we can find honest parties in their subnets (rather than just active) via $\AlgGetPeersInCell$.
\else
\begin{myitemize}
\item Fully joined nodes are in their appropriate subnets.
\item $\EventColumnGoodUntilNoArgs$ guarantees honest parties with $\EventIsInSubnetNoArgs$ (rather than just active),
\item $\EventGoodCellNoArgs$ means we can find honest parties in their subnets (rather than just active) via $\AlgGetPeersInCell$.
\end{myitemize}
\fi
%
%
\begin{corollary}[\EventFullyJoinedNoArgs\ and \EventIsInSubnetNoArgs]\label{cor:FullyJoinedImpliesSubnets}
Let $\JoinLeaveSchedule$ be a join-leave schedule that respects bootstrap nodes and uses good bootstrap nodes. Consider a run of our protocol $\ourprot$.
\ifnum\ccs=0
Assume $\EventSubnetprotGood{T+1}$ and $\EventColumnGoodUntil{c}{T}{2\SubnetDelay+2}$ hold for some lifetime $T$.
\else
Let us assume $\EventSubnetprotGood{T+1}$ and $\EventColumnGoodUntil{c}{T}{2\SubnetDelay+2}$ hold for some lifetime $T$.
\fi
Let $\party$ be an honest party with $\AlgCell(\party) = (r,c)$. Let $0\leq \tau \leq T$ be any time slot. 
Then we have the following:
\begin{itemize}
 \item $\EventFullyJoined{\party}{\tau} \implies \EventIsInSubnet{\subnetid_r}{\party}{\tau}$.
 \item $\EventFullyJoined{\party}{\tau} \iff \EventIsInSubnet{\subnetid_c}{\party}{\tau}$.
\end{itemize}
\end{corollary}
\begin{proof}[Proof of \cref{cor:FullyJoinedImpliesSubnets}]
Let $\tauJoin$ be the time when $\party$ joins. If $\tau < \tauJoin$, the statements hold true, since none of the events can occur. So let $\tau \geq \tauJoin$.
In this case, $\tau \leq T$ implies $\tauJoin \leq T$. 
Now, the corollary follows from \cref{lem:join_rows_works} and the fact that $\ourprot.\iJoin$ takes exactly $2+\SubnetDelay$ rounds and $\ourprot.\iInit$ finishes the same round it starts.
\end{proof}
\begin{corollary}[Strengthened $\EventColumnGoodUntilNoArgs$]\label{cor:strengthend_subnets}
Let $\JoinLeaveSchedule$ be a join-leave schedule that respects bootstrap nodes and uses good bootstrap nodes.
Consider a run of our protocol $\ourprot$.
Let $\Delta \geq 2\SubnetDelay + 2$ and consider some lifetime $T\geq 0$.
Assume that $\EventSubnetprotGood{T}$ and $\EventColumnGoodUntil{c}{T}{\Delta}$ hold for some column $c\in [k_2]$.

Then for any $0\leq \tau \leq T$, there exists an honest party $\party$ with $\EventStaysInSubnet{\subnetid_c}{\party}{\tau}{(\tau+\Delta')}$, where $\Delta' \coloneqq \Delta - \SubnetDelay - 2$.
\end{corollary}
\begin{proof}[Proof of \cref{cor:strengthend_subnets}]
Since $\EventColumnGoodUntil{c}{T}{\Delta}$ holds, there exists an honest initial party $\party_0$ with $\AlgCol(\party_0) = c$ and $\EventActiveDuration{\party_0}{0}{\Delta}$.
By \cref{lem:join_rows_works}, we get $\EventStaysInSubnet{\subnetid_c}{\party_0}{0}{\Delta}$, which shows the claim for $0 \leq \tau \leq \SubnetDelay + 2$.

So assume $\SubnetDelay + 2 < \tau \leq T$.
Again, since $\EventColumnGoodUntil{c}{T}{\Delta}$ holds,
\ifnum\ccs=0
there exists some party $\party$ with $\AlgCol(\party) = c$ and
\else
there must exist some party $\party$ with $\AlgCol(\party) = c$ and
\fi
$\EventActiveDuration{\party}{(\tau - \SubnetDelay-2)}{(\tau - \SubnetDelay - 2 + \Delta)}$, i.e.\ $\EventActiveDuration{\party}{(\tau-\SubnetDelay - 2)}{(\tau + \Delta')}$.
Applying \cref{lem:join_rows_works} (with $T$ set to $T-\SubnetDelay$) gives $\EventStaysInSubnet{\subnetid_c}{\party}{\tau}{(\tau+\Delta')}$.
\end{proof}
\begin{corollary}[Strengthened $\EventGoodCellNoArgs$]\label{cor:GoodCellImpliesFindInCell}
Let $\JoinLeaveSchedule$ be a join-leave schedule that respects bootstrap nodes and uses good bootstrap nodes and consider a run our protocol $\ourprot$.
Let $T$ be some lifetime bound, $r\in[k_1]$ a row, $c\in[k_2]$ a column and $\Delta \geq 2\SubnetDelay +2$.
Assume $\EventSubnetprotGood{T}$ and $\EventColumnGoodUntil{c}{T}{\Delta}$ hold.
Then the following hold for $\tau_1\leq \tau_2$ with $\tau_1 \leq T$.
\begin{myenumerate}
\item If $\EventGoodCell{r}{c}{\tau_1}{\tau_2}$, then there exists an honest $\party$ with
        $\AlgCell(\party) = (r,c)$\ifnum\ccs=0 and \else, \fi
        $\EventStaysInSubnet{\subnetid_r}{\party}{\tau_1}{\tau_2}$ and
        $\EventStaysInSubnet{\subnetid_c}{\party}{\tau_1}{\tau_2}$.\label{item:goodcellitem}
\item If $\EventGoodCell{r}{c}{\tau_1}{\tau_2}$ \ifnum\ccs=1 holds \fi and some honest party $\party'$ with $\EventIsInSubnet{\subnetid_r}{\party'}{\tau}$ calls $\mathbf{P}\coloneqq \AlgGetPeersInCell(r,c)$ at some time $\tau_1\leq \tau \leq \tau_2$ with $\tau \leq T$,
then there exists a $\party\in\mathbf{P}$ with $\party$ as in \cref{item:goodcellitem}.\label{item:findable_by_GetPeersInCell}
\end{myenumerate}
\end{corollary}
\begin{proof}[Proof of \cref{cor:GoodCellImpliesFindInCell}]
By definition, $\EventGoodCell{r}{c}{\tau_1}{\tau_2}$ means that there exists an honest party $\party$ with $\AlgCell(\party)=(r,c)$ and $\EventActiveDuration{\party}{\max\{0,\tau_1-\SubnetDelay-2\}}{\tau_2}$.
This means that $\party$ must have joined at time $\tauJoin \leq \max\{0,\tau_1-\SubnetDelay-2\}$.
If $\tau_1 < \SubnetDelay + 2$, then $\party$ is an initial party and we get $\EventStaysInSubnet{\subnetid_r}{\party}{0}{\tau_2}$ and $\EventStaysInSubnet{\subnetid_c}{\party}{0}{\tau_2}$ by \cref{lem:join_rows_works}.
Hence, $\EventStaysInSubnet{\subnetid_r}{\party}{\tau_1}{\tau_2}$ and $\EventStaysInSubnet{\subnetid_c}{\party}{\tau_1}{\tau_2}$.
Otherwise, we must have $\tau_1 > \SubnetDelay + 2$, so $\tauJoin \leq \tau_1 - \SubnetDelay - 2$. This implies $\tauJoin + 1 \leq T$ and we can again use \cref{lem:join_rows_works} (with $T$ set to $T-1$) to deduce that $\EventStaysInSubnet{\subnetid_r}{\party}{\tau_1}{\tau_2}$ and $\EventStaysInSubnet{\subnetid_c}{\party}{\tau_1}{\tau_2}$.

\cref{item:findable_by_GetPeersInCell} is a direct consequence of the definition of $\AlgGetPeersInCell$ and $\EventSubnetprotGoodNoArgs$.
\end{proof}
Let us now prove \cref{lem:join_rows_works}:
\begin{proof}
For the ``only if'' part of the statements, we start by observing that the only places where we call $\subnetprot.\iJoinSubnet$ resp.\ $\subnetprot.\iCreateSubnet$ are in $\ourprot.\iJoin$ resp.\ $\ourprot.\iInit$.
For $\tauJoin = 0$, the ``only if''-statement is trivial. For $\tauJoin > 0$, observe that $\party$ calls $\subnetprot.\iJoinSubnet$ for row subnets at time $\tauJoin$ and for column subnets at time $\tauJoin+2$.
The ``only if''-claim then follows from the definition of $\EventIsInSubnetNoArgs$.

\smallskip
The proof for the ``if'' part works by induction over $\tauJoin \geq 0$.

\medskip
Let us start with the $\tauJoin = 0$ case, so $\party$ is among the initial parties $\{\party_1,\ldots,\party_\ell\}$ with roles $\AuxRole_1,\ldots,\AuxRole_\ell$.
Let $\RowParties_{r'} \coloneqq \{\party_i \mid \AlgRow(\party_i) = r' \lor \AuxRole_i = 1\}$ be as in $\iInit$ for any $r'\in[k_1]$. Also, let $\ColParties_{c'} \coloneqq \{\party_i \mid \AlgCol(\party_i) = c'\}$ for any $c'\in [k_2]$.
In our protocol, the only calls to $\subnetprot.\iCreateSubnet$ made by honest parties are at time $\tauJoin = 0$ and made by exactly the parties from $\RowParties_{r'}$ resp.\ $\ColParties_{c'}$ with argument $\RowParties_{r'}$ resp.\ $\ColParties_{c'}$.
Since $\party\in\RowParties_{r}$, this implies $\EventCreatedSubnet{\subnetid_r}{\RowParties_r}$ and $\EventIsInSubnet{\subnetid_r}{\party}{0}$.
\ifnum\ccs=0
Similarly, $\party\in\ColParties_{c}$, so we get $\EventCreatedSubnet{\subnetid_{c}}{\ColParties_c}$ and $\EventIsInSubnet{\subnetid_c}{\party}{0}$.
\else
Similarly, $\party\in\ColParties_{c}$, so we know that $\EventCreatedSubnet{\subnetid_{c}}{\ColParties_c}$ and $\EventIsInSubnet{\subnetid_c}{\party}{0}$ holds.
\fi
This implies $\EventIsInSubnet{\subnetid_r}{\party}{\tauActive}$ resp.\ $\EventIsInSubnet{\subnetid_c}{\party}{\tauActive}$ for any $\tauActive\geq 0$ for which party $\party$ is still active.
If furthermore $\AuxRole = 1$, then $\party\in\RowParties_{r'}$ for all $r'$, so we get $\EventIsInSubnet{\subnetid_{r'}}{\party}{\tauActive}$ for all $r'\in[k_1]$ and all $\tauActive\geq 0$ for which party $\party$ is still active.

\medskip
Let us proceed with the induction step, where we assume that $\tauJoin > 0$ and all statements of the lemma are proven for all $\tau_0 < \tauJoin$.
Since $\tauJoin > 0$, the party $\party$ joins by $\iJoin$ via some bootstrap nodes $\party_1,\ldots,\party_t$.
Since $\JoinLeaveSchedule$ respects bootstrap nodes and uses good bootstrap nodes, there exists a prospective bootstrap node $\party'$ among them that is active at least from time $\max\{\tauJoin-\SubnetDelay,0\}$ to $\tauJoin + \SubnetDelay$.
We first claim that $\EventStaysInSubnet{\subnetid_{r'}}{\party'}{\tauJoin}{(\tauJoin+\SubnetDelay)}$ holds for all row indices $r\in [k_1]$:
If $\tauJoin \geq \SubnetDelay$, then $\party'$ must have joined at $\tauJoin-\SubnetDelay$ or earlier. Our first claim is thus implied by the induction hypothesis.
Otherwise, $\tauJoin < \SubnetDelay$ and $\party'$ is active from time 0, so it must be one of the initial parties. Using the induction hypothesis on $\party'$ with $\tauJoin = 0$ then also shows our first claim about $\party'$.

\medskip
We now continue with the induction step for the row subnets.
If $\AuxRole = 0$, then party $\party$ will call $\subnetprot.\iJoinSubnet(\subnetid_r,\party')$ at time $\tauJoin$ as part of its $\iJoin$ call.
By definition of $\EventIsInSubnetNoArgs$, this implies $\EventIsInSubnet{\subnetid_r}{\party}{\tauActive}$ for any $\tauActive \geq \tauJoin + \SubnetDelay$ for which $\party$ is still active.
\ifnum\ccs=0
Similarly, if $\AuxRole = 1$, then party $\party$ calls $\subnetprot.\iJoinSubnet(\subnetid_{r'},\party_i)$ at time $\tauJoin$ for all rows $r'\in[k_1]$ as part of its $\iJoin$ call.
\else
Similarly, if $\AuxRole = 1$, then during $\iJoin$, party $\party$ calls $\subnetprot.\iJoinSubnet(\subnetid_{r'},\party_i)$ at time $\tauJoin$ for all rows $r'\in[k_1]$ as part of its $\iJoin$ call.
\fi
This implies $\EventIsInSubnet{\subnetid_{r'}}{\party}{\tauActive}$ for any $\tauActive \geq \tauJoin + \SubnetDelay$ for which $\party$ is still active, finishing the part of the lemma with the row subnets.

\medskip
Finally, we prove the induction step for the column subnets.
Set $\tauShifted\coloneqq \max\{0,\tauJoin - \SubnetDelay\}$.
Note that since $\tauJoin > 0$, we have $\tauShifted < \tauJoin$.
Because $\EventColumnGoodUntil{c}{T}{2\SubnetDelay+2}$ holds with $\tauShifted \leq T$, there exists some honest party $\widetilde{\party}$ with
$\EventActiveDuration{\widetilde{\party}}{\tauShifted}{(\tauShifted + 2\SubnetDelay+2)}$ and $\AlgCol(\widetilde{\party}) = c$.
By definition of $\tauShifted$, we have
\[
    \tauShifted + 2\SubnetDelay + 2\geq \tauJoin- \SubnetDelay + 2\SubnetDelay + 2 = \tauJoin + \SubnetDelay + 2.
\]
So $\widetilde{\party}$ joined at $\tauShifted$ or earlier and stays active until at least
$\tauJoin+\SubnetDelay+2$.

Let $\widetilde{r}\coloneqq \AlgRow(\widetilde{\party})$ be the row assigned to $\widetilde{\party}$.
In case $\tauShifted>0$, using the induction hypothesis for the rows, applied to $\tauShifted < \tauActive$, gives us
$\EventStaysInSubnet{\subnetid_{\tilde{r}}}{\widetilde{\party}}{(\tauShifted+\SubnetDelay)}{(\tauJoin + \SubnetDelay+2)}$.
Using the induction hypothesis for the columns gives $\EventStaysInSubnet{\subnetid_{c}}{\widetilde{\party}}{(\tauShifted+\SubnetDelay+2)}{(\tauJoin+\SubnetDelay+2)}$.
We can rewrite those using the definition of $\tauShifted$, for $\tauShifted>0$ as
\ifnum\ccs=0
\[
\EventStaysInSubnet{\subnetid_{\tilde{r}}}{\widetilde{\party}}{\tauJoin}{(\tauJoin + \SubnetDelay+2)}\text{ and } \EventStaysInSubnet{\subnetid_c}{\widetilde{\party}}{(\tauJoin+2)}{(\tauJoin+\SubnetDelay+2)}  \enspace.
\]
\else
\begin{align*}
&\EventStaysInSubnet{\subnetid_{\tilde{r}}}{\widetilde{\party}}{\tauJoin}{(\tauJoin + \SubnetDelay+2)}\\
\text{and }&\EventStaysInSubnet{\subnetid_c}{\widetilde{\party}}{(\tauJoin+2)}{(\tauJoin+\SubnetDelay+2)}  \enspace.
\end{align*}
\fi
If $\tauShifted = 0$, we know that $\widetilde{\party}$ is one of the initial parties and the exact same statements
$\EventStaysInSubnet{\subnetid_{\tilde{r}}}{\widetilde{\party}}{\tauJoin}{(\tauJoin + \SubnetDelay+2)}$
and
$\EventStaysInSubnet{\subnetid_{c}}{\widetilde{\party}}{(\tauJoin+2)}{(\tauJoin+\SubnetDelay+2)}$
hold by induction hypothesis applied to $\widetilde{\party}$ and $\tauShifted=0$.

Now, during the $\iJoin$-call of $\party$, the party $\party$ sends a $\msgJoin$-message to $\party'$ at time $\tauJoin$.
We already established $\EventStaysInSubnet{\subnetid_{\tilde{r}}}{\party'}{\tauJoin}{(\tauJoin+\SubnetDelay)}$, so $\party'$ is still active at time $\tauJoin+1$ and will receive that message.
Then $\party'$ will call (among others) $\subnetprot.\iGetPeers(\subnetid_{\tilde r})$ at time $\tauJoin+1$.
Since $\EventSubnetprotGood{T+1}$, $\tauJoin+1\leq T+1$, and $\EventIsInSubnet{\subnetid_{\tilde{r}}}{\widetilde{\party}}{\tauJoin+1}$ hold,
this call will return $\widetilde{\party}$ (among others)
and be included in the local computation of $\ColParties_c$ by $\party'$.
Then $\party'$ will send back a $(\msgJoinRsp,\ColParties_c)$-message to $\party$ with $\widetilde{\party}\in\ColParties_c$.
This message is received by $\party$ at time $\tauJoin+2$, so $\party$ will call $\subnetprot.\iJoinSubnet(\subnetid_c,\widetilde{\party})$ at time $\tauJoin+2$.
Since $\EventStaysInSubnet{\subnetid_c}{\widetilde{\party}}{(\tauJoin+2)}{(\tauJoin+\SubnetDelay+2)}$ holds, we obtain $\EventIsInSubnet{\subnetid_c}{\party}{\tauActive}$ for any $\tauActive \geq \tauJoin + 2 + \SubnetDelay$ by definition of $\EventIsInSubnetNoArgs$, finishing the proof.
\end{proof}
\superparagraph{Storing data works.} We now proceed to show that after calling $\ourprot.\iStore$, data is stored by the column subnet.
In fact, data will be stored by \emph{all} honest parties in a given column subnet that have been active for a sufficiently long time.
We introduce the following events:
\begin{definition}[Storage Events]\label{def:relevant_events}
For $h\in\HandleSpace$, $i\in[\FileLength]$, $x\in\SymbolSpace$, we define the following events:
\begin{itemize}
    \item Event $\EventInSymbolStorage{\party}{\tau}{h}{i}{x}$:\ifnum\ccs=1\\\fi
            This event occurs if party $\party$ is honest, active at time $\tau$, and at the beginning of time slot $\tau$,
            we have $\SymbolsMap[h,i] = x$ for $\party$'s local map $\SymbolsMap$.
    \item Event $\EventStoredInColumnNew{c}{\tau_0}{\tau_1}{h}{i}{x}$ for times $\tau_0\leq \tau_1$ and column $c$:\ifnum\ccs=1\\\fi
            \ifnum\ccs=0
            This event occurs if for every honest party $\party$ with $\EventStaysInSubnet{\subnetid_c}{\party}{\tau_0}{\tau_1}$, we have $\EventInSymbolStorage{\party}{\tau_1}{h}{i}{x}$.
            \else
            This occurs if for every honest party $\party$ with $\EventStaysInSubnet{\subnetid_c}{\party}{\tau_0}{\tau_1}$, we have $\EventInSymbolStorage{\party}{\tau_1}{h}{i}{x}$.
            \fi
    \item Event $\EventStoredAndRetain{c}{\tau_0}{\tau_1}{h}{i}{x}$:\ifnum\ccs=1\\\fi
        This event occurs for a column $c$ and times $\tau_0 \leq \tau_1$ if there exists an honest party $\party$ with $\EventInSymbolStorage{\party}{\tau_0}{h}{i}{x}$ and $\EventStaysInSubnet{\subnetid_c}{\party}{(\tau_0-1)}{\tau_1}$.
 \end{itemize}
We remark that if we write a symbol during time slot $\tau$, the corresponding event $\EventInSymbolStorageNoArgs$ only holds starting from time $\tau+1$.
The meaning of $\EventInSymbolStorageNoArgs$ is that we can \emph{reliably query} the symbol, which is not guaranteed in the same round that we write the symbol, absent intra-slot ordering constraints on event processing.
This also means that in the definition of $\EventStoredAndRetain{c}{\tau_0}{\tau_1}{h}{i}{x}$, the symbol was stored by $\party$ in round $\tau_0-1$ or earlier, which guarantees that $\tau_0 - 1\geq 0$.
\end{definition}

\noindent We start by observing that honest parties will not overwrite data if $\rdspred$ is position-binding:
\begin{lemma}\label{lem:StaysInSymbolStorage}
Let $\party$ be an honest party, $\tau_1\leq \tau_2$ some time slots, $h\in\HandleSpace$, $x\in\SymbolSpace$ and $i\in[m]$. Assume that $\rdspred$ is position-binding and consider a run of $\ourprot$. Then
\ifnum\ccs=0
\[
 \EventInSymbolStorage{\party}{\tau_1}{h}{i}{x} \land \EventActive{\party}{\tau_2} \implies \EventInSymbolStorage{\party}{\tau_2}{h}{i}{x}\enspace.
\]
\else
\begin{align*}
 & \EventInSymbolStorage{\party}{\tau_1}{h}{i}{x} \land \EventActive{\party}{\tau_2}\\ \implies & \EventInSymbolStorage{\party}{\tau_2}{h}{i}{x}\enspace.
\end{align*}
\fi
Furthermore, if $\party$ stores $\SymbolsMap[h,i]\coloneqq x$ at time $\tau$ and is active at time $\tau+1$, then $\EventInSymbolStorage{\party}{\tau+1}{h}{i}{x}$ holds.
\end{lemma}
\begin{proof}
In our protocol, each honest party initializes $\SymbolsMap$ only once as part of $\iInit$ resp.\ $\iJoin$ before any writes are made to it. After that, every write $\SymbolsMap[h,i]\coloneqq x$ is guarded by a check that $\rdspred(h,i,x) = 1$.
This implies that no value stored to $\SymbolsMap$ is ever overwritten by a different value (otherwise, this would directly violate position-binding of $\rdspred$). The claims follow.
\end{proof}

With this, we can state some easy properties of the events defined in \cref{def:relevant_events}
\begin{proposition}\label{prop:relevant_events}
Let $\JoinLeaveSchedule$ be a join-leave schedule that respects bootstrap nodes and uses good bootstrap nodes.
Let $h\in\HandleSpace$, $i\in [ \FileLength ]$, $x\in\SymbolSpace$. Assume that $\rdspred$ is position-binding and consider a run of $\ourprot$. Then we have
\begin{enumerate}[label=(\alph*)]
\item \ifnum\ccs=0
        $\EventStoredInColumnNew{c}{\tau_0}{\tau_1}{h}{i}{x}\implies \EventStoredInColumnNew{c}{\tau_0}{\tau_2}{h}{i}{x}$ for all $\tau_0\leq \tau_1\leq \tau_2$.
      \else
        \begin{align*}
        &\EventStoredInColumnNew{c}{\tau_0}{\tau_1}{h}{i}{x}\\\implies &\EventStoredInColumnNew{c}{\tau_0}{\tau_2}{h}{i}{x}
        \end{align*}
        for all $\tau_0\leq \tau_1\leq \tau_2$.
      \fi
\label{propitem:EventStoreInColumnExtends}
\item Assume $\EventColumnGoodUntil{c}{T}{\Delta}$ for some lifetime $T$ and $\Delta\geq 2\SubnetDelay+2$. Set $\Delta'\coloneqq \Delta-\SubnetDelay - 2$. For $0<\tau_0 \leq \tau_1$, we have
\ifnum\ccs=0
\begin{align*}
     \left(\begin{array}{rl}
                & \tau_0 \leq T+1\\
        \land   & \tau_0 + \Delta' -1 \geq \tau_1\\
        \land   & \EventStoredInColumnNew{c}{\tau_0}{\tau_1}{h}{i}{x}
     \end{array}\right)
     \implies   \EventStoredAndRetain{c}{\tau_1}{(\tau_0 + \Delta'-1)}{h}{i}{x}\enspace.
\end{align*}
\else
\begin{align*}
                &   \tau_0 \leq T+1 \\
     {}\land{}     &     \tau_0 + \Delta' -1 \geq \tau_1\\
     {}\land{}     &     \EventStoredInColumnNew{c}{\tau_0}{\tau_1}{h}{i}{x}\\
     {}\implies{}  &    \EventStoredAndRetain{c}{\tau_1}{(\tau_0 + \Delta'-1)}{h}{i}{x}\enspace.
\end{align*}
\fi
\label{propitem:StoredAndRetain}
\end{enumerate}
\end{proposition}
\begin{proof}
\cref{propitem:EventStoreInColumnExtends} is a direct consequence of \cref{lem:StaysInSymbolStorage} and the definition of $\EventStoredInColumnNoArgs$.

For \cref{propitem:StoredAndRetain}, $\EventColumnGoodUntil{c}{T}{\Delta}$ implies by \cref{cor:strengthend_subnets} that there exists some honest $\party$ with
$\EventStaysInSubnet{\subnetid_c}{\party}{(\tau_0-1)}{(\tau_0-1+\Delta')}$.

Since $\tau_0 - 1 < \tau_1 \leq \tau_0-1+\Delta'$, we have $\EventIsInSubnet{\subnetid_c}{\party}{\tau_1}$ and $\EventStaysInSubnet{\subnetid_c}{\party}{(\tau_1-1)}{(\tau_0+\Delta'-1)}$.
\ifnum\ccs=0
By definition, $\EventStoredInColumnNew{c}{\tau_0}{\tau_1}{h}{i}{x}$ then implies $\EventInSymbolStorage{\party}{\tau_1}{h}{i}{x}$, which shows $\EventStoredAndRetain{c}{\tau_1}{(\tau_0+\Delta'-1)}{h}{i}{x}$.
\else
By definition of $\EventStoredInColumnNoArgs$, $\EventStoredInColumnNew{c}{\tau_0}{\tau_1}{h}{i}{x}$ implies $\EventInSymbolStorage{\party}{\tau_1}{h}{i}{x}$, which shows $\EventStoredAndRetain{c}{\tau_1}{(\tau_0+\Delta'-1)}{h}{i}{x}$ holds.
\fi

\end{proof}

\begin{lemma}[Store Works]\label{lemma:store_works}
    Let $\JoinLeaveSchedule$ be an join-leave schedule that respects bootstrap nodes and uses good bootstrap nodes.
    Assume $\rdspred$ is position-binding and consider a run of our protocol $\ourprot$.
    Let $\PartyStoring$ be an honest party with $\AlgRow(\PartyStoring) = r$. 
    Let $x\in\SymbolSpace, i\in [\FileLength], h\in\HandleSpace$ for which $\rdspred(h,i,x)$ holds and let $c \coloneqq\AlgGetColForSymbol(i)$. 
    Assume $\EventColumnGoodUntil{c}{T}{2\SubnetDelay+2}$ and $\EventSubnetprotGood{T+1}$ hold for some lifetime $T$.
    Then we have for all $0 \leq \tau \leq T$
    \ifnum\ccs=0
    \begin{align*}
     \left(\begin{array}{rl}
                    & \EventGoodCell{r}{c}{\tau}{(\tau+1)}\\
            \land   & \EventFullyJoined{\PartyStoring}{\tau}\\
            \land   & \EventStored{\PartyStoring}{\tau}{h}{i}{x}
           \end{array}\right)
    \implies \EventStoredInColumnNew{c}{(\tau+1)}{(\tau+3)}{h}{i}{x}\enspace.
    \end{align*}
    \else
    \begin{align*}
                    & \EventGoodCell{r}{c}{\tau}{(\tau+1)}\\
            {}\land{}   & \EventFullyJoined{\PartyStoring}{\tau}\\
            {}\land{}   & \EventStored{\PartyStoring}{\tau}{h}{i}{x}\\
    {}\implies{} &\EventStoredInColumnNew{c}{(\tau+1)}{(\tau+3)}{h}{i}{x}\enspace.
    \end{align*}
    \fi
    \begin{proof}
    Let $\party$ be an honest party with $\EventStaysInSubnet{\subnetid_c}{\party}{(\tau+1)}{(\tau+3)}$.
    We need to show that $\party$ stores the data $(h,i,x)$ at the beginning of time slot $\tau+3$.

    Since $\EventFullyJoined{\PartyStoring}{\tau}$ holds, we have $\EventIsInSubnet{\subnetid_r}{\PartyStoring}{\tau}$ by \cref{cor:FullyJoinedImpliesSubnets}.
    Since $\EventGoodCell{r}{c}{\tau}{(\tau+1)}$ holds, there exist an honest party $\party'$ with $\AlgCell(\party') = (r,c)$ with $\EventStaysInSubnet{\subnetid_r}{\party'}{\tau}{(\tau+1)}$ and $\EventStaysInSubnet{\subnetid_c}{\party'}{\tau}{(\tau+1)}$ by \cref{cor:GoodCellImpliesFindInCell}.

    By assumption, $\PartyStoring$ calls $\ourprot.\iStore(h,i,x)$ at time $\tau$.
    During this call, $\PartyStoring$ internally calls $\mathbf{P} \coloneqq \AlgGetPeersInCell(r,c)$ at time $\tau$.
    Since $\EventSubnetprotGood{T}$ holds, this finds $\party'\in\mathbf{P}$ by \cref{cor:GoodCellImpliesFindInCell} and so $\PartyStoring$ sends $(\msgStore,h,i,x)$ to $\party'$ at time $\tau$.

    \smallskip\noindent
    This peer $\party'$ receives $(\msgStore, h,i,x)$ at time $\tau+1$ and is still active.
    Party $\party'$ calls $\mathbf{P}' \coloneqq \subnetprot.\iGetPeers(\subnetid_c)$ at time $\tau+1$. 
    Again, $\EventSubnetprotGood{T+1}$ holds, which implies that $P \in \mathbf{P}'$.
    So party $\party'$ sends $(\msgStoreFwd, h,i,x)$ to $\party$.
    \smallskip\noindent
    The party $\party$ is still active at time $\tau+2$ by definition, so it will receive $(\msgStoreFwd,h,i,x)$ at time $\tau+2$.
    It reacts by storing the symbol at time $\tau+2$, so we get $\EventInSymbolStorage{\party}{\tau+3}{h}{i}{x}$ by \cref{lem:StaysInSymbolStorage}, which finishes the proof.
\end{proof}
\end{lemma}

\superparagraph{Retaining Data.} We next show that data that was stored in a column is retained.

\begin{lemma}[Retaining Data]\label{lemma:retaining_data}
Let $\JoinLeaveSchedule$ be a join-leave schedule that respects bootstrap nodes and uses good bootstrap nodes. Consider a run of our protocol $\ourprot$.
Let $x\in\SymbolSpace, i\in [\FileLength], h\in\HandleSpace$ for which $\rdspred(h,i,x)$ holds and let $c \coloneqq\AlgGetColForSymbol(i)$.
Assume that $\EventColumnGoodUntil{c}{T+1}{2\SubnetDelay+2}$ and $\EventSubnetprotGood{T+2}$ hold for some lifetime $T$.
Then we have for $0 \leq \tau \leq T$
\ifnum\ccs=0
\begin{align*}
    \left(\begin{array}{rl}
                    & \EventStoredInColumnNew{c}{\tau}{\tau + \SyncDelay+1}{h}{i}{x}\\
            \land   & \EventStoredAndRetain{c}{(\tau+2)}{(\tau+\SyncDelay)}{h}{i}{x}
          \end{array}\right)
     \implies \EventStoredInColumnNew{c}{\tau+1}{\tau+\SyncDelay+2}{h}{i}{x}
\end{align*}
\else
\begin{align*}
                    & \EventStoredInColumnNew{c}{\tau}{\tau + \SyncDelay+1}{h}{i}{x}\\
            {}\land{}   & \EventStoredAndRetain{c}{(\tau+2)}{(\tau+\SyncDelay)}{h}{i}{x}\\
     {}\implies{} &\EventStoredInColumnNew{c}{\tau+1}{\tau+\SyncDelay+2}{h}{i}{x}
\end{align*}
\fi
\end{lemma}
\begin{proof}
Let $\party$ be some honest party with $\EventStaysInSubnet{\subnetid_c}{\party}{(\tau+1)}{(\tau + 2 + \SyncDelay)}$.
Since honest parties only join their own column, this implies $\AlgCol(\party) = c$.
We need to show $\EventInSymbolStorage{\party}{\tau+2+\SyncDelay}{h}{i}{x}$.
In the case that we have $\EventIsInSubnet{\subnetid_c}{\party}{\tau}$, we get $\EventInSymbolStorage{\party}{\tau+\SyncDelay+1}{h}{i}{x}$ by assumption and the claim follows from \cref{lem:StaysInSymbolStorage}.

So we may assume that $\EventIsInSubnet{\subnetid_c}{\party}{\tau}$ does \emph{not} hold.
By \cref{cor:FullyJoinedImpliesSubnets}, this means that $\EventFullyJoined{\party}{\tau}$ does not hold, whereas $\EventFullyJoined{\party}{\tau+1}$ does hold.
In particular, $\party$ cannot be one of the initial parties and since $\ourprot.\iJoin$ takes exactly $2+\SubnetDelay$ time slots, $\party$ must have joined via $\ourprot.\iJoin$ exactly at time $\tau-\SubnetDelay-1 > 0$.
Looking at the definition of $\iJoin$, we know that in time slot $\tau+1$, the party $\party$ calls $\ColParties_c \coloneqq \subnetprot.\iGetPeers(\subnetid_c)$.

Because $\EventStoredAndRetain{c}{(\tau+2)}{(\tau+\SyncDelay)}{h}{i}{x}$ holds, there exists some honest party $\party'$ with $\EventStaysInSubnet{\subnetid_c}{\party'}{(\tau+1)}{(\tau+\SyncDelay)}$ and $\EventInSymbolStorage{\party'}{\tau+2}{h}{i}{x}$.
\ifnum\ccs=0
Since $\EventSubnetprotGood{T+2}$, we deduce that $\party'\in\ColParties_c$.
\else
Since we know that $\EventSubnetprotGood{T+2}$ holds, we deduce that $\party'\in\ColParties_c$.
\fi
Looking further at the definition of $\iJoin$, $\party$ will send $\msgSync$ to $\party'$ in round $\tau+1$, to be received in round $\tau+2$.
We know that $\party'$ is still active at time $\tau+2$. 
Our protocol specifies that $\party'$ should then send back some $(\msgSyncRsp,\mathbf{S})$ response to $\party$ at time $\tau+2+(\SyncDelay-2)$ or earlier.
Since $\party'$ is still active at time $\tau+\SyncDelay$, it will actually send this message.
Since $\EventInSymbolStorage{\party'}{\tau+2}{h}{i}{x}$, we have $(h,i,x)\in\mathbf{S}$.
The $(\msgSyncRsp,\mathbf{S})$ response is received by $\party$ at time $\tau+\SyncDelay + 1$ or earlier, upon which $\party$ stores the symbol at time $\tau+\SyncDelay +1$ or earlier.
Consequently, $\EventInSymbolStorage{\party}{\tau+\SyncDelay+2}{h}{i}{x}$ holds by \cref{lem:StaysInSymbolStorage}, finishing the proof.
\end{proof}

\noindent
We now combine \cref{lemma:store_works}, \cref{lemma:retaining_data} and \cref{prop:relevant_events}. 

\begin{lemma}[Store works long-term]\label{lemma:long_term_storage}
Let $\JoinLeaveSchedule$ be a join-leave schedule that respects bootstrap nodes and uses good bootstrap nodes and assume $\rdspred$ is position-binding. Let $T$ be some lifetime bound and consider a run of our protocol $\ourprot$.
Let $x\in\SymbolSpace, i\in [\FileLength], h\in\HandleSpace$ for which $\rdspred(h,i,x)$ holds and let $c \coloneqq\AlgGetColForSymbol(i)$.
Assume some honest party $\PartyStoring$ with $\AlgRow(\PartyStoring)=r$ stored the symbol at some time $\tauStore \leq T-2$, i.e.\ assume $\EventFullyJoined{\PartyStoring}{\tauStore}$ and $\EventStored{\PartyStoring}{\tauStore}{h}{i}{x}$ hold.
Further, assume that the following good events hold:
\begin{myitemize}
\item $\EventGoodCell{r}{c}{\tauStore}{(\tauStore+1)}$
\item $\EventSubnetprotGood{T+1}$
\item $\EventColumnGoodUntil{c}{T+1}{\Delta}$ with $\Delta \geq \max\{2\SubnetDelay+2,2\SyncDelay+\SubnetDelay +2\}$  
\end{myitemize}
Then we have $\EventStoredInColumnNew{c}{\tauStore+1}{\tauStore+3}{h}{i}{x}$ and we have
\[
 \EventStoredInColumnNew{c}{\tau}{(\tau+\SyncDelay+1)}{h}{i}{x}
\]
for all $\tauStore+1 \leq \tau \leq T$.
\end{lemma}
\begin{proof}
First, we get $\EventStoredInColumnNew{c}{\tauStore+1}{\tauStore+3}{h}{i}{x}$ by \cref{lemma:store_works}, noting that the conditions of this lemma are satisfied.
We now show the main claim by induction over $\tau$:

\smallskip\noindent
We start the induction at $\tau = \tauStore+1$. Here, the claim follows from $\EventStoredInColumnNew{c}{\tauStore+1}{\tauStore+3}{h}{i}{x}$ and \cref{prop:relevant_events}\ref{propitem:EventStoreInColumnExtends}.

\smallskip\noindent
\ifnum\ccs=0
For the induction step with $\tau\geq \tauStore+2$, we assume that $\EventStoredInColumnNew{c}{\tau'}{(\tau'+\SyncDelay+1)}{h}{i}{x}$ holds for all $\tauStore +1 \leq \tau' < \tau$.
\else
We now procceed with induction step, where we may assume that $\tau\geq \tauStore+2$ and that $\EventStoredInColumnNew{c}{\tau'}{(\tau'+\SyncDelay+1)}{h}{i}{x}$ holds for all $\tauStore +1 \leq \tau' < \tau$.
\fi
Our first claim is that $\EventStoredAndRetain{c}{(\tau+1)}{(\tau+\SyncDelay-1)}{h}{i}{x}$ holds, for which we need to make a case distinction, depending on whether $\tau \leq \tauStore+\SyncDelay$ or not.

\smallskip\noindent
In the case $\tau \leq \tauStore+\SyncDelay$, we use $\EventStoredInColumnNew{c}{\tauStore+1}{\tauStore+3}{h}{i}{x}$.
\ifnum\ccs=0
By \cref{prop:relevant_events}\ref{propitem:EventStoreInColumnExtends}, this implies
$\EventStoredInColumnNew{c}{\tauStore+1}{\tau+1}{h}{i}{x}$, because $\tau \geq \tauStore + 2$.
\else
This implies
$\EventStoredInColumnNew{c}{\tauStore+1}{\tau+1}{h}{i}{x}$ using \cref{prop:relevant_events}\ref{propitem:EventStoreInColumnExtends}, because $\tau \geq \tauStore + 2$.
\fi
We want to use \cref{prop:relevant_events}\ref{propitem:StoredAndRetain} (with $\tau_0 = \tauStore+1,\tau_1 = \tau+1$).
We know that $\EventColumnGoodUntil{c}{T}{\Delta}$ by assumption with $\Delta \geq 2\SubnetDelay +2$.
Set $\Delta' \coloneqq \Delta - \SubnetDelay - 2$, so $\Delta' \geq 2\SyncDelay$.
Observe that $\tauStore+1\leq T+1$ and
\[
\tauStore+1+\Delta'-1 = \tauStore+ \Delta' \geq \tauStore + 2\SyncDelay \geq \tau+\SyncDelay \geq \tau + 1\enspace.
\]
This means the conditions of \cref{prop:relevant_events}\ref{propitem:StoredAndRetain} are satisfied and we obtain
$\EventStoredAndRetain{c}{(\tau+1)}{(\tauStore+\Delta')}{h}{i}{x}$.
We showed above that $\tauStore+\Delta' \geq \tau+\SyncDelay$, so we can weaken this to $\EventStoredAndRetain{c}{(\tau+1)}{(\tau+\SyncDelay-1)}{h}{i}{x}$, as required.

\smallskip\noindent
So consider the other case $\tau \geq \tauStore+\SyncDelay+1$. So $\tau-\SyncDelay \geq \tauStore+1$. By induction hypothesis (with $\tau'=\tau-\SyncDelay$), we know that
$\EventStoredInColumnNew{c}{\tau-\SyncDelay}{\tau+1}{h}{i}{x}$ holds.
Again, we want to use \cref{prop:relevant_events}\ref{propitem:StoredAndRetain} (with $\tau_0 = \tau-\SyncDelay$, $\tau_1 = \tau+1$).
We have $\tau -\SyncDelay \leq T+1$ and $\tau-\SyncDelay + \Delta' - 1 \geq \tau + \SyncDelay - 1 \geq \tau +1$, so the conditions are satisfied and we obtain
$\EventStoredAndRetain{c}{(\tau+1)}{(\tau-\SyncDelay + \Delta' - 1)}{h}{i}{x}$. Since $\Delta' \geq 2\SyncDelay$, we can weaken this to
$\EventStoredAndRetain{c}{(\tau+1)}{(\tau+\SyncDelay - 1)}{h}{i}{x}$. So we showed this first claim in both cases.

\smallskip\noindent
\ifnum\ccs=0
By induction hypothesis (with $\tau' = \tau-1$), we know that $\EventStoredInColumnNew{c}{\tau-1}{\tau+\SyncDelay}{h}{i}{x}$ holds.
Also, we know $\EventColumnGoodUntil{c}{T+1}{2\SubnetDelay+2}$, $\EventSubnetprotGood{T+1}$ by assumption.
\else
We know $\EventStoredInColumnNew{c}{\tau-1}{\tau+\SyncDelay}{h}{i}{x}$ by induction hypothesis (with $\tau' = \tau-1$).
Also, $\EventColumnGoodUntil{c}{T+1}{2\SubnetDelay+2}$ and $\EventSubnetprotGood{T+1}$ hold by assumption.
\fi

Combining this with $\EventStoredAndRetain{c}{(\tau+1)}{(\tau+\SyncDelay-1)}{h}{i}{x}$ allows us to use \cref{lemma:retaining_data} (with $\tau = \tau-1, T = T-1$) to finally obtain
$\EventStoredInColumnNew{c}{\tau}{\tau+\SyncDelay +1}{h}{i}{x}$, finishing the proof.
\end{proof}

\superparagraph{Retrieving data works.} We finally show that $\ourprot.\iGet$ allows us to retrieve stored data.

\begin{lemma}[Get works]\label{lemma:get_works}
Let $\JoinLeaveSchedule$ be a join-leave schedule that respects bootstrap nodes and uses good bootstrap nodes. Assume $\rdspred$ is position-binding and consider a run of $\ourprot$.
Let $x\in\SymbolSpace, i\in [\FileLength], h\in\HandleSpace$ for which $\rdspred(h,i,x)$ holds and let $c \coloneqq\AlgGetColForSymbol(i)$.
Assume that $\EventSubnetprotGood{T}$ and $\EventColumnGoodUntil{c}{T}{2\SubnetDelay+2}$ hold for some lifetime $T$.
Let $\PartyGetting$ be some fully joined honest party in row $r=\AlgRow(\PartyGetting)$ that tries to retrieve the symbol at time $0 \leq \tauGet \leq T$, i.e.,
$\EventFullyJoinedDuration{\PartyGetting}{\tauGet}{(\tauGet+2)}$ and $\EventCalledGet{\PartyGetting}{\tauGet}{h}{i}$ hold.
Then if any one of the three options
\ifnum\ccs=0
\begin{myitemize}
\item $\tauGet \geq \SyncDelay$, $\EventGoodCell{r}{c}{(\tauGet-\SyncDelay)}{(\tauGet+1)}$ and $\EventStoredInColumnNew{c}{\tauGet-\SyncDelay}{\tauGet+1}{h}{i}{x}$, or
\item $\tauGet \geq \SyncDelay$, $\EventGoodCell{r}{c}{(\tauGet-\SyncDelay)}{(\tauGet+1)}$ and $\EventStoredInColumnNew{c}{\tau'}{\tau'+2}{h}{i}{x}$ for some $\tauGet-\SyncDelay \leq \tau' \leq \tauGet-1$, or
\item $\tauGet < \SyncDelay$, $\EventGoodCell{r}{c}{0}{(\tauGet+1)}$ and $\EventStoredInColumnNew{c}{\tau'}{\tau'+2}{h}{i}{x}$ for some $0\leq \tau' \leq \tauGet-1$
\end{myitemize}
\else
\begin{myitemize}
\item $\tauGet \geq \SyncDelay$, $\EventGoodCell{r}{c}{(\tauGet-\SyncDelay)}{(\tauGet+1)}$ and\\
    $\EventStoredInColumnNew{c}{\tauGet-\SyncDelay}{\tauGet+1}{h}{i}{x}$
\item $\tauGet \geq \SyncDelay$, $\EventGoodCell{r}{c}{(\tauGet-\SyncDelay)}{(\tauGet+1)}$ and\\
    $\EventStoredInColumnNew{c}{\tau'}{\tau'+2}{h}{i}{x}$ for any $\tauGet-\SyncDelay \leq \tau' \leq \tauGet-1$
\item $\tauGet < \SyncDelay$, $\EventGoodCell{r}{c}{0}{(\tauGet+1)}$ and\\
    $\EventStoredInColumnNew{c}{\tau'}{\tau'+2}{h}{i}{x}$ for any$0\leq \tau' \leq \tauGet-1$
\end{myitemize}
\fi
holds, we get
\[
 \EventGotResult{\PartyGetting}{\tauGet}{h}{i}{\GetDelay}{x}
\]
with $\GetDelay = 2$.
\end{lemma}
\begin{proof}
By \cref{cor:GoodCellImpliesFindInCell}, there exists some honest party $\party$ with $\AlgCell(\party) = (r,c)$, $\EventIsInSubnet{\subnetid_r}{\party}{\tauGet}$ and $\EventIsInSubnet{\subnetid_c}{\party}{\tauGet}$
that stays active until (at least) $\tauGet+1$.
Further, $\EventIsInSubnet{\subnetid_c}{\party}{\max\{0,\tauGet-\SyncDelay\}}$.
\ifnum\ccs=1
Observe that in either of the three cases, we have $\EventStoredInColumnNew{c}{\tau_0}{\tau_1}{h}{i}{x}$ with
\else
In either of the three cases, we have $\EventStoredInColumnNew{c}{\tau_0}{\tau_1}{h}{i}{x}$ with
\fi
$\max\{0,\tauGet-\SyncDelay\} \leq \tau_0 \leq \tau_1 \leq \tauGet+1$ for some $\tau_0,\tau_1$.
\ifnum\ccs=0
This implies $\EventInSymbolStorage{\party}{\tauGet+1}{h}{i}{x}$ by \cref{prop:relevant_events}\ref{propitem:EventStoreInColumnExtends} and the definition of $\EventStoredInColumnNoArgs$.
\else
By \cref{prop:relevant_events}\ref{propitem:EventStoreInColumnExtends} and the definition of $\EventStoredInColumnNoArgs$,
we get $\EventInSymbolStorage{\party}{\tauGet+1}{h}{i}{x}$.
\fi

Now, $\PartyGetting$ calls $\ourprot.\iGet(h,i)$ at time $\tauGet$. As part of this call, it retrieves $\mathbf{P} \coloneqq \AlgGetPeersInCell(r,c)$.
Since $\EventSubnetprotGood{T}$ holds with $\tauGet \leq T$, we have $\party\in\mathbf{P}$, again by \cref{cor:GoodCellImpliesFindInCell}.

Then, $\PartyGetting$ sends $(\msgGet,h,i)$ to $\party$, which is retrieved by $\party$ at time $\tauGet+1$, where we are guaranteed that $\party$ is still active.
\ifnum\ccs=0
Since $\EventInSymbolStorage{\party}{\tauGet+1}{h}{i}{x}$ holds, $\party$ will respond with $(\msgGetRsp,h,i,x)$.
\else
Since $\EventInSymbolStorage{\party}{\tauGet+1}{h}{i}{x}$ holds, $\party$ actually stores the data and will respond with $(\msgGetRsp,h,i,x)$.
\fi
This response is then retrieved by $\PartyGetting$ at time $\tauGet+2$ (where it is still active by assumption), and $\PartyGetting$ outputs $x$, finishing the proof.
\end{proof}

\superparagraph{Retrieving Stored Data.} We now combine the above lemmas to show that we can retrieve data that we stored, i.e.\ we prove that the property we want for robustness holds, provided good events hold.
\begin{lemma}\label{lemma:ConditionalRobustness}
Let $\JoinLeaveSchedule$ be a join-leave schedule that respects bootstrap nodes and uses good bootstrap nodes. Assume $\rdspred$ is position-binding and consider a run of $\ourprot$.
Let $x\in\SymbolSpace, i\in [\FileLength], h\in\HandleSpace$ for which $\rdspred(h,i,x)$ holds, let $T$ be some lifetime and let $c \coloneqq\AlgGetColForSymbol(i)$.
\ifnum\ccs=0
Let $\PartyGetting$ and $\PartyStoring$ be honest parties with $\EventStored{\PartyStoring}{\tauStore}{h}{i}{x}$, $\EventFullyJoined{\PartyStoring}{\tauStore}$, $\EventCalledGet{\PartyGetting}{\tauGet}{h}{i}$ and $\EventFullyJoinedDuration{\PartyGetting}{\tauGet}{(\tauGet+2)}$,
where $0\leq \tauStore$ and $\tauStore+2\leq \tauGet \leq T$.
\else
Consider times $0\leq \tauStore$ and $\tauStore+2\leq \tauGet \leq T$.
Let $\PartyGetting$ be an honest party for which $\EventStored{\PartyStoring}{\tauStore}{h}{i}{x}$ and $\EventFullyJoined{\PartyStoring}{\tauStore}$ hold.
Similarly, let $\PartyStoring$ be an honest party for which $\EventCalledGet{\PartyGetting}{\tauGet}{h}{i}$ and $\EventFullyJoinedDuration{\PartyGetting}{\tauGet}{(\tauGet+2)}$ hold.
\fi
Set $\rStore\coloneqq \AlgRow(\PartyStoring)$ and $\rGet \coloneqq \AlgRow(\PartyGetting)$
Then, if all of
\begin{myitemize}
    \item $\EventSubnetprotGood{T+1}$ and
    \item $\EventColumnGoodUntil{c}{T+1}{\Delta}$, where $\Delta \geq \max\{2\SubnetDelay + 2, 2\SyncDelay + \SubnetDelay + 2\}$ and
    \item $\EventGoodCell{\rStore}{c}{\tauStore}{(\tauStore+1)}$ and
    \item $\EventGoodCell{\rGet}{c}{\max\{0,\tauGet - \SyncDelay\}}{(\tauGet + 1)}$
\end{myitemize}
hold, we have $\EventGotResult{\PartyGetting}{\tauGet}{h}{i}{\GetDelay}{x}$ with $\GetDelay = 2$.
\end{lemma}
\begin{proof}
Observe that the conditions of \cref{lemma:long_term_storage} are satisfied, so we have
\ifnum\ccs=0
\begin{align*}
  & \EventStoredInColumnNew{c}{\tauStore+1}{\tauStore+3}{h}{i}{x} \qquad \text {and}\\
  & \EventStoredInColumnNew{c}{\tau}{(\tau+\SyncDelay+1)}{h}{i}{x} \ \text{ for all }\tauStore+1 \leq \tau \leq T\enspace.
\end{align*}
\else
\begin{align*}
  & \EventStoredInColumnNew{c}{\tauStore+1}{\tauStore+3}{h}{i}{x} \qquad \text {and}\\
  & \EventStoredInColumnNew{c}{\tau}{(\tau+\SyncDelay+1)}{h}{i}{x} \forall \tauStore+1 \leq \tau \leq T\enspace.
\end{align*}
\fi
We now show that the conditions of \cref{lemma:get_works} are satisfied, for which we distinguish 3 cases that match the conditions of \cref{lemma:get_works}.

\smallskip\noindent
Assume first that $\tauGet < \SyncDelay$. Then $\EventStoredInColumnNew{c}{\tau'}{\tau'+2}{h}{i}{x}$ holds for some $0\leq \tau' \leq \tauGet - 1$ by setting $\tau' = \tauStore + 1$. So \cref{lemma:get_works} is applicable in that case.

\smallskip\noindent
Otherwise, $\tauGet \geq \SyncDelay$. Assume first that $\tauGet \geq \tauStore + \SyncDelay + 1$, so $\tauStore + 1 \leq \tauGet - \SyncDelay \leq T$.
\ifnum\ccs=0
Then we know that $\EventStoredInColumnNew{c}{(\tauGet - \SyncDelay)}{(\tauGet+1)}{h}{i}{x}$ holds, which shows that
\cref{lemma:get_works} is applicable.
\else
Then $\EventStoredInColumnNew{c}{(\tauGet - \SyncDelay)}{(\tauGet+1)}{h}{i}{x}$ holds, which shows that
\cref{lemma:get_works} is applicable.
\fi

\smallskip\noindent
The remaining case is $\tauGet \geq \SyncDelay$ and $\tauGet \leq \tauStore + \SyncDelay$. Here, we have $\EventStoredInColumnNew{c}{\tau'}{\tau'+2}{h}{i}{x}$ for some $\tauGet - \SyncDelay \leq \tau' \leq \tauGet - 1$ by setting $\tau' = \tauStore + 1$. So again, \cref{lemma:get_works} is applicable.

\smallskip\noindent
By applying \cref{lemma:get_works}, we thereby obtain $\EventGotResult{\PartyGetting}{\tauGet}{h}{i}{2}{x}$ in every case.
\end{proof}

\ifnum\ccs=0
\subsubsection{Good Events Occur for Admissible Schedules}
\else
\subsection{Good Events Occur for Admissible Schedules}
\fi

In the second part of our proof, we show that the good events that appeared as conditions in \cref{lemma:ConditionalRobustness} hold true with sufficiently high probability, provided we have an admissible schedule.


\begin{lemma}[Probability of Bad Columns]\label{lemma:prob_bad_columns}
Let $\JoinLeaveSchedule$ be a join-leave schedule that guarantees $N$ honest parties with overlap $\OverlapTime$. Then, for any column $c\in[k_2]$, lifetime bound $T$ and $\Delta \leq \OverlapTime$, we have
\[
 \Prob{ \EventColumnGoodUntil{c}{T}{\Delta} } \geq 1 - \Bigl\lceil\frac{T+1}{\OverlapTime-\Delta-1}\Bigr\rceil \cdot e^{-\frac{N}{k_2}}\enspace,
\]
where the randomness is over the choice of the random oracle.
Further,
\ifnum\ccs=0
\[
 \Prob{ \EventColumnGoodUntil{c}{T}{\Delta} \text{ for all $c\in [k_2]$}} \geq 1 - \Bigl\lceil\frac{T+1}{\OverlapTime-\Delta-1}\Bigr\rceil \cdot k_2\cdot e^{-\frac{N}{k_2}}\enspace.
\]
\else
\begin{align*}
 &\Prob{ \EventColumnGoodUntil{c}{T}{\Delta} \text{ for all $c\in [k_2]$}}\\ {}\geq{} &1 - \Bigl\lceil\frac{T+1}{\OverlapTime-\Delta-1}\Bigr\rceil \cdot k_2\cdot e^{-\frac{N}{k_2}}\enspace.
\end{align*}
\fi

\end{lemma}
\begin{proof}
First, we note that the join-leave schedule and the random oracle fully determine whether $\EventColumnGoodUntil{c}{T}{\Delta}$ holds or not.
For any time slots $\tau$, let $A_\tau$ denote the event that there exist no honest party $\party$ with $\EventActiveDuration{\party}{\tau}{\tau+\Delta}$ and $\AlgCol(\party) = c$.
By definition, $\EventColumnGoodUntil{c}{T}{\Delta}$ holds iff none of the $T+1$ many events $A_0,\ldots, A_T$ holds.

Consider any given $\tau$. By assumption on the join-leave schedule, we have $N$ honest parties $\party$ with $\EventActiveDuration{\party}{\tau}{(\tau+\OverlapTime)}$.
The probability (over the choice of random oracle) that none of these $N$ parties is assigned to column $c\in[k_2]$ is given by
\[
 \Prob{\text{None of the $N$ parties in column $c$}} = \Bigl(1-\frac{1}{k_2}\Bigr)^{N} < e^{-\frac{N}{k_2}}\enspace.
\]
If any one of those $N$ parties is assigned to column $c$, then none of the $\OverlapTime-\Delta+1$ many events $A_\tau,\ldots, A_{\tau + \OverlapTime - \Delta}$ holds, so we have
\[
 \Prob{A_\tau \lor \ldots \lor A_{\tau+\OverlapTime - \Delta}}  < e^{-\frac{N}{k_2}}\enspace.
\]
We get the first result by dividing $\{0,\ldots,T\}$ into $({T+1})/({\OverlapTime-\Delta+1})$ intervals of length $\OverlapTime - \Delta + 1$ and taking a union bound over them. The claim is obtained by taking a union bound over the columns.
\end{proof}

\begin{lemma}[Probability of Bad Cells]\label{lemma:prob_bad_cells}
Let $\JoinLeaveSchedule$ be a join-leave schedule that guarantees $N$ honest parties with overlap $\OverlapTime$.
Let $T$ be a lifetime bound, $0 \leq \varepsilon \leq 1$ be some fraction and consider some $\Delta$ with $\Delta + \SubnetDelay \leq \OverlapTime$.
Let $\EventSmallCorruption{\varepsilon}{T}{\Delta}$ be the event that for every row $r$ and time $0\leq \tau \leq T$, $\EventGoodCell{r}{c}{\tau}{(\tau+\Delta)}$ holds for all but $\varepsilon\cdot k_2$ many columns $c$.

Then we have
\ifnum\ccs=0
\begin{gather*}
    \Prob{\EventSmallCorruption{\varepsilon}{T}{\Delta}} \geq 1 - \Bigl\lceil\frac{T+1}{\OverlapTime - \Delta-\SubnetDelay +1}\Bigr\rceil\cdot 2^{\binaryEntropy(\varepsilon)k_2} \cdot k_1 \cdot e^{-\frac{\varepsilon N}{k_1}}\enspace,
\end{gather*}
\else
\begin{align*}
    &\Prob{\EventSmallCorruption{\varepsilon}{T}{\Delta}}\\ {}\geq{}& 1 - \Bigl\lceil\frac{T+1}{\OverlapTime - \Delta-\SubnetDelay +1}\Bigr\rceil\cdot 2^{\binaryEntropy(\varepsilon)k_2} \cdot k_1 \cdot e^{-\frac{\varepsilon N}{k_1}}\enspace,
\end{align*}
\fi
where $\binaryEntropy(\varepsilon) = -\varepsilon\log_2(\varepsilon) - (1-\varepsilon)\log_2(1-\varepsilon)$ is the binary entropy function and the randomness is over the choice of the random oracle.

\end{lemma}
\begin{proof}
Again, we note that the join-leave schedule and the random oracle fully determine whether any $\EventGoodCellNoArgs$ events hold or not.
Consider some row $r\in [k_1]$ and consider some time slot $\tau$. By assumption, we are guaranteed that there exists a set $\mathbf{P}$ of least $N$ honest parties $\party$ with $\EventActiveDuration{\party}{\max\{0,\tau-2-\SubnetDelay\}}{(\max\{0,\tau - 2 - \SubnetDelay\} + \OverlapTime)}$.

For any given subset $C\subset [k_2]$ of the columns with size $\abs{C} = \varepsilon k_2$, the probability that none of $\mathbf{P}$ is assigned to $(r,c)$ with $c\in C$ is
\[
 \Prob{ \text{no }\party\in\mathbf{P}\text{ has }\AlgCell(\party)\in\{r\}\times C } = \Bigl( 1 - \frac{\varepsilon k_2}{k_1 k_2}\Bigr)^N\enspace.
\]
Taking a union bound over all $\binom{k_2}{\varepsilon k_2}$ possible choices of $C$ gives
\ifnum\ccs=0
\[
 \Prob{ \Bigl| \{c \text{ with } \AlgCell(\party)=(r,c) \text { for some }\party\in\mathbf{P} \}\Bigr| < (1-\varepsilon) k_2} \leq \binom{k_2}{\varepsilon k_2}\Bigl(1-\frac{\varepsilon}{k_1}\Bigr)^N\enspace.
\]
\else
\begin{align*}
 & \Prob{ \Bigl| \{c \text{ with } \AlgCell(\party)=(r,c) \text { for some }\party\in\mathbf{P} \}\Bigr| < (1-\varepsilon) k_2} \\ {}\leq{}& \binom{k_2}{\varepsilon k_2}\Bigl(1-\frac{\varepsilon}{k_1}\Bigr)^N\enspace.
\end{align*}
\fi
Since $\binom{k_2}{\varepsilon k_2} \leq 2 ^{h(\varepsilon)k_2}$ and $(1-\frac{\varepsilon}{k_1})^N < e^{-\frac{\varepsilon N}{k_1}}$ by \cref{lemma:exponential_decay} and \cref{lemma:binomial_entropy}, we get
\ifnum\ccs=0
\[
 \Prob{ \Bigl| \{c \text{ with } \AlgCell(\party)=(r,c) \text { for some }\party\in\mathbf{P} \}\Bigr| < (1-\varepsilon) k_2} \leq 2^{h(\varepsilon)k_2} e^{-\frac{\varepsilon N}{k_1}}\enspace.
\]
\else
\begin{align*}
 & \Prob{ \Bigl| \{c \text{ with } \AlgCell(\party)=(r,c) \text { for some }\party\in\mathbf{P} \}\Bigr| < (1-\varepsilon) k_2}\\ {}\leq{}& 2^{h(\varepsilon)k_2} e^{-\frac{\varepsilon N}{k_1}}\enspace.
\end{align*}
\fi
By definition, if $\bigl|\{c \text{ with } \AlgCell(\party)=(r,c) \text { for some }\party\in\mathbf{P} \}\bigr| \geq (1-\varepsilon) k_2$, then $\EventGoodCell{r}{c}{\tau}{\tau+\Delta}$ must hold for all except at most $\varepsilon k_2$ many $c\in [k_2]$.
In fact, analogous to the argument in \cref{lemma:prob_bad_cells}, $\EventGoodCell{r}{c}{\tau'}{\tau'+\Delta}$ holds for the $\OverlapTime - \SubnetDelay - \Delta + 1$ many consecutive times $\tau'\in\{\tau,\ldots,\tau + \OverlapTime - \SubnetDelay - \Delta\}$.
Splitting the $T+1$ time slots from $0$ to $T$ into such interval of length $\OverlapTime - \SubnetDelay - \Delta +1$ and taking a union bound and also taking a union bound over the possible choices of $r$ then gives the claim.
\end{proof}

\ifnum\ccs=0
\subsubsection{Proof of \cref{theorem:ourprot:analysis:maintheorem}}\label{subsection:proof_of_main_theorem}
\else
\subsection{Proof of \cref{theorem:ourprot:analysis:maintheorem}}\label{subsection:proof_of_main_theorem}
\fi

We now finally piece together the lemmas from above to prove \cref{theorem:ourprot:analysis:maintheorem}.
\begin{proof}[Proof of~\cref{theorem:ourprot:analysis:maintheorem}]
First, we need to bound the sizes of the sets $\Corrupted_{\party}^{\tau}$.
Since $\AlgGetColForSymbol$ distributes the symbols evenly among the columns and $k_2$ divides $m$, we have
\[
 \frac{\abs{\Corrupted_{\party}^{\tau}}}{m}
 =
 \frac{\abs{\{ c \mid \neg \EventGoodCell{r}{c}{\max\{0,\tau-\SyncDelay\}}{(\tau+1)}\}}}{k_2}\enspace,
\]
where $r=\AlgRow(\party)$.
So by \cref{lemma:prob_bad_cells}, we have $\abs{\Corrupted_{\party}^{\tau}} \leq \RelPosCorrThresh m$ for all honest parties $\party$ and all $0\leq \tau \leq T$, except with probability at most
\ifnum\ccs=0 
\begin{align*}
 \Prob{\exists0\leq \tau \leq T, \text{ honest }\party \text{ with }\abs{\Corrupted_{\party}^{\tau}} > \RelPosCorrThresh m}
  & <
     \Bigl\lceil\frac{T+1}{\OverlapTime - (\SyncDelay + 1) - \SubnetDelay + 1}\Bigr\rceil 2^{h(\varepsilon)k_2}k_1e^{-\frac{\varepsilon N}{k_1}}
\end{align*}
\else
\begin{align*}
 & \Prob{\exists0\leq \tau \leq T, \text{ honest }\party \text{ with }\abs{\Corrupted_{\party}^{\tau}} > \RelPosCorrThresh m}\\
  {}<{} &
     \Bigl\lceil\frac{T+1}{\OverlapTime - (\SyncDelay + 1) - \SubnetDelay + 1}\Bigr\rceil 2^{h(\varepsilon)k_2}k_1e^{-\frac{\varepsilon N}{k_1}}
\end{align*}
\fi

Second, we need to show the correctness if not corrupted property, for which we use \cref{lemma:ConditionalRobustness}.
Let $\rStore = \AlgRow(\PartyStoring), \rGet = \AlgRow(\PartyGetting)$ and $c = \AlgGetColForSymbol(i)$. Let us argue that the conditions of this lemma are satisfied except with small probability:
\begin{myitemize}
\item \ifnum\ccs=0
        By definition, we have $\Prob{\EventSubnetprotGood{T+1}} = \Prob{\EventSubnetprotGood{\ProtocolLifetimeSubnet}} \geq 1 - \RobustnessErrorSubnet$.
      \else
        By definition of $T$, we have $\Prob{\EventSubnetprotGood{T+1}} =\allowbreak \Prob{\EventSubnetprotGood{\ProtocolLifetimeSubnet}} \geq 1 - \RobustnessErrorSubnet$.
      \fi
\item By \cref{lemma:prob_bad_columns}, we have
    \ifnum\ccs=0
    \begin{align*}
        \Prob{ \EventColumnGoodUntil{c}{T+1}{\OverlapTimeMin} \text { for all } c}
                \geq
                1-\Bigl\lceil \frac{T+2}{1+\OverlapTime - \OverlapTimeMin} \Bigr\rceil \cdot k_2e^{-\frac{N}{k_2}}
    \end{align*}
    \else
    \begin{align*}
        &\Prob{ \EventColumnGoodUntil{c}{T+1}{\OverlapTimeMin} \text { for all } c}\\
                {}\geq{}&
                1-\Bigl\lceil \frac{T+2}{1+\OverlapTime - \OverlapTimeMin} \Bigr\rceil \cdot k_2e^{-\frac{N}{k_2}}
    \end{align*}
    \fi
\item   \ifnum\ccs=0
        By definition, since $i\notin \Corrupted_{\PartyStoring}^{\tauStore}$, we have $\EventGoodCell{\rStore}{c}{\max\{0, \tauStore-\SyncDelay\}}{(\tauStore + 1)}$, which implies $\EventGoodCell{\rStore}{c}{\tauStore}{(\tauStore+1)}$.
        \else
        By definition of $\Corrupted_{\PartyStoring}^{\tauStore}$ and since $i\notin \Corrupted_{\PartyStoring}^{\tauStore}$, we have $\EventGoodCell{\rStore}{c}{\max\{0, \tauStore-\SyncDelay\}}{(\tauStore + 1)}$, which implies $\EventGoodCell{\rStore}{c}{\tauStore}{(\tauStore+1)}$.
        \fi
\item   \ifnum\ccs=0
        By definition, since $i\notin \Corrupted_{\PartyGetting}^{\tauGet}$, we have $\EventGoodCell{\rGet}{c}{\max\{0, \tauGet-\SyncDelay\}}{(\tauGet + 1)}$
        \else
        By definition of $\Corrupted_{\PartyGetting}^{\tauGet}$ and since we have $i\notin \Corrupted_{\PartyGetting}^{\tauGet}$, we obtain $\EventGoodCell{\rGet}{c}{\max\{0, \tauGet-\SyncDelay\}}{(\tauGet + 1)}$
        \fi
\end{myitemize}
This means we can apply \cref{lemma:ConditionalRobustness}, which directly gives us the correctness if not corrupted property with error $\RobustnessError$ bounded by
\ifnum\ccs=0
\begin{align*}
 \RobustnessError &< \RobustnessErrorSubnet
 +
 \Bigl\lceil\frac{T+1}{\OverlapTime - \SyncDelay - \SubnetDelay}\Bigr\rceil 2^{h(\varepsilon)k_2}k_1e^{-\frac{\varepsilon N}{k_1}}
 +
 \Bigl\lceil \frac{T+2}{1+\OverlapTime - \OverlapTimeMin} \Bigr\rceil \cdot k_2e^{-\frac{N}{k_2}}\\
 &<
 \RobustnessErrorSubnet + \Bigl\lceil \frac{T+2}{1+\OverlapTime - \OverlapTimeMin} \Bigr\rceil \cdot
 \Bigl( 2^{h(\varepsilon)k_2}k_1e^{-\frac{\varepsilon N}{k_1}} + k_2e^{-\frac{N}{k_2}} \Bigr)\enspace,
\end{align*}
\else 
\begin{align*}
& \RobustnessError < \RobustnessErrorSubnet
 +
 \Bigl\lceil\frac{T+1}{\OverlapTime - \SyncDelay - \SubnetDelay}\Bigr\rceil 2^{h(\varepsilon)k_2}k_1e^{-\frac{\varepsilon N}{k_1}}\\
&\qquad\quad\ +
 \Bigl\lceil \frac{T+2}{1+\OverlapTime - \OverlapTimeMin} \Bigr\rceil \cdot k_2e^{-\frac{N}{k_2}}\\
& \!\! <
 \RobustnessErrorSubnet + \Bigl\lceil \frac{T+2}{1+\OverlapTime - \OverlapTimeMin} \Bigr\rceil \cdot
 \Bigl( 2^{h(\varepsilon)k_2}k_1e^{-\frac{\varepsilon N}{k_1}} + k_2e^{-\frac{N}{k_2}} \Bigr),
\end{align*}
\fi
finishing the proof.
\end{proof}

\clearpage
    \begin{figure*}
        \centering
        \begin{tikzpicture}
            \def\kOne{4}       
            \def\kTwo{5}       
            \def\rowHeight{0.3} 
            \def\colWidth{1.1} 
            \def\spacing{0.2}  
            \def\highlightCol{4} 

            \pgfmathsetmacro{\colHeight}{\kOne * (\rowHeight + \spacing)+\spacing}
            \pgfmathsetmacro{\rowWidth}{\kTwo * (\colWidth + \spacing)+\spacing}

            \begin{scope}
                \foreach \i in {1,...,\kOne} {
                    \pgfmathsetmacro{\yCoord}{-\i * (\rowHeight + \spacing)}
                    \draw[fill=gray!10]
                        (0, \yCoord)
                        rectangle (\rowWidth, \yCoord + \rowHeight);
                }

                \foreach \j in {1,...,\kTwo} {
                    \pgfmathsetmacro{\xCoord}{\j * (\colWidth + \spacing)}
                    \ifthenelse{\j=\highlightCol}{
                        \draw[draw=blue, thick] (\xCoord - \colWidth, 0) rectangle (\xCoord, -\colHeight);
                    }{
                        \draw[] (\xCoord - \colWidth, 0) rectangle (\xCoord, -\colHeight);
                    }
                    \ifthenelse{\j=4}{
                        \node[anchor=north, font=\bfseries] at (\xCoord - 0.5*\colWidth, -\colHeight - 0.1) {\parbox{1.5cm}{\centering\scriptsize store \\ symbol \j}};
                    }{
                        \node[anchor=north] at (\xCoord - 0.5*\colWidth, -\colHeight - 0.1) {\parbox{1.5cm}{\centering\scriptsize store \\ symbol \j}};
                    }
                }

                \pgfmathsetmacro{\nodeY}{-3 * (\rowHeight + \spacing) + 0.25}
                \pgfmathsetmacro{\nodeXOne}{2 * (\colWidth + \spacing) - 0.5 * \colWidth}
                \pgfmathsetmacro{\nodeXTwo}{4 * (\colWidth + \spacing) - 0.5 * \colWidth}

                \node[circle, fill=black, inner sep=2pt] (A) at (\nodeXOne, \nodeY) {};
                \node[circle, fill=black, inner sep=2pt] (B) at (\nodeXTwo, \nodeY) {};

                \draw[->, thick, bend left=20] (A) to (B);

                \draw[->, thick] (B) ++(0.15,0) to ++(0.2,0.1);

                \foreach \i in {1,...,\kOne} {
                    \ifthenelse{\i=3}{}{
                        \pgfmathsetmacro{\yCoord}{-\i * (\rowHeight + \spacing) + 0.25}
                        \draw[->, thick, bend right=20] (B) to (\nodeXTwo - 0.15, \yCoord);
                    }
                }
            \end{scope}

            \begin{scope}[xshift=8.5cm] 
                \foreach \i in {1,...,\kOne} {
                    \pgfmathsetmacro{\yCoord}{-\i * (\rowHeight + \spacing)}
                    \draw[fill=gray!10]
                        (0, \yCoord)
                        rectangle (\rowWidth, \yCoord + \rowHeight);
                }

                \foreach \j in {1,...,\kTwo} {
                    \pgfmathsetmacro{\xCoord}{\j * (\colWidth + \spacing)}
                    \ifthenelse{\j=\highlightCol}{
                        \draw[draw=blue, thick] (\xCoord - \colWidth, 0) rectangle (\xCoord, -\colHeight);
                    }{
                        \draw[] (\xCoord - \colWidth, 0) rectangle (\xCoord, -\colHeight);
                    }
                    \ifthenelse{\j=4}{
                        \node[anchor=north, font=\bfseries] at (\xCoord - 0.5*\colWidth, -\colHeight - 0.1) {\parbox{1.5cm}{\centering\scriptsize store \\ symbol \j}};
                    }{
                        \node[anchor=north] at (\xCoord - 0.5*\colWidth, -\colHeight - 0.1) {\parbox{1.5cm}{\centering\scriptsize store \\ symbol \j}};
                    }
                }

                \pgfmathsetmacro{\nodeYTop}{-1 * (\rowHeight + \spacing) + 0.25}
                \pgfmathsetmacro{\nodeXFive}{5 * (\colWidth + \spacing) - 0.5 * \colWidth}
                \pgfmathsetmacro{\nodeXFour}{4 * (\colWidth + \spacing) - 0.5 * \colWidth}

                \node[circle, fill=black, inner sep=2pt] (C) at (\nodeXFive, \nodeYTop) {};
                \node[circle, fill=black, inner sep=2pt] (D) at (\nodeXFour, \nodeYTop) {};

                \draw[->, thick, bend right=20] (C) to (D);
                \draw[->, thick, bend right=20] (D) to (C);
            \end{scope}
        \end{tikzpicture}

        \caption{Visualization of our protocol $\ourprot$ with \(k_1 = 4\) rows and \(k_2 = 5\) columns. Each row and column represents a subnet (i.e., clique) for the subnet discovery protocol $\subnetprot$ as per \cref{def:subnetdiscovery}. Parties in column $j \in [k_2]$ store the $j$th chunk of the data. For simplicity, we assume that the number of symbols is $\FileLength = k_2$ in this visualization.
        \emph{Left:} node in column subnet $2$ and row subnet $3$ storing a symbol that must be stored in column subnet $4$. First, the node sends it to \emph{all} nodes in cell $(3,4)$, then every such node forwards it to all nodes in that column. \emph{Right:} node in column subnet $5$ and row subnet $1$ querying a symbol that must be stored in column subnet $4$. First, the node sends a query to \emph{all} nodes in cell $(1,4)$. Then, all nodes respond with the symbol.}
        \label{fig:tikz_grid_store_get}
        \label{fig:tikz_grid}
    \end{figure*}
    \begin{figure*}[!htb]
    \centering
    \nicoresetlinenr
    \noindent\fbox{\parbox{0.95\textwidth}{
        \begin{minipage}[t]{0.52\textwidth}%
            \underline{\textbf{Interface} $\iInit(\mathbf{P}, \AuxRole_1,\ldots,\AuxRole_{\ell})$}
            \begin{nicodemus}
                \smallskip
                \item[]\textit{// Initialize Subnet protocol}
                \item $\subnetprot.\iInit(\mathbf{P},\AuxRole_1,\ldots,\AuxRole_{\ell})$
                \item \textbf{inherit message handlers from} $\subnetprot$
                \item \textbf{parse} $\mathbf{P}=[\party_1,\ldots,\party_{\ell} ]$
                \smallskip
                \item[]\textit{// create row subnets}
                \item $\pcfor r \in [k_1]\colon$
                \item $\pcind \RowParties_r \coloneqq \{\party_i \mid \AlgRow(\party_i) = r \lor \AuxRole_i = 1\}$
                \item $\pcind \pcif \self \in \RowParties_r\colon$
                \item $\pcind \pcind\subnetprot.\iCreateSubnet(\subnetid_r, \RowParties_r)$
                \smallskip
                \item[]\textit{// create column subnets}
                \item $c\coloneqq \AlgCol(\self)$
                \item $\pcind \ColParties_c \coloneqq \{\party_i \mid \AlgCol(\party_i)  = c\}$
                \item[]\textit{// note: $\self \in \ColParties_c$}
                \item $\subnetprot.\iCreateSubnet(\subnetid_c, \ColParties_c)$
                \smallskip
                \item[]\textit{// Initialize Symbols map}
                \item $\SymbolsMap[\cdot,\cdot] \coloneqq \emptyset$
                \item[]\textit{// Memorize whether we can act as bootstrap node}
                \item $\IsBootstrapNode\coloneqq \AuxRole_i$ \textbf{where} $\self = \party_i$.
            \end{nicodemus}
            \medskip\noindent
            \underline{\calgo $\AlgGetPeersInCell(r,c)$}
            \begin{nicodemus}
                \smallskip
                \item[]\textit{// get peers in row $r$}
                \item $\mathbf{P}' \coloneqq \subnetprot.\iGetPeers(\subnetid_r)$
                \smallskip
                \item[]\textit{// only return the ones in column $c$}
                \item $\pcreturn \mathbf{P} \coloneqq \{\party \in \mathbf{P}' \mid \AlgCell(\party) = (r,c)\}$
            \end{nicodemus}
        \end{minipage}\quad
        \begin{minipage}[t]{0.45\textwidth}%
            \underline{\calgo $\AlgCell(\party)$}
            \begin{nicodemus}
                \smallskip
                \item[]\textit{// cell in which party resides}
                \item $\pcreturn \Hash(\party)$
            \end{nicodemus}
            \medskip\noindent
            \underline{\calgo $\AlgRow(\party)$}
            \begin{nicodemus}
                \smallskip
                \item[]\textit{// row in which party resides}
                \item $(r,c) \coloneqq \AlgCell(\party)$
                \item $\pcreturn r$
            \end{nicodemus}
            \medskip\noindent
            \underline{\calgo $\AlgCol(\party)$}
            \begin{nicodemus}
                \smallskip
                \item[]\textit{// column in which party resides}
                \item $(r,c) \coloneqq \AlgCell(\party)$
                \item $\pcreturn c$
            \end{nicodemus}
            \medskip\noindent
            \underline{\calgo $\AlgGetRowSid(r)$}
            \begin{nicodemus}
                \smallskip
                \item[]\textit{// label rows from $1$ to $k_1$}
                \item $\pcreturn r$
            \end{nicodemus}
            \medskip\noindent
            \underline{\calgo $\AlgGetColSid(c)$}
            \begin{nicodemus}
                \smallskip
                \item[]\textit{// label columns from $k_1+1$ to $k_1+k_2$}
                \item $\pcreturn k_1 + c$
            \end{nicodemus}
            \medskip\noindent
            \underline{\calgo $\AlgGetColForSymbol(i)$}
            \begin{nicodemus}
                \item $\pcreturn c~\textbf{s.t.}~(c-1)\FileLength/k_2 < i \leq c\FileLength/k_2$
            \end{nicodemus}
        \end{minipage}
    }}
    \caption{Initialization and helper algorithms for our distributed array protocol $\ourprot$.
    It makes use of a subnet discovery protocol $\subnetprot$ for $\numberofsubnets = k_1 + k_2$ and a random oracle $\Hash\colon\boolstar\rightarrow[k_1]\times[k_2]$.
    We present code for joining in \cref{fig:ourprotocol:joining} and interfaces $\iStore$ and $\iGet$ in \cref{fig:ourprotocol:storeget}.}
    \label{fig:ourprotocol:initialization}
\end{figure*}

\begin{figure*}[!htb]
    \centering
    \nicoresetlinenr
    \noindent\fbox{\parbox{0.95\textwidth}{
        \begin{minipage}[t]{0.47\textwidth}%
            \underline{\textbf{Interface} $\iJoin(\{\party_1,\ldots,\party_t\},\AuxRole)$}
            \begin{nicodemus}
                \smallskip
                \item \textbf{let $\tau$ be the time when $\iJoin$ is called}
                \item $\subnetprot.\iJoin(\emptyset,\perp)$
                \item \textbf{inherit message handlers from} $\subnetprot$
                \item $(r,c) \coloneqq \AlgCell(\self)$
                \item $\SymbolsMap[\cdot,\cdot] \coloneqq \emptyset,~\IsBootstrapNode \coloneqq \AuxRole$
                \smallskip
                \item $\pcif $\IsBootstrapNode$ = 1\colon$
                \item[]$\pcind$ \textit{// join all row subnet via all bootstrap nodes}
                \item $\pcind\pcfor i \in [t], r' \in [k_1]\colon$
                \item $\pcind\pcind \subnetprot.\iJoinSubnet(\subnetid_{r'},\party_i)$
                \item $\pcelse\colon$
                \item[]$\pcind$ \textit{// join own row subnet via all bootstrap nodes}
                \item $\pcind\pcfor i \in [t]\colon~\subnetprot.\iJoinSubnet(\subnetid_r, \party_i)$.
                \smallskip
                \item[]\textit{// ask bootstrap nodes for column peers}
                \item $\ColParties_c\coloneqq \emptyset$.
                \item $\pcfor i \in [t]\colon~\textbf{send}~\msgJoin~\textbf{to}~\party_i$
                \item $\textbf{on receiving}~(\msgJoinRsp,\mathbf{P}^{(i)}_c)~\textbf{from each}~\party_i$
                \item[]\textbf{within 2 rounds, do in round $\tau+2$}$\colon$
                \item $\pcind\pcfor \party'\in\mathbf{P}^{(i)}_c\setminus\ColParties_c \colon$
                \item $\pcind\pcind\subnetprot.\iJoinSubnet(\subnetid_c,\party')$
                \item $\pcind\ColParties_c \coloneqq \ColParties_c \union \mathbf{P}^{(i)}_c$.
                \smallskip
                \item[]\textit{// Wait until we can be sure other nodes find $\self$,}
                \item[]\textit{// so relevant new data will be learnt via $\iStore$}
                \item \textbf{wait for} $\SubnetDelay$ \textbf{until} $\tau + 2 +\SubnetDelay$
                \item[]\textit{// sync old data via column subnet}
                \item[]\textit{// ask everyone for their data}
                \item $\ColParties_c \coloneqq \subnetprot.\iGetPeers(\subnetid_c)$
                \item $\pcfor \party' \in \ColParties_c\colon~\textbf{send}~\msgSync~\textbf{to}~\party'$
                \smallskip
                \item \textbf{terminate $\iJoin$}
            \end{nicodemus}
        \end{minipage}\quad
        \begin{minipage}[t]{0.52\textwidth}%
            \underline{\textbf{On Message} $\msgSync$ \textbf{from} $\party$}
            \begin{nicodemus}
                \item $\mathbf{S} \coloneqq \{(h,i,\SymbolsMap[h,i]) \mid \SymbolsMap[h,i] \neq \bot\}$
                \item[] \textit{// model that this may be slow.}
                \item[] \textit{// see \cref{rmk:SyncDelay}}
                \item \textbf{wait for arbitrary time} $T \leq \SyncDelay-2$
                \item $\textbf{send}~(\msgSyncRsp,\mathbf{S})~\textbf{to}~\party$
            \end{nicodemus}
            \medskip\noindent
            \underline{\textbf{On Message} $(\msgSyncRsp,\mathbf{S})$ from $\party$}
            \begin{nicodemus}
                \item $\pcfor (h,i,x) \in \mathbf{S}\colon$
                \item $\pcind \pcif \rdspred(h,i,x) = 1\colon~\SymbolsMap[h,i]\coloneqq x$

            \end{nicodemus}
            \medskip\noindent
            \underline{\textbf{On Message} $\msgJoin$ \textbf{from} $\party$}
            \begin{nicodemus}
                \item $\pcif \IsBootstrapNode = 0\colon$
                \item $\pcind\pcreturn$
                \item $c \coloneqq \AlgCol(\party)$
                \smallskip
                \item[]\textit{// collect all nodes we know}
                \item $\mathbf{P} \coloneqq \emptyset$
                \item $\pcfor r \in [k_1]\colon$
                \item $\pcind \mathbf{P} \coloneqq \mathbf{P} \cup \subnetprot.\iGetPeers(\subnetid_r)$
                \smallskip
                \item[]\textit{// respond with nodes in appropriate column}
                \item $\ColParties_c \coloneqq \{\party' \in \mathbf{P} \mid \AlgCol(\party') = c\}$
                \item $\textbf{send}~(\msgJoinRsp,\mathbf{P}_c)~\textbf{to}~\party$
            \end{nicodemus}
        \end{minipage}
    }}
    \caption{Joining instructions for our distributed array protocol $\ourprot$ with file space $\SymbolSpace^\FileLength$ and handle space $\HandleSpace$, for a predicate $\rdspred\colon \HandleSpace \times [\FileLength] \times \SymbolSpace \to\bool$. It makes use of a subnet discovery protocol $\subnetprot$ for $\numberofsubnets = k_1 + k_2$ and a random oracle $\Hash\colon\boolstar\to[k_1]\times[k_2]$. The party that is executing the instructions is referred to as $\self$. We present the initialization code in \cref{fig:ourprotocol:initialization}, interfaces $\iStore$ and $\iGet$ in \cref{fig:ourprotocol:storeget}.}
    \label{fig:ourprotocol:joining}
\end{figure*}

\begin{figure*}[!htb]
    \centering
    \nicoresetlinenr
    \noindent\fbox{\parbox{0.95\textwidth}{
        \begin{minipage}[t]{0.43\textwidth}%
            \underline{\textbf{Interface} $\iStore(h,i,x)$}
            \begin{nicodemus}
                \item $\pcif \rdspred(h,i,x) = 0\colon~\pcreturn$
                \smallskip
                \item[]\textit{// column in which the symbol should be stored}
                \item $c \coloneqq \AlgGetColForSymbol(i)$
                \smallskip
                \item[]\textit{// my own row }
                \item $r \coloneqq \AlgRow(\self)$
                \smallskip
                \item[]\textit{// get all peers in cell}
                \item $\mathbf{P} \coloneqq \AlgGetPeersInCell(r,c)$
                \smallskip
                \item[]\textit{// send them the symbol}
                \item $\pcfor \party \in \mathbf{P}\colon~\textbf{send}~(\msgStore,h,i,x)~\textbf{to}~\party$
            \end{nicodemus}
            \medskip\noindent
            \underline{\textbf{On Message} $(\msgStore,h,i,x)$}
            \begin{nicodemus}
                \item $\pcif \rdspred(h,i,x) = 0\colon~\pcreturn$
                \smallskip
                \item[]\textit{// store this symbol}
                \item $\SymbolsMap[h,i] \coloneqq x$
                \smallskip
                \item[]\textit{// forward to everyone in column}
                \item $c \coloneqq \AlgGetColForSymbol(i)$
                \item $\mathbf{P} \coloneqq \subnetprot.\iGetPeers(\subnetid_c)$
                \item $\pcfor \party \in \mathbf{P}\colon~\textbf{send}~(\msgStoreFwd,h,i,x)~\textbf{to}~\party$
            \end{nicodemus}
            \medskip\noindent
            \underline{\textbf{On Message} $(\msgStoreFwd,h,i,x)$}
            \begin{nicodemus}
                \item $\pcif \rdspred(h,i,x) = 0\colon~\pcreturn$
                \smallskip
                \item[]\textit{// only store this symbol}
                \item $\SymbolsMap[h,i] \coloneqq x$
            \end{nicodemus}
        \end{minipage}\qquad
        \begin{minipage}[t]{0.52\textwidth}%
            \underline{\textbf{Interface} $\iGet(h,i)$}
            \begin{nicodemus}
                \item \textbf{let $\tau$ be the current time slot}
                \smallskip
                \item[]\textit{// column in which symbol $i$ is stored}
                \item $c \coloneqq \AlgGetColForSymbol(i)$
                \smallskip
                \item[]\textit{// my own row }
                \item $r \coloneqq \AlgRow(\self)$
                \smallskip
                \item[]\textit{// get all peers in cell}
                \item $\mathbf{P} \coloneqq \AlgGetPeersInCell(r,c)$
                \smallskip
                \item[]\textit{// ask them for symbol}
                \item $\pcfor \party \in \mathbf{P}\colon~\textbf{send}~(\msgGet,h,i)~\textbf{to}~\party$
                \smallskip
                \item[]\textit{// return the first valid response}
                \item $\textbf{on receiving}~(\msgGetRsp,h,i,x)~\textbf{at $\tau+2$}\\
                        \pcind\pcind\pcind\textbf{s.t.}~\rdspred(h,i,x) = 1\colon$
                \item \pcind $\pcreturn x$
                \item $\textbf{if no such response within two rounds:}$
                \item \pcind $\pcreturn \bot$~\textbf{at time $\tau+3$}
            \end{nicodemus}
            \medskip\noindent
            \underline{\textbf{On Message} $(\msgGet,h,i)$ \textbf{from} $\party$}
            \begin{nicodemus}
                \smallskip
                \item[]\textit{// only respond if I store this symbol}
                \item $\pcif \SymbolsMap[h,i] = \bot\colon~\pcreturn$
                \item $x \coloneqq \SymbolsMap[h,i]$
                \item \textbf{send} $(\msgGetRsp,h,i,x)$ \textbf{to} $\party$
            \end{nicodemus}
        \end{minipage}
    }}
    \caption{Interfaces $\iStore$ and $\iGet$ of $\ourprot$ with file space $\SymbolSpace^\FileLength$ and handle space $\HandleSpace$, for a predicate $\rdspred\colon \HandleSpace \times [\FileLength] \times \SymbolSpace \rightarrow \bool$. It makes use of a subnet discovery protocol $\subnetprot$ for $\numberofsubnets = k_1 + k_2$ and a random oracle $\Hash\colon\boolstar\rightarrow[k_1]\times[k_2]$. The party that is executing the instructions is referred to as $\self$. Initialization and helper algorithms are defined in \cref{fig:ourprotocol:initialization}, and code for joining in \cref{fig:ourprotocol:joining}.}
    \label{fig:ourprotocol:storeget}
\end{figure*}

\begin{figure*}[htbp]
    \centering

    \pgfplotscreateplotcyclelist{my cycle list}{
        {black, mark=x},
        {green!50!black, mark=square*},
        {blue, mark=diamond*},
    }

    \begin{minipage}{0.48\textwidth}
        \resizebox{\textwidth}{!}{
        \begin{tikzpicture}
            \begin{axis}[
                width=\estfigWidth,
                height=\estfigHeight,
                xlabel={$k_2$},
                ylabel={$\max k_1$},
                grid=major,
                legend style={at={(0.5,1.1)}, anchor=south, legend columns=2},
                mark repeat=60,
                cycle list name=my cycle list,
                xmode=log, 
                ymode=log, 
            ]

            \addplot table [x=k2, y=k1, col sep=comma] {csvdata/estimates_data_5_100000.csv}; 
            \addplot table [x=k2, y=k1, col sep=comma] {csvdata/estimates_data_5_10000.csv}; 
            \addplot table [x=k2, y=k1, col sep=comma] {csvdata/estimates_data_5_5000.csv}; 
            \end{axis}
        \end{tikzpicture}
        }
    \end{minipage}
    \begin{minipage}{0.48\textwidth}
        \resizebox{\textwidth}{!}{
        \begin{tikzpicture}
            \begin{axis}[
                width=\estfigWidth,
                height=\estfigHeight,
                xlabel={$k_2$},
                ylabel={$\iJoin$},
                grid=major,
                cycle list name=my cycle list,
                mark repeat=60,
                xmode=log, 
                ymode=log, 
            ]
            \addplot table [x=k2, y=join_complexity, col sep=comma] {csvdata/estimates_data_5_100000.csv};
            \addplot table [x=k2, y=join_complexity, col sep=comma] {csvdata/estimates_data_5_10000.csv};
            \addplot table [x=k2, y=join_complexity, col sep=comma] {csvdata/estimates_data_5_5000.csv};
            \end{axis}
        \end{tikzpicture}
        }
    \end{minipage}

    \begin{minipage}{0.48\textwidth}
        \resizebox{\textwidth}{!}{
        \begin{tikzpicture}
            \begin{axis}[
                width=\estfigWidth,
                height=\estfigHeight,
                xlabel={$k_2$},
                ylabel={$\iGet$},
                grid=major,
                cycle list name=my cycle list,
                mark repeat=60,
                xmode=log, 
                ymode=log, 
            ]
            \addplot table [x=k2, y=get_complexity, col sep=comma] {csvdata/estimates_data_5_100000.csv};
            \addplot table [x=k2, y=get_complexity, col sep=comma] {csvdata/estimates_data_5_10000.csv};
            \addplot table [x=k2, y=get_complexity, col sep=comma] {csvdata/estimates_data_5_5000.csv};
            \end{axis}
        \end{tikzpicture}
        }
    \end{minipage}
    \begin{minipage}{0.48\textwidth}
        \resizebox{\textwidth}{!}{
        \begin{tikzpicture}
            \begin{axis}[
                width=\estfigWidth,
                height=\estfigHeight,
                xlabel={$k_2$},
                ylabel={$\iStore$},
                grid=major,
                cycle list name=my cycle list,
                mark repeat=60,
                xmode=log, 
                ymode=log, 
            ]
            \addplot table [x=k2, y=store_complexity, col sep=comma] {csvdata/estimates_data_5_100000.csv};
            \addplot table [x=k2, y=store_complexity, col sep=comma] {csvdata/estimates_data_5_10000.csv};
            \addplot table [x=k2, y=store_complexity, col sep=comma] {csvdata/estimates_data_5_5000.csv};
            \end{axis}
        \end{tikzpicture}
        }
    \end{minipage}
    \vspace{-1em}
    \begin{tikzpicture}
        \begin{axis}[
            hide axis,
            scale only axis,
            height=0pt,
            width=0pt,
            legend columns=4,
            legend style={
                at={(0.5,0.5)},
                anchor=center,
                column sep=1ex
            },
            cycle list name=my cycle list,
        ]
        \addplot coordinates {(-1e10,-1e10)}; \addlegendentry{$N = 100{,}000$, $\RelPosCorrThresh=0.05$};
        \addplot coordinates {(-1e10,-1e10)}; \addlegendentry{$N = 10{,}000$, $\RelPosCorrThresh=0.05$};
        \addplot coordinates {(-1e10,-1e10)}; \addlegendentry{$N = 5{,}000$, $\RelPosCorrThresh=0.05$};
        \end{axis}
    \end{tikzpicture}
    \vspace{0.5em}
    \caption{Trade-off between $k_1$ (number of rows) and $k_2$ (number of columns) and resulting complexities for $\RelPosCorrThresh = 0.05$. The complexities show the expected total message complexity when an honest party calls the respective interfaces. All plots assume $\RobustnessError \leq 10^{-9}$ and the given $N$. Scales are logarithmic.}
    \label{fig:estimatesB}
\end{figure*}
\clearpage
\begin{figure*}[ht]
    \centering
    \begin{tabular}{lll}
    \toprule
    \textbf{Event} & \textbf{Intuition} & \textbf{Reference} \\
    \midrule
    $\EventActive{\party}{\tau}$ & $\party$ is active at time $\tau$ & \cref{def:EventActive} \\
    $\EventActiveDuration{\party}{\tau}{\tau'}$ & $\party$ is active from time $\tau$ to time $\tau'$ & \cref{def:EventActive} \\
    $\EventFullyJoined{\party}{\tau}$ & $\party$ active at $\tau$ and $\iJoin$ already finished at $\tau$ & \cref{def:EventActive} \\
    $\EventFullyJoinedDuration{\party}{\tau}{\tau'}$ & $\party$ active from $\tau$ to $\tau'$ and $\iJoin$ already finished at $\tau$  & \cref{def:EventActive} \\
    $\EventStored{\party}{\tau}{h}{i}{x}$ & $\party$ calls $\iStore(h,i,x)$ at $\tau$ & \cref{def:rda:robustness}\\
    $\EventCalledGet{\party}{\tau}{h}{i}$ & $\party$ calls $\iGet(h,i)$ at $\tau$ & \cref{def:rda:robustness}\\
    $\EventGotResult{\party}{\tau}{h}{i}{\Delta}{x}$& With a delay of $\leq \Delta$, $\party$ got $x$ from call to $\iGet(h,i)$  & \cref{def:rda:robustness}\\
    $\EventCreatedSubnet{\subnetid}{\mathbf{P}}$ & Honest parties in $\mathbf{P}$ created subnet $\subnetid$ at $\tau = 0$ & \cref{def:subnetdiscovery:robustness} \\
    $\EventStaysInSubnetDelay{\subnetid}{\party}{\tau_0}{\tau_1}$ & $\party$ active from $\tau_0$ to $\tau_1$ and properly joined subnet $\subnetid$ & \cref{def:subnetdiscovery:robustness} \\
    $\EventIsInSubnetDelay{\subnetid}{\party}{\tau}$ & $\party$ active at $\tau$ and properly joined subnet $\subnetid$ & \cref{def:subnetdiscovery:robustness} \\
    $\EventCalledGetPeers{\subnetid}{\party}{\tau}$ & $\party$ calls $\iGetPeers(\subnetid)$ at $\tau$ & \cref{def:subnetdiscovery:robustness} \\
    $\EventGotPeer{\subnetid}{\party}{\party'}{\tau}$ & $\party$ got peer $\party'$ from call to $\iGetPeers(\subnetid)$ & \cref{def:subnetdiscovery:robustness} \\
       $\EventSubnetprotGood{T}$ & Subnet protocol works until time $T$ & \cref{def:GoodEvents} \\
       $\EventColumnGoodUntil{c}{T}{\Delta}$ & always an honest party in $c$ until $T$ (with overlap) & \cref{def:GoodEvents} \\
       $\EventGoodCell{r}{c}{\tau_1}{\tau_2}$ & $\exists$ honest party in cell $(r,c)$ from $\tau_1$ to $\tau_2$ & \cref{def:GoodEvents} \\
    %
    $\EventInSymbolStorage{\party}{\tau}{h}{i}{x}$& $\party$ stores $(h,i,x)$ at $\tau$ & \cref{def:relevant_events} \\
    $\EventStoredInColumnNew{c}{\tau_0}{\tau_1}{h}{i}{x}$& All honest parties in $c$ store $(h,i,x)$ & \cref{def:relevant_events} \\
    $\partyknows{\party}{\party'}{\tau}{\subnetid}$ & $\party$ knows $\party'$ at $\tau$ in subnet $\subnetid$ & \cref{def:simplesubnetprot:peernotation} \\
    \bottomrule
    \end{tabular}
    \caption{Table of events used throughout the paper. Events typically assume that $\party$ is honest. For the precise definitions of the events, the reader should consult the definitions.}
    \label{table:eventsoverview}
\end{figure*}

\end{document}